\documentclass[aps,
final,%
notitlepage,%
oneside,%
twocolumn,
nofootinbib,%
superscriptaddress,%
floatfix, pra]%
{revtex4-1}
\pdfoutput=1

\usepackage{color}
\newcommand{\rd}{{\mathrm d}}
\newcommand{\R}{\mathbb R}

\newcommand{\e}{{\mathrm e}}
\newcommand{\Ccal}{{\mathcal C}}

\newcommand{\vertiii}[1]{{\left\vert\kern-0.25ex\left\vert\kern-0.25ex\left\vert #1 
		\right\vert\kern-0.25ex\right\vert\kern-0.25ex\right\vert}}

\newcommand{\eq}[1]{\begin{equation} \begin{aligned}#1\end{aligned} \end{equation}}
\newcommand{\tr}{\mathrm{Tr}}

\newcommand{\eqarray}[1]{\begin{eqnarray} #1 \end{eqnarray}}
\newcommand{\ket}[1]{\vert #1 \rangle}
\newcommand{\bra} [1] {\langle #1 \vert}

\newcommand{\ketbra}[2]{| #1 \rangle \langle #2 |}

\newcommand{\dd}{\mathrm{d}}

\newcommand{\trans}{T}

\newcommand{\pur}{\mathcal{P}}
\newcommand{\entr}{\mathcal{S}}

\newcommand{\cent}{\lambda}

\usepackage{graphicx}
\usepackage{dcolumn}
\usepackage{bm}
\usepackage{color}
\usepackage{soul}
\usepackage{amsmath}
\usepackage[dvipsnames]{xcolor}
\usepackage{mathtools}
\usepackage{dsfont}
\usepackage{amssymb}
\usepackage[utf8]{inputenc}

\usepackage{amsthm}
\theoremstyle{plain}
\newtheorem{thm}{\protect\theoremname}
  \theoremstyle{plain}
  \newtheorem{lem}[thm]{\protect\lemmaname}
  \theoremstyle{remark}
  
  \theoremstyle{plain}
  \newtheorem{conjecture}[thm]{\protect\conjname}
    \theoremstyle{plain}
  \newtheorem{corollary}[thm]{\protect\corollaryname}
  
\usepackage{babel}
  \providecommand{\claimname}{Claim}
  \providecommand{\corollaryname}{Corollary}
  \providecommand{\lemmaname}{Lemma}
  \providecommand{\factname}{Fact}
    \providecommand{\conjname}{Conjecture}
\providecommand{\theoremname}{Theorem}

\usepackage{hyperref}
\hypersetup{
	colorlinks=true,
	linkcolor=ForestGreen,
	filecolor=TealBlue,      
	urlcolor=Orchid,
	citecolor=Peach
}
\usepackage{pifont}
\usepackage{amssymb}

\usepackage{titlesec}
\usepackage{enumitem}

\begin{document}
		
	\title{Equalities and inequalities from entanglement, loss, and beam splitters}
 
\author{Anaelle Hertz}
\affiliation{National Research Council of Canada, 100 Sussex Drive, Ottawa, Ontario K1N 5A2, Canada}
\author{Noah Lupu-Gladstein}
\affiliation{National Research Council of Canada, 100 Sussex Drive, Ottawa, Ontario K1N 5A2, Canada}
\affiliation{Department of Physics, University of Ottawa, 25 Templeton Street, Ottawa, Ontario, K1N 6N5 Canada}
\author{Khabat Heshami}
\affiliation{National Research Council of Canada, 100 Sussex Drive, Ottawa, Ontario K1N 5A2, Canada}
\affiliation{Department of Physics, University of Ottawa, 25 Templeton Street, Ottawa, Ontario, K1N 6N5 Canada}
\affiliation{Institute for Quantum Science and Technology, Department of Physics and Astronomy, University of Calgary, Alberta T2N 1N4, Canada}
\author{Aaron Z. Goldberg}
\affiliation{National Research Council of Canada, 100 Sussex Drive, Ottawa, Ontario K1N 5A2, Canada}

	\setcounter{tocdepth}{1}

\begin{abstract}
Quantum optics bridges esoteric notions of entanglement and superposition with practical applications like metrology and communication. Throughout, there is an interplay between information theoretic concepts such as entropy and physical considerations such as quantum system design, noise, and loss. Therefore, a fundamental result at the heart of these fields has numerous ramifications in development of applications and advancing our understanding of quantum physics. Our recent proof for the entanglement properties of states interfering with the vacuum on a beam splitter led to monotonicity and convexity properties for quantum states undergoing photon loss [Lupu-Gladstein {\it et al.}, arXiv:2411.03423 (2024)] by breathing life into a decades-old conjecture. In this work, we extend these fundamental properties to measures of similarity between states, provide inequalities for creation and annihilation operators beyond the Cauchy-Schwarz inequality, prove a conjecture [Hertz {\it et al.}, PRA {\bf 110}, 012408 (2024)] dictating that nonclassicality through the quadrature coherence scale is uncertifiable beyond a loss of 50\%, and place constraints on quasiprobability distributions of all physical states. These ideas can now circulate afresh throughout quantum optics.
\end{abstract}
 	\maketitle

Addressing problems at the intersection of quantum information theory and physics of quantum systems is critical to the development of quantum technologies. In quantum optics, and particularly continuous-variable quantum information processing, optical loss is an omnipresent impediment. It is paramount to capture its effect on notions of nonclassicality and on the effectiveness of encoding and decoding quantum information stored in physical states. 

We proved in Ref.~\cite{CompanionShortarXiv} key properties of the entanglement beam splitters generate when they interfere quantum states with the vacuum, resolving a long-standing conjecture~\cite{Asbothetal2005,Asboth2024} about the optimality of balanced (``50/50'') beam splitters. Due to these intricate connections that stem from the ubiquitousness of beam splitters for both theoretical~\cite{YuenShapiro1978,YuenShapiro1980,Yurke1985,MandelWolf1995,SalehTeich2007,AaronsonArkhipov2013} and practical~\cite{Grangier_1986,HongOuMandel1987,BornWolf1999,LIGO2011,Lucamarinietal2018,Flaminietal2019} considerations, our investigations led to many related results that we detail here.

In this work, we consider a state to be nonclassical if it cannot be described as a stochastic mixture of coherent states. The effect of loss on nonclassicality has been studied to evaluate the limits of quantum communication~\cite{GrosshansGrangier2002}, decoherence rates~\cite{Zurek1993,Dodonov_2000,Zurek2003,Hertz,Rosiek2024},
and the spoiling of quantum advantages in photonic quantum computation~\cite{MariEisert2012} and boson sampling~\cite{RahimiKesharietal2016,OszmaniecBrod2018,Qietal2020}.
In Refs.~\cite{Goldbergetal2023,HGH}, we analytically and numerically observed that nonclassicality for well-known resource states measured through the quantum coherence scale (QCS) stops being certifiable when the states lose 50\% or more of their photons, which led to a conjecture that all states lose their QCS-certified nonclassicality at the maximum loss of 50\%. In this work, we prove this conjecture by directly linking the QCS of a state undergoing loss to the evolution of purity as a function of the loss parameter.

Further, since many properties of a state, such as purity or nonclassicality, can be described with quasiprobability distributions in phase space, our work also gives robust equalities and inequalities for quasiprobability distributions. These go beyond previous relationships that are typically used for inspecting nonclassicality  \cite{SperlingVogel2009negativequasi,Bohmann,SteuernagelLee2023arxiv,Chabaudetal2024}.

These results for quasiprobability distributions are especially valuable because quasiprobability distributions are not always well-behaved functions~\cite{Sudarshan1963} and quasiprobability representations often obscure the boundary between physical and nonphysical states. Most known bounds are specifically for classical states where the distributions become more regular~\cite{Bohmann}. In contrast, we find properties of the distributions that are satisfied by all states or by all pure states. Similarly, for creation and annihilation operators, where infinite dimensional, rigged Hilbert spaces abound, our results on entanglement generation enable different results that could be valuable as uncertainty-type relations~\cite{Englert2024}.

The paper is structured as follows. In Section~\ref{sec:preliminaries}, we review the basic notions of beam splitters and entanglement, loss channels, quasiprobability distributions and the QCS as a nonclassicality measure. In Section~\ref{sec:ResultsfromPRL}, we report on the main lemmas,  obtained in Ref.~\cite{CompanionShortarXiv}, regarding the convexity of the purity and the entropy as functions of the transmission parameter of the beam splitter. Complementing those results, we prove in Section~\ref{sec:CorrolariesfromPRL} some additional properties for similarity measures under loss.
 In Section~\ref{sec:BeyondCS}, we give some inequalities for creation and annihilation operators in a loss channel that go beyond the Cauchy-Schwarz ones.
In Section~\ref{sec:QCS}, we show how the entanglement results have implications for the QCS and, in particular, we prove that nonclassicality cannot be certifiable with the QCS after 50\% loss.  Finally, in Sections~\ref{sec:InequalityIntegrals} and \ref{sec:CharactersiticFunctions}, we translate those results into terms of $s$-quasiprobability distributions and characteristic functions in  phase space and obtain inequalities that apply to all quantum states.

\section{Preliminary notions}
\label{sec:preliminaries}
\subsection{Beam splitters and loss channel}
A beam splitter, acting on a two-mode state with annihilation operators $a_1$ and $a_2$, is described by the unitary
$B(T)=\e^{\arccos\sqrt{\trans}(a_1 a_2^\dagger-a_1^\dagger a_2)}$ and enacts~\cite{weedbrook} \eq{a_1&\to B(\trans)a_1 B(\trans)^\dagger=\sqrt{\trans}a_1+\sqrt{1-\trans}a_2,\\
a_2&\to B(\trans)a_2 B(\trans)^\dagger=-\sqrt{1-\trans}a_1+\sqrt{\trans}a_2.}
The standard model for optical loss is a beam splitter where the second mode begins in the vacuum state and is ignored after the beam splitter. This is a channel $\mathcal{E}_\trans$ with transmission probability $\trans$ (loss probability $1-\trans$). A relative phase between $a_1$ and $a_2$ imparted
by the beam splitter imparts an optical phase on the transmitted state, which we set to zero without loss
of generality. Hence, a state $\rho$ undergoing loss is  written as
 \eq{\rho_\trans=\mathcal{E}_\trans[\rho_1]=\tr_2[B(T)(\rho_1\otimes\ketbra{0}{0})B^\dag(\trans)].}
No loss is represented by $\trans\!=\!1$ and thus we typically identify the initial state as $\rho_1$. Note that loss acts multiplicatively such that two loss channels with transmissions $\trans_1$ and $\trans_2$ comprise a monolithic loss with $\trans=\trans_1\trans_2$, which is helpful for practical purposes where multiple sources of loss can be considered as one.

The evolution of a lossy state can also be described continuously. It obeys the master equation for a damped harmonic oscillator, which is sometimes presented with other loss parametrizations like $\e^{-\gamma t}=\cos^2(\theta/2)=\eta=\trans$~\cite{NielsenChuang2000} and which provides the derivative
    \eq{
    \frac{\partial \rho_\trans}{\partial \trans}=-\frac{1}{2\trans}(2a\rho_\trans a^\dagger-a^\dagger a\rho_\trans-\rho_\trans a^\dagger a).\label{eq:derivFromMasterEq} }
    One can also describe a lossy state as $\rho_T=~\!\sum_n K_n(T)\rho_1 K_n(T)^\dag$ where \eq{K_n(\trans)=\langle n|_2 B(\trans)|0\rangle_2&=\frac{\left(\sqrt{\frac{1-\trans}{\trans}}a_1\right)^n\sqrt{\trans}^{a_1^\dagger a_1}}{\sqrt{n!}}\\&=\frac{\sqrt{\trans}^{a_1^\dagger a_1}(\sqrt{1-\trans}a_1)^n}{\sqrt{n!}} \label{kraus}}  are the Kraus operators comprising the completely positive, trace-preserving quantum channel $\mathcal{E}_\trans$ \cite{Goldberg2024}.

For example, a Fock state $|n\rangle=a^{\dagger n}|0\rangle/\sqrt
{n!}$ evolves under loss to a binomial distribution of states with definite photon number:
\begin{equation}
    \mathcal{E}_T[|n\rangle\langle n|]=\sum_{k=0}^n \binom{n}{k}T^k(1-T)^{n-k}|k\rangle\langle k|.
    \label{eq:Fock loss binom}
\end{equation}

\subsection{Quasiprobability distributions in phase space}
When using the phase-space approach to quantum mechanics, one cannot define true probability distributions for the quantum state in phase space but, rather, quasiprobability distributions that relax some of Kolmogorov's axioms yet fully represent all of the information present in a quantum state $\rho$. The renowned Glauber-Sudarshan, Wigner, and Husimi distributions for a continuous-variable quantum system are, respectively, the $s=1$, $s=0$, and $s=-1$ cases of the $s$-ordered quasiprobability distributions 
\begin{equation}\label{eq:s-quasi}
    P_\rho(\alpha ,s)=\frac{1}{\pi^2}\int \dd^2\beta \ \e^{s\frac{|\beta |^2}{2}+\beta ^*\alpha -\beta \alpha ^*}\tr[\rho D(\beta )],
\end{equation} 
where $\dd^2\alpha =\dd\Re(\alpha )\dd\Im(\alpha )$, the quadratures are defined as $x=\sqrt{2}\Re(\alpha )$, $p=\sqrt{2}\Im(\alpha )$,  and the displacement operator $D(\beta)=\exp(\beta a^\dagger-\beta^* a)$ acts on coherent states as $D(\beta )|\alpha \rangle=|\alpha +\beta \rangle$. Coherent states $|\alpha\rangle\propto\sum_n (\alpha^n/\sqrt{n!})|n\rangle$ are said to be classical and evolve as $\mathcal{E}_\trans(|\alpha\rangle)=|\trans\alpha\rangle$ under loss.

The $s$-ordered quasiprobability distributions are normalized $\int P_\rho(\alpha ,s)\dd^2\alpha =1$ but not necessarily positive or regular everywhere and do not represent mutually exclusive events for different values of $\alpha$. The value of $s$ can be shifted to any $s + \Delta s \in \mathbb{R}$ via the convolution 
\begin{equation}\label{eq:convolution}
    P_\rho(\alpha ,s + \Delta s)=\frac{2}{\pi |\Delta s|}\int \dd^2\beta \ P_\rho(\beta ,s)\exp\left(\frac{2|\alpha -\beta |^2}{\Delta s}\right) .
\end{equation}

Using the Fourier transform of the displacement operator \cite{CahillGlauber}
\begin{equation}
    \Delta(\alpha ,s)=\frac1{\pi^2}\int\dd^2\beta \  D(\beta  )\e^{s|\beta |^2/2}\e^{\alpha \beta ^*-\alpha ^*\beta },
\end{equation} known as the $s$-ordered operator kernel,
    Eq.~\eqref{eq:s-quasi} can be written as     $P_\rho(\alpha ,s)=\tr[\rho \Delta(\alpha ,s)]$. In addition, the density matrix of the state is also defined as \cite{CahillGlauber}
\begin{equation}
    \rho=\pi\int\dd^2\alpha \  P_\rho(\alpha ,s)\Delta(\alpha ,-s).
\end{equation}
This implies that the
overlap between two states $\tr[\rho\sigma]$ can be written as 
\begin{equation}\label{overlap}
    \tr[\rho\sigma]=\pi\int \dd^2\alpha \ P_\rho(\alpha ,s)P_\sigma(\alpha ,-s).
\end{equation}

The $s$-ordered quasiprobability function of the output of a loss channel, $P_{\rho_\trans}(\alpha ,s)$, can be written as \cite{RadimFilip}
\eq{P_{\rho_\trans}(\alpha ,s)=\frac1{\trans}P_{\rho_1}\left(\frac{\alpha }{\sqrt{\trans}},\frac{s+\trans-1}{\trans}\right),\label{s_quasi_distribution_evolution}}
where $P_{\rho_1}(\alpha ,s)$ is the s-ordered quasiprobability function of the input and is replaced by the input's quasiprobability function of order $s'=(s+\trans-1)/\trans$. The Wigner function corresponds to $s=0$, which implies that the Wigner function evolves to a distribution with $s'\leq-1$ as soon as $\trans\leq\frac12$. All s-ordered quasiprobability distributions with $s\leq-1$ are positive, so the Wigner function  becomes positive after 50\% loss~\cite{RadimFilip}.
This is demonstrated in the upper portion of Fig.~\ref{fig:Fock1}, where a single-photon Fock state is seen to have its Wigner negativity vanish with increasing loss.

\begin{figure}
    \centering
    \includegraphics[width=0.85\linewidth]{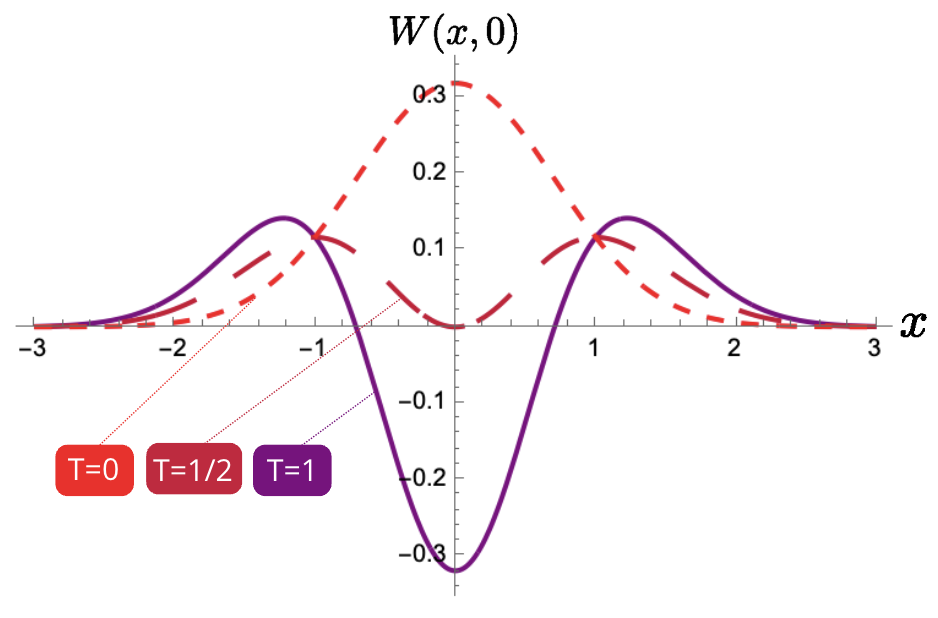}
        \includegraphics[width=0.85\linewidth]{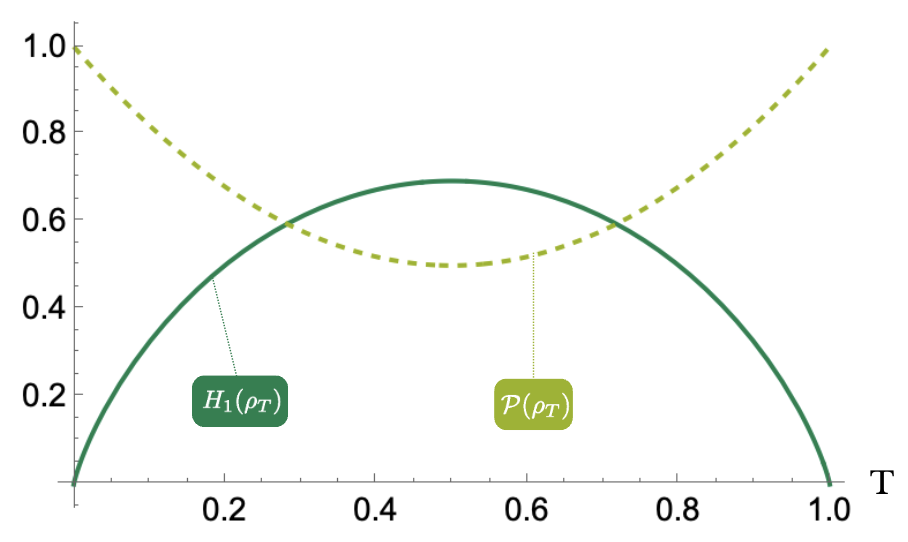}
    \caption{Top: zero-momentum slice of the Wigner function $W(x,0)$ of a Fock state $\ket{1}$ undergoing some loss $\trans=1$ (solid), $1/2$ (long dashed), $0$ (dashed). Bottom: its von Neumann entropy $H_1(\rho_\trans)$ (solid) and purity $\pur(\rho_T)$ (dashed).}
    \label{fig:Fock1}
\end{figure}

Defining  the $s$-ordered displacement operator as
\begin{equation}
    D(\alpha, s) = D(\alpha) e^{s|\alpha|^2/2},
\end{equation}
Eq.~\eqref{eq:s-quasi} can be written as
\begin{equation}
P_\rho(\alpha,s)=\frac1{\pi^2}\int\dd^2\beta\,\e^{\beta^*\alpha-\beta\alpha^*}\chi_\rho(\beta,s),
\end{equation}
where $\chi_\rho(\beta,s)= \tr[\rho D(\alpha, s)]$ is the characteristic function of the state $\rho$ \cite{CahillGlauber} and 
\begin{equation}
\rho = \frac{1}{\pi} \int \dd^2 \alpha\ D(\alpha, s)^{-1} \chi_\rho(\alpha, s) \ .
\end{equation}
The $s$-ordered characteristic function is normalized $\chi_\rho(0, s) = 1$ and its growth is bounded by $|\chi_\rho(\alpha, s)|^2 \leq e^{-s|\alpha|^2}$. 

Similarly to Eq.~\eqref{s_quasi_distribution_evolution}, the $s$-ordered characteristic function of an initial state $\rho_1$ affected by loss in a channel of transmittance $\trans$ is \cite{se06}
\begin{equation}
    \chi_{\rho_T}(\alpha, s) = \chi_{\rho_1} \left (\sqrt{T} \alpha, \frac{s + T - 1}{T} \right ).\label{eq:evolutionofchi}
\end{equation}
Finally, one can also compute the purity of a state $\rho$ with its $s$-ordered characteristic function:
\begin{equation}
    \tr[\rho^2] = \int \frac{\dd^2 \alpha}{\pi}\ e^{-s |\alpha|^2} | \chi_\rho(\alpha, s) |^2 \ . \label{eq:puritywithChi}
\end{equation}

\subsection{Entanglement and loss}
Beamsplitters convert optical nonclassicality into entanglement. After interfering with the vacuum at a beam splitter, all pure or Gaussian states whose Glauber-Sudarshan distributions are negative somewhere will generate two-mode entangled states~\cite{Kimetal2002} and the same is a necessary condition for all mixed input states to generate entanglement~\cite{Xiangbin2002}. These entanglement generation properties have been studied from numerous perspectives over the years~\cite{Tanetal1991,Sanders1992,HuangAgarwal1994,ArvindMukunda1999,Paris1999,Wolfetal2003,Asbothetal2005,Ivanetal2006arxiv,UshaDevietal2006,Tahiraetal2009,Springeretal2009,Lietal2010,Ivanetal2011,Pianietal2011,Ivanetal2012generation,AbdelKhaleketal2012,Jiangetal2013,Berradaetal2013,Killoranetal2014,VogelSperling2014,Monteiroetal2015,Geetal2015,Miranowiczetal2015,Streltsovetal2015,Brunellietal2015,Arkhipovetal2016,RohithSudheesh2016,GholipourShahandeh2016,Maetal2016,Killoranetal2016,BoseKumar2017,Parketal2017,GoldbergJames2018,Kwonetal2019,Tasginetal2020,Fuetal2020,Tserkisetal2020,GoldbergHeshami2021,Ataman2022,Akellaetal2022,Lietal2023,Simonettietal2023,DeepakChatterjee2023,Steinhoff2024,SerranoEnsastigaetal2024arxiv,Lietal2024arxiv}.

Our companion paper~\cite{CompanionShortarXiv} proved the conjecture~\cite{Asbothetal2005} that, among all possible linear optical arrangements, the ones that convert nonclassicality into the most entanglement are all equivalent to a balanced beam splitter, up to optical phases before and after.
For pure states, two-mode entanglement is encoded in the lack of purity of a reduced system. For a pure state interfering with vacuum on a beam splitter, the reduced system of the input mode is equivalent to the input state subject to loss, $\rho_T$. Thus studies of entanglement and loss are intimately connected, as was recognized by the definition of entanglement potential in Ref.~\cite{Asbothetal2005}.

One entanglement measure for two-mode pure states is the mixedness or linear entropy \cite{Manfredi2000} defined as $1-\pur_\psi(\trans)$, where 
\eq{\pur_\psi(\trans)=\pur(\rho_\trans)=\tr[\rho_\trans^2]} 
is the purity of the state undergoing loss. 
Another common entanglement monotone is the entanglement entropy that corresponds to the von Neumann entropy of $\rho_\trans$ \cite{VonNeumann1932}: \eq{\entr_\psi(\trans)=H_{1}(\rho_\trans)=-\tr[\rho_\trans\log \rho_\trans],}
which when linearized in a Mercator series expansion around $1-\rho$ yields the linear entropy. We write this with the notation $H_1$ because it is among the family of R\'enyi entropies \eq{H_\alpha(\rho)=\frac{1}{1-\alpha}\log\tr[\rho^\alpha]} 
that completely determines the eigenvalues and thus the stochastic uncertainty of a state $\rho$. One can see that the mixedness is related to the 2-entropy by $H_2(\rho_T)=-\log \pur(\rho_T)$ and that the entanglement entropy requires a limiting procedure for $\alpha\to 1$. More importantly, all of these entropies $H_\alpha(\rho_T)$ vanish when no entanglement is generated at the beam splitter and are bona fide entanglement measures. Due to their convexity properties, their convex roof extensions are valid entanglement measures for mixed states as well~\cite{Vidal2000,BengtssonZyczkowski2006,Horodeckietal2009}.

The lower portion of Fig.~\ref{fig:Fock1} shows an example of how these quantities evolve for a single-photon Fock state subject to loss, computed by setting $n=1$ in Eq.~\eqref{eq:Fock loss binom}. Then, using this to quantify the two-mode entanglement generated by a single photon input to a beam splitter, we see that the purity decreases monotonically and the von Neumann entropy increases monotonically toward $\trans=1/2$ from both sides, demonstrating that the entanglement between the output ports increases as the beam splitter becomes more balanced.

\subsection{Quadrature coherence scale}
The QCS  of a one-mode state $\rho$ is defined as \cite{Debievre,Hertz}
\begin{equation}\label{eq:QCS}
\Ccal^2(\rho)=\frac1{2\pur(\rho_\trans)}\left(\tr[\rho, X][X,\rho]+\tr[\rho, P][P,\rho]\right), 
\end{equation}
where $X=\frac{a^\dagger+a}{\sqrt2},\, P=\frac{i (a^\dagger-a)}{\sqrt2}$. 
 It quantifies how much coherence a state has in any quadrature basis together with the strength of the coherence, such that superpositions of farther-apart positions are given more weight. This can be better seen by 
 rewriting Eq.~\eqref{eq:QCS} as
\begin{multline}\label{eq:QCS2}
\Ccal^2(\rho)=\frac1{2\pur(\rho_\trans)}\left(\int\rd x\rd x'\ (x-x')^2 |\rho(x,x')|^2\right. +\\ \left.\int \rd p\rd p'\ (p-p')^2|\rho(p,p')|^2 \right).
\end{multline}
Here $\rho(x,x')$ and $\rho(p,p')$ are the operator kernels of $\rho$ in the $X$ and $P$-representations. The QCS is unchanged by rotating the quadratures: an eigenstate of $X$ would have $[X,\rho]$ vanish but not $[P,\rho]$, so the QCS equally weights all quadratures to provide a unified scale for their overall coherence.

An essential feature of the QCS is that it acts as a certifier of nonclassicality. 
A state $\rho$ is said to be classical~\cite{Titulaer}  if it can be written as a mixture of coherent states. In other words, a state is classical if and only if  there exists an everywhere positive Sudarshan-Glauber function $P(\alpha )\equiv P(\alpha ,1)$ defined as\footnote{Note that this decomposition is not unique and a classical state may be described by several $P$-functions with not all of them being positive.}
 	\begin{equation}\label{eq:class-state}
 	\rho=\int \rd \alpha \ P(\alpha ) | \alpha \rangle\langle \alpha | ,
 	\end{equation}
where $\ket{\alpha }$ is a coherent state. 
 	
 Determining if such a $P$-function exists  is known  to be a challenging problem, resulting in the development of many nonclassicality measures, such as the distance to the set of classical states \cite{Hillery1} or witnesses such as the negativity of the Wigner functions~\cite{Kenfack}, among many others~\cite{Bach,Hillery3,Lee, Agarwal2,Lutkenhaus, Dodonov, Marian,Richter, Ryl,Sperling,Killoran,Alexanian,Nair,Ryl2, Yadin, Kwon2,  Takagi18, Horoshko, Luo, Bohmann,Tan2020}. Recently, it was shown that the QCS is a certifier of nonclassicality~\cite{Hertz}: $\Ccal>1$ certifies that a state is nonclassical. The QCS also has the advantage that it witnesses the nonclassicality of squeezed states while Wigner negativity does not. In addition, while the QCS is not itself a distance measure, it gives a good estimate for distances via $\Ccal-1\leq d(\rho, C_{cl})\leq \Ccal$ where $d(\rho, C_{cl})$, is the distance to the set of classical states $C_{cl}$ \cite{Hillery1}.

\section{Theorems from companion paper}
\label{sec:ResultsfromPRL}

To prove maximum entanglement at balanced beam splitter in Ref.~\cite{CompanionShortarXiv}, we considered various entanglement monotones. We thus proved several theorems for the purity and its derivatives that have more corollaries beyond maximizing entanglement. We state here the relevant ones and refer the interested reader to our companion paper for more details~\cite{CompanionShortarXiv}.

The following theorems and lemmas show that the purity of a pure state $\ket{\psi}$ undergoing loss is convex in the beam splitter transmission probability $\trans$ and symmetric about $\trans\!=\!1/2$, which implies that $\frac{\partial \mathcal{P}_\psi(\trans)}{\partial\trans}$ is negative (positive) when $\trans<1/2$ ($\trans>1/2$) and equal to 0 at $\trans=1/2$. 
:
\begin{lem}[Symmetry of purity (Lemma 1 in Ref.~\cite{CompanionShortarXiv})]\label{lem:sym}
A pure state subject to loss has the same purity for $\trans \leftrightarrow 1\!-\!\trans$; i.e., $\pur_\psi(T)=\pur_\psi(1\!-\!T)$. Therefore, there is a local extremum at $\trans\!=\!1/2$ and, thus, if the derivative is always nonnegative on one side of $\trans\!=\!1/2$ it is always nonpositive on the other side.
\end{lem}
\begin{thm}[Convexity of purity (Theorem 7 in Ref.~\cite{CompanionShortarXiv})]\label{thm:pur}
    The purity of a pure state subject to loss $\pur_\psi(\trans)$ is convex in $\trans \in \mathbb{R}$; $\partial ^2 \pur_\psi(\trans)/\partial \trans^2\geq~0$.
\end{thm}

The purity can be seen as a special case of the Hilbert-Schmidt inner product $O_\trans=\tr[\sigma_\trans^\dagger \rho_\trans]$ that measures the overlap between any two operators $\sigma_\trans$ and $\rho_\trans$. For this more general quantity, we proved the following Lemma. In the next section, this will allow us to analyze the derivative of the purity for a mixed state.

\begin{lem}
    [Positive polynomial expansion of Hilbert-Schmidt norm (Lemma 6 in Ref.~\cite{CompanionShortarXiv})]\label{lem:HS positive} The overlap between any two positive operators $\rho_1$ and $\sigma_1$ subject to loss $\trans$
    is a polynomial in $\cent\equiv (1-2\trans)$ with nonnegative coefficients.
\end{lem}

It is well known that purity can be experimentally determined by preparing two identical copies of a state and then measuring the photon number statistic $p_{n_-}$ on the  ``difference'' mode---or the dark port---that is the mode associated to $a_-=\frac1{\sqrt{2}}(a_1-a_2)$. 
 The purity is then given by the mean photon-number parity: $\sum_{n_-} (-1)^{n_-}p_{n_-}$\cite{Daleyetal2012,Islametal2015}.
 With this picture in mind, Lemma~\ref{lem:HS positive} can be understood as follows: if a state undergoes some loss $\trans$, the loss can be moved on the other side of the beam splitter (as for all linear optical networks~\cite{Goldberg2024}). The purity is obtained with the same scheme, only we now measure $(1-2\trans)^{N_-}$, with $N_-=~a_-^\dag a_-$ being the number operator associated to the dark port. This is summarized in Figure~\ref{fig:proofPurity}.

 The actual expression for the purity is then
\begin{equation}\label{eq:polynomPurity}
        \pur(\rho_\trans)\!=\!\sum_{m\geq 0}\!\lambda^m \tr[({\rho_1}\otimes{\rho_1}) B(\tfrac{1}{2})(\openone \!\otimes\!|m\rangle\langle m|)B(\tfrac{1}{2})^\dagger],
    \end{equation}
or more simply   \begin{equation}\label{purityDarkPort}
\pur(\rho_\trans)\!=\!\tr[{\rho}\otimes{\rho}\,\lambda^{N_-} ].
    \end{equation}
This expression is useful as it allows us to compute the exact value of the minimum purity. Indeed, since the minimium is obtained at $\lambda=0$ (per Lemma~\ref{lem:sym} and Theorem~\ref{thm:pur}), we only need to keep the 0 power of the polynomial, which is equivalent to sending two copies of the state to a balanced beam splitter and projecting the second mode onto the vacuum. When the initial state is a pure state $\ket{\psi}=\sum_j \psi_j\ket{j}$, minimum purity is equivalent to maximal entanglement $1\!-\! \pur(\rho_\trans)_{\textrm{min}}$ and Eq.~\eqref{eq:polynomPurity} implies that the optimal value is (see Appendix~\ref{Appendix:purity Fock} for the expression for all $\trans$)
\eq{ \pur(\rho_\trans)_{\textrm{min}}=\!\sum_{n,k,k'}\frac{\psi_k\psi_{k'}^*\psi_{n-k}\psi_{n-k'}^*}{2^n}\sqrt{\binom{n}{k}\binom{n}{k'}}.}

\begin{figure*}
    \centering
    \includegraphics[width=0.99\textwidth]{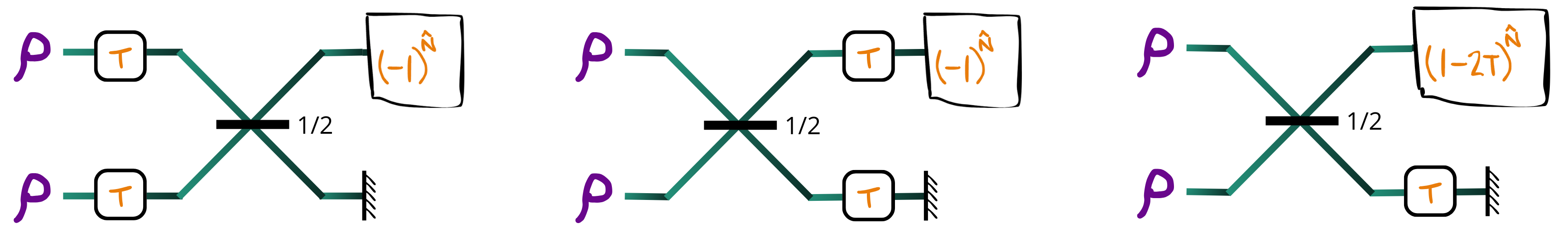}
    \caption{Two-copy circuit used to measure the purity of a state $\rho$. The left scheme is the usual setup used to measure purity using two copies of a state~\cite{Daleyetal2012, Islametal2015}, the other two are equivalent setups when the same loss $\trans$ is applied on both copies of the state and are essential to proving the convexity properties of lossy states and the entanglement-generation properties of beam splitters.}
    \label{fig:proofPurity}
\end{figure*}

Moving from linear entropies to entanglement entropies, we proved that $H_{1}(\rho_\trans)=H_{1}(\rho_{1-\trans})$ and that the von Neumann entropy is concave:
\begin{thm}[Concavity of entropy (Theorem 4 in Ref.~\cite{CompanionShortarXiv})]\label{thm:entr}
    The von Neumann entropy of a state subject to loss is concave in $\trans$ for $\trans \in [0, 1]$; $\partial ^2 H_1(\rho_\trans)/\partial \trans^2\leq~0$.
\end{thm}

\section{Corollaries for similarity measures}
\label{sec:CorrolariesfromPRL}

In the previous section, we revisited the principal results established in Ref.~\cite{CompanionShortarXiv}. In the process of developing those proofs, we also arrived at several noteworthy supplementary findings for similarity measures. For example, the next corollary shows that not only the von Neumann entropy of $\rho_\trans$ is concave in $\trans$, but also the mutual information between the two outputs of the beam splitter ($\rho_\trans$ corresponds to one output, while instead tracing out the first mode gives the second output).
\begin{corollary}
    [Mutual information] The quantum mutual information between the two outputs of a beam splitter with a vacuum input is concave in $\trans$.
\end{corollary}
\begin{proof}
    The quantum mutual information after the beam splitter is given by $I=H_1(\tr_a[\omega_\trans])+H_1(\tr_b[\omega_\trans])-H_1(\omega_\trans)$ for $\omega_\trans=B(\trans)\rho_1\otimes|0\rangle\langle 0|B(\trans)^\dagger$. Entropies are invariant under unitaries, implying $H_1(\omega_\trans)=H_1(\omega_1)$ is independent of $\trans$. Beam splitters have the symmetry $\tr_b[\omega_\trans]=\tr_a[\omega_{1-\trans}]$. Both $H_1(\tr_a[\omega_\trans])$ and $H_1(\tr_a[\omega_{1-\trans}])$ are concave in $\trans$ by Theorem~\ref{thm:entr}.
\end{proof}

When an initial pure state is subject to loss, Lemma~\ref{lem:sym} shows that the derivative of the purity will always be 0 at exactly $50\%$ loss. We show in the next corollary that if the initial state is mixed, the first derivative will always be nonpositive and the second derivative always nonnegative as long as $\trans\leq1/2$. 

\begin{corollary}
    [Hilbert-Schmidt inner-product and loss]\label{thm:overlap}
    When loss is greater than 50\%, the similarity between two states as measured by the Hilbert-Schmidt inner product increases with increasing loss $1-\trans$ and is convex. 
\end{corollary}
\begin{proof}
 By Lemma~\ref{lem:HS positive}, $O_\trans=\sum_{m\geq 0}\lambda^m p_m$ for coefficients $p_m\geq 0$  and $\lambda=1-2\trans$. Thus $\partial O_\trans/\partial \trans=-2\sum_{{m \geq 1}}m\lambda^{m-1} p_m$ and $\partial^2 O_\trans/\partial \trans^2=4\sum_{{m \geq 2}}m(m-1)\lambda^{m-2} p_m$. Since $\lambda\geq 0$ when $\trans\leq 1/2$, $\partial O_\trans/\partial \trans\leq 0$ and $\partial^2 O_\trans/\partial \trans^2\geq 0$ for $\trans\leq 1/2$.
\end{proof}

The previous result together with Theorem~\ref{thm:pur} proves the convexity of the purity of a state under loss with respect to the transmission probability $\trans$, but a bounty of evidence suggests that a much stronger property holds: the logarithm of the purity appears also to be convex in $\trans \in [0, 1]$. Equivalently, the order-$2$ Renyi entropy appears to be concave, which leads us to the following conjecture:
\begin{conjecture}\label{log-conjecture}
    The purity of a  state $\rho_1$ subject to loss is log-convex in $T \in [0, 1]$; $\partial^2 \log \tr[\rho_\trans^2] / \partial \trans^2 \geq 0$.
\end{conjecture}

Conjecture~\ref{log-conjecture} can equivalently be written as
\begin{equation}
    \left (\frac{\partial}{\partial \trans}\tr[\rho_\trans^2] \right )^2 \leq \tr[\rho_\trans^2] \frac{\partial^2}{\partial\trans^2} \tr[\rho_\trans^2].
\end{equation}
This form of the inequality is reminiscent of the Cauchy-Schwarz inequality. In fact, we can use the Cauchy-Schwarz inequality to prove a weaker result: log-convexity  of the purity with respect to the modified parameter $\ell = \log (1 - 2 \trans)$ for $\trans \in (-\infty, 1/2]$. Using Eq.~\eqref{purityDarkPort},
log-convexity in $\ell$ amounts to the inequality
\begin{equation}
    \tr[\rho_1 \otimes \rho_1\, \e^{\ell N_-} N_-]^2 \leq \tr[\rho_1 \otimes \rho_1\, \e^{\ell N_-}] \tr[\rho_1 \otimes \rho_1\, \e^{\ell N_-} N_-^2].
\end{equation}
For all $\trans \leq 1/2$, $\ell$ is real and the inequality above holds for all mixed states $\rho$ via the Cauchy-Schwarz inequality
\begin{equation}
    |\langle A, B \rangle|^2 \leq \langle A, A \rangle \langle B, B \rangle
\end{equation}
in which the inner-product is chosen to be \eq{\langle A, B \rangle = \tr[\rho_1 \otimes \rho_1 \, \e^{\ell N_-} A^\dagger B]}
and we have set $A = 1$ and $B = N_-$.

The problem of log-convexity with respect to the original loss parameter $\trans$ is harder and despite a multitude of evidence, we do not yet have a proof. In Appendix~\ref{Appendix:logconvexity}, we include more details of possible paths toward a conclusive proof.

\section{Corollaries for ladder operators beyond the Cauchy-Schwarz inequalities}
\label{sec:BeyondCS}

It is known that applying the Cauchy-Schwarz inequality to creation and annihilation operators gives the following relation
\eq{|\tr(\rho \, \hat{a})|^2 \leq \tr(\rho \, \hat{a}^\dagger \hat{a}). \label{eq:CS-a}
} This can also been derived from properties of lossy states, as will be seen below.
However, with the results of the previous section, we can also derive some new inequalities. From Eq.~\eqref{eq:derivFromMasterEq}, we can 
write the derivative of the purity as the expectation value of the Lindbladian for a loss channel:
\begin{equation}
    \begin{aligned}
        \frac{\partial \mathcal{P}(\rho_\trans)}{\partial \trans}&=   \frac{\partial \tr[\rho_\trans^2]}{\partial \trans}=\frac{2}{\trans}\tr[a^\dagger a\rho_\trans\rho_\trans-a\rho_\trans a^\dagger \rho_\trans].\\
    \end{aligned}\label{eq:derPwithaadag}
\end{equation}
Since  Corollary~\ref{thm:overlap} shows that $\frac{\partial \mathcal{P}(\rho_\trans)}{\partial \trans}\leq0$ for all states with $\trans\leq1/2$, we have a first inequality:
\begin{corollary}
For all states $\rho_1$ undergoing loss $\trans\leq1/2$ 
   $$ \tr[a^\dagger a\rho_\trans\rho_\trans]\leq\tr[a\rho_\trans a^\dagger \rho_\trans].$$
\end{corollary}
If Conjecture~\ref{log-conjecture} is true then this Corollary can be extended to all $\trans$ and the inequality can be made tighter in a state-dependent way.

When $\trans=1$ and the initial state is a pure state $\ket{\psi}$, Eq.~\eqref{eq:derPwithaadag} combined with Lemma~\ref{lem:sym} and Theorem~\ref{thm:pur} reduces to Eq.~\eqref{eq:CS-a}. That pure states become more mixed after an infinitesimal amount of loss is thus equivalent to a Cauchy-Schwarz inequality.

Let us now focus on initial pure states $\ket\psi$ such that $\rho_T=X(\trans)X(\trans)^\dag$ with
 $X(\trans)=\sum_n K_n(\trans)\ket{\psi}\bra{n}$ and where $K_n$ is the Kraus operator defined in Eq.~\eqref{kraus}; these are used routinly in the proofs of Ref.~\cite{CompanionShortarXiv}. Then
\begin{equation}
    \begin{aligned}
        a X(\trans)&=\!\sum_n\frac{a^{n+1}}{\sqrt{(n+1)!}}\sqrt{1-\trans}^n\sqrt{\trans}^{a^\dag a-n}\ketbra{\psi}{n}\sqrt{n\!+\!1}\nonumber\\
        &=\!\sum_n\!\frac{a^{n+1}}{\sqrt{(n+1)!}}\frac{\sqrt{1\!-\!\trans}^{n+1}\sqrt{\trans}^{a^\dag a-n-1}}{\sqrt{1\!-\!\trans}/\sqrt{\trans}}\ketbra{\psi}{n\!+\!1}a^\dag\nonumber\\
        &=\!X(T)a^\dag\frac{\sqrt{\trans}}{\sqrt{1-\trans}}.
    \end{aligned}
\end{equation}
Next, remark that
\begin{equation}
    \begin{aligned}
    \tr[a\rho_\trans a^\dag\rho_\trans]&=\tr[a X(\trans)X(\trans)^\dag a^\dag X(\trans)X(\trans)^\dag]\\
    &=\tr[X(\trans)a^\dag a X(\trans)^\dag X(\trans)X(\trans)^\dag]\frac{\trans}{1-\trans}\\
    &=\tr[a^\dag a\rho_{1-\trans}^2]\frac{\trans}{1-\trans},
    \label{eq:transpose trick}
  \end{aligned}
\end{equation} 
where in the last line we used $X(\trans)=X(1-\trans)^T$ proven in Ref.~\cite{CompanionShortarXiv}. This ``transpose trick'' allows us to convert expressions with $\rho_\trans$ to ones with $\rho_{1-\trans}$. From Eq. \eqref{eq:derPwithaadag}, we can now write
\begin{equation}
    \begin{aligned}
        \frac{\partial \mathcal{P}(\rho_\trans)}{\partial \trans}
        &=2\left(\frac{\tr[a^\dagger a \rho_\trans^2]}{T}-\frac{\tr[a^\dagger a \rho_{1-\trans}^2]}{1-T}\right).
    \end{aligned}
\end{equation}
This equation combined with Lemma~\ref{lem:sym} leads to the following corollary:
\begin{corollary}
    For initial pure states and $\trans\leq1/2$,
    $$\frac{\tr[a^\dagger a \rho_\trans^2]}{T}\leq\frac{\tr[a^\dagger a \rho_{1-\trans}^2]}{1-T}.$$
\end{corollary}
Again, if Conjecture~\ref{log-conjecture} is true then this Corollary can be extended to all $\trans$ and the inequality can be made tighter in a state-dependent way.

In Appendix~\ref{app: second deriv}, we also compute the second derivative of a state subject to loss. All of the expressions there will thus be positive for all pure states and for all states with $\trans\leq 1/2$ (as usual, if Conjecture~\ref{log-conjecture} is true then this can be extended to all $\trans$). For example, considering the first expression and any pure state, we find
\begin{equation}
    \begin{aligned}
        4\Re\langle a^\dagger a^2\rangle\langle a^\dagger \rangle-|\langle a^2\rangle|^2\leq  2\langle a^\dagger a\rangle^2-\langle a^\dagger a\rangle+\langle (a^\dagger a)^2\rangle,
    \end{aligned}
\end{equation} which cannot be found from any Cauchy-Schwarz inequality of which we are aware.

Similar inequalities follow from all higher-order derivatives. Because $(-1)^m\partial^m \pur/\partial T^m\geq 0$ for any state with $T\leq 1/2$ and for all pure states with $T \in (1/2,\infty)$ the expression is positive (negative) for even (odd) $m$, and because each derivative introduces one more creation and annihilation operator through the master equation Eq.~\eqref{eq:derivFromMasterEq}, the $m$th derivative will yield an inequality for up to $2m$ creation and annihilation operators that is quadratic in the quantum state.

\section{Corollaries for the quadrature coherence scale}
\label{sec:QCS}

In Refs.~\cite{Goldbergetal2023,HGH} we studied the evolution of the QCS for different pure states under loss and remarked that all QCS-nonclassicality was lost at exactly 50\% while mixed states undergoing at least 50\% loss always had a QCS below one. This led us to conjecture in Ref.~\cite{HGH} that all states lose their nonclassicality, as measured by the QCS, at the latest after 50\% loss. This is now proven in Theorem \ref{thm:QCSunderlossPure} and \ref{thm:QCSunderlossMixed} using the result in our companion paper~\cite{CompanionShortarXiv} showing that maximal entanglement is generated by
beam splitters with equal reflection and transmission probabilities.

We start by introducing the following Lemma expressing the QCS in terms of the rate of purity change. 

\begin{lem}[QCS as rate of purity change]\label{thm:QCSintermsofpurity}
    The QCS is equal to 
    \eq{\Ccal^2(\rho_\trans)=\frac{\trans}{\mathcal{P}(\rho_\trans)}\frac{\partial \mathcal{P}(\rho_\trans)}{\partial\trans}+1.\label{eq:QCSlosschannel}}
\end{lem}
\begin{proof}
    The initial state $\rho_1$ can be written in terms of the $P$ function as in Eq.~\eqref{eq:class-state}.
In a loss channel of transmittance $\trans$, it is known that a coherent state evolves as $\ket{\alpha }\rightarrow \ket{\sqrt{\trans}\alpha }$.
Therefore, the state $\rho_1$ becomes
\eq{\rho_\trans=\int \dd^2\alpha  \ P(\alpha )\ketbra{\sqrt{\trans}\alpha  }{\sqrt{\trans}\alpha  }.\label{eq:rhoT}}
It was shown in Ref.~\cite{Griffet} that the QCS can be computed as 
\eq{\Ccal^2(\rho)=\frac{\tr[\rho\otimes\rho \hat{N}]}{\tr [\rho\otimes\rho \hat{S}]}\label{QCS_twocopies}.}
$\hat{S}$ is the SWAP operator, acting as $\hat S|\psi\rangle\otimes|\phi\rangle=|\phi\rangle\otimes|\psi\rangle$ and is known to allow a measure of purity using two copies of a state~\cite{Bovinoetal2005,Islametal2015} because it is equivalent to a balanced beam splitter acting on the parity operator.
In the numerator,
$\hat{N}$ is the two-copy observable \eq{\hat{N}=\frac12 \left((X_1-X_2)^2+(P_1-P_2)^2\right)\hat S\label{Nobservable}.}

Let $a_j=(x_j+ip_j)/\sqrt{2}$ for $j+1,2$ be the mode operators with $[a_j,a_j^\dag]=1$. Then, using \eqref{Nobservable}, we can compute
\eqarray{(a_1^\dag\!-\!a_2^\dag)(a_1\!-\!a_2)\hat S
&=&\frac12 \big((X_1\!-\!X_2)^2+(P_1\!-\!P_2)^2-2\big)\hat S\nonumber\\
&=&\hat N-\hat S.\label{eq:NmoinsS}}
Using Eqs.~\eqref{eq:rhoT} and \eqref{eq:NmoinsS}, we see that 
\eq{\tr\left[\rho_\trans\otimes\rho_\trans(\hat N-\hat S)\right]
=\trans\frac{\partial\mathcal{P}(\rho_\trans)}{\partial\trans}}
and also, from Eqs.~\eqref{QCS_twocopies} and $\tr [\rho\otimes\rho \hat{S}]=\tr [\rho^2]$, that
\eqarray{\tr\left[\rho_\trans\otimes\rho_\trans(\hat N-\hat S)\right]&=&
\left(\Ccal^2(\rho_\trans)-1\right)\mathcal{P}(\rho_\trans).}
Equating the two previous equations completes the proof.
\end{proof}
Note that Eq. \eqref{eq:QCSlosschannel} can also be rewritten as 
\eq{\Ccal^2(\rho_\trans)
=-\trans\frac{\partial H_2(\rho_\trans)}{\partial\trans}+1\label{QCSwithlog} ,}
where $H_2(\rho_\trans)$ is the 2-Rényi entropy.

Equation~\eqref{eq:QCSlosschannel} was already derived in a more general context in Ref.~\cite{Hertz} where the focus was on the evolution of the QCS over time\footnote{Note that in this previous work time increases from 0, while in this paper loss decreases from 1.} for a state in the presence of an environment that was not only the vacuum but any thermal state.
It was also proven in Ref.~\cite{Hertz} that the QCS is directly linked to the rate of decoherence of a state; viz., its lossy evolution: the greater the initial nonclassicality, the faster the state loses its nonclassicality. The following theorem now conclusively shows that, for all nonclassical pure states, the QCS will decrease to 1 after exactly 50\% loss.
This proves the conjecture in Ref.~\cite{HGH}.
\begin{thm}[QCS for pure states subject to loss]\label{thm:QCSunderlossPure}
    QCS=1 for all pure states that lose 50\% of their photons. QCS$\leq$1 for states that lose more than 50\% of their photons. QCS $\geq$ 1 for all pure states that lose less than 50\% of their photons.
\end{thm}
\begin{proof}
By Lemma~\ref{thm:QCSintermsofpurity}, the QCS can be rewritten as a function of the derivative of the purity. Lemma~\ref{lem:sym} and Theorem~\ref{thm:pur} establish when the QCS is greater, smaller or equal to~1.
\end{proof}

\begin{thm}[QCS for mixed states subject to loss; Conjecture \cite{HGH}]\label{thm:QCSunderlossMixed}
    $QCS \leq 1$ for all states subject to 50\% loss or more; no state subject to 50\% or more loss can be certified as nonclassical by the QCS.
\end{thm}

\begin{proof}
    This a direct consequence of Lemmas~\ref{thm:QCSintermsofpurity} and Corollary~\ref{thm:overlap}.
\end{proof}
    Theorem \ref{thm:QCSunderlossMixed} implies that the value of the QCS will never be above 1 when a state suffers more than 50\% loss. At this turning point, no classicality can be detected with the QCS witness and none will be detected for any further loss. Remark that the $P$-function of a lossy state is a rescaling of the initial $P$-function (see Eq.~\eqref{s_quasi_distribution_evolution}); if the latter is nonpositive, so will be that of the lossy state. Nevertheless, as the $P$-function is not unique, it is not enough to conclude that the state is still nonclassical. We will call those states weakly nonclassical because the QCS, which is also an estimate of the distance to the set of classical states \cite{Debievre}, is small.

Finally, using the results of the previous section (see Eq.~\eqref{eq:derPwithaadag}) and the formulation of the QCS as in Lemma~\ref{thm:QCSintermsofpurity}, we can write the QCS as the expectation value of the loss channel Lindbladian:
\eq{       \begin{aligned}
   \Ccal^2(\rho_\trans)&=   
   \frac{2 }{\tr[\rho_\trans^2]}\tr[a^\dagger a\rho_\trans\rho_\trans-a\rho_\trans a^\dagger \rho_\trans]+1\\
   &= \frac{2 \trans}{\tr[\rho_\trans^2]}\left(\frac{\tr[a^\dagger a \rho_\trans^2]}{T}-\frac{\tr[a^\dagger a \rho_{1-\trans}^2]}{1-T}\right)+1.
       \end{aligned}}

\section{Corollaries for quasiprobability distributions}
\label{sec:InequalityIntegrals}
Expressing the previous results in terms of quasiprobability distributions allows us to define some equalities and inequalities for integrals of those distributions.

We start by expressing the purity of a lossy state in terms of its $P$-function using the definition of $\rho_\trans$ in Eq.~\eqref{eq:rhoT}:
\begin{align}\label{thm:PurWithP}
&\mathcal{P}(\rho_\trans)=\tr[\rho_\trans^2]\nonumber\\
&=\!\tr\left[\!\int \!\dd^2\alpha \,\dd^2\beta  P_{\rho_1}(\alpha )P_{\rho_1}(\beta )\ketbra{\sqrt{\trans}\alpha }{\sqrt{\trans}\alpha }\sqrt{\trans}\beta \rangle\bra{\sqrt{\trans}\beta }\right]\nonumber\\
& =\!\int \dd^2\alpha \,\dd^2\beta \  P_{\rho_1}(\alpha )P_{\rho_1}(\beta ) \e^{-\trans|\alpha -\beta |^2}.
\end{align}

Derivatives of this expression will result in several inequalities for quasiprobability distributions. Using Eq.~\eqref{thm:PurWithP}, the $k$th derivative of the purity is given by
\eq{\frac{\partial^k \mathcal{P}(\rho_\trans)}{\partial\trans^k}\!=\!(-1)^k\!\!\int \dd^2\alpha \,\dd^2\beta  \,P_{\rho_1}(\alpha )P_{\rho_1}(\beta ) \\
\times|\alpha -\beta |^{2k}\e^{-\trans|\alpha -\beta |^2}.}
Then, according to Lemma~\ref{lem:sym}, Theorem~\ref{thm:pur}, and Corollary~\ref{thm:overlap} we have the following results:
\begin{corollary}[Inequalities for quasiprobability distributions]
\label{cor:inequalityforquasiprob}
    \begin{itemize}[wide=0.5em, leftmargin =*, before = \leavevmode\vspace{2pt}]
    \item  For pure states $\rho_1=\ketbra{\psi}{\psi}$, the integral
\eq{\int \dd^2\alpha \,\dd^2\beta \, P_{\rho_1}(\alpha )P_{\rho_1}(\beta ) |\alpha -\beta |^2\e^{-\trans|\alpha -\beta |^2}\label{eq:1derP}} is symmetric about $\trans=1/2$, positive for $\trans<1/2$ negative for $\trans>1/2$, and zero for $\trans=1/2$. 
\item For mixed states, the integral \eqref{eq:1derP} is nonnegative for $\trans\leq 1/2$.
\item Setting $\trans=0$ in Eq.~\eqref{eq:1derP}, for all mixed states, 
\eq{\int \dd^2\alpha \,\dd^2\beta \  P_{\rho_1}(\alpha )P_{\rho_1}(\beta ) |\alpha -\beta |^2\geq0. }
%
    \item For pure states $\rho_1=\ketbra{\psi}{\psi}$ and $\trans\in[0,1]$, and for all mixed states and $\trans\leq1/2$, \eq{\int \dd^2\alpha \,\dd^2\beta  \ P_{\rho_1}(\alpha )P_{\rho_1}(\beta ) |\alpha -\beta |^4\e^{-\trans|\alpha -\beta |^2}\geq0. \label{eq:2derP}}
\item Setting $\trans=0$ in Eq.~\eqref{eq:2derP}, $\forall\rho_1$
    \eq{\int \dd^2\alpha \,\dd^2\beta  \ P_{\rho_1}(\alpha )P_{\rho_1}(\beta ) |\alpha -\beta |^4\geq0.}

\item  For pure states $\rho_1=\ketbra{\psi}{\psi}$ and odd $k$, the integral
\eq{\int \dd^2\alpha \,\dd^2\beta \, P_{\rho_1}(\alpha )P_{\rho_1}(\beta ) |\alpha -\beta |^{2k}\e^{-\trans|\alpha -\beta |^2}\label{eq:kderP}} is symmetric about $\trans=1/2$, positive for $\trans<1/2$ negative for $\trans>1/2$, and zero for $\trans=1/2$. 
\item  For pure states $\rho_1=\ketbra{\psi}{\psi}$, and even $k$, the integral~\eqref{eq:kderP} is nonnegative for $\trans\in \mathbb{R}$.
\item For mixed states, the integral \eqref{eq:kderP} is nonnegative for $\trans\leq 1/2$.

\end{itemize}
\end{corollary}

Remark that all the inequalities in Corollary~\ref{cor:inequalityforquasiprob} are in contrast to standard manipulations with $P$-functions that produce inequalities by considering $P$ to be a valid probability distribution (such in Ref.~\cite{Richter,Bohmann}); whereas, here these relationships are true even if $P$ is negative or more singular than a delta function. Expressions for other quasiprobability distributions can be found in the same manner by using Eq.~\eqref{overlap} to express the purity of the lossy state.

Similarly, notice that  \eq{\int d^2\beta \ P_{\rho_1}(\beta )\e^{-\trans|\alpha \!-\!\beta |^2} =\frac{\pi}{\trans}P_{\rho_1}\left(\alpha ,1\!-\!\frac{2}{\trans}\right)} is another quasiprobability distribution (see Eq.~\eqref{eq:convolution}). Hence, the purity of a lossy state can also be expressed as
\eq{\tr[\rho_\trans^2]=\int \dd^2\alpha \ \frac{\pi}{\trans} P_{\rho_1}(\alpha )P_{\rho_1}\left(\alpha ,1-\frac{2}{\trans}\right),}
and taking derivatives of this expression leads to inequalities similar to the ones in Corollary~\ref{cor:inequalityforquasiprob}.

What these mean is that the extent of the singularities or negativities of quasiprobability distributions in phase space is always constrained for all physical states. Whenever a quasiprobability distribution is not regular, it becomes more regular when convolved with a Gaussian~\cite{CahillGlauber,AgarwalWolf1970,Lee1991,Lee1992,KieselVogel2010,Agudeloetal2013,Sperling2016,Lemosetal2018}, so our inequalities with various values of $\trans$ inform the required width of the Gaussians $\exp(-\trans|\alpha-\beta|^2)$ for making the distributions more regular. They also bound all of the phase space moments of the form $|\alpha-\beta|$ when taken in terms of the purity function $P_{\rho_1}(\alpha)P_{\rho_1}(\beta)\e^{-\trans|\alpha-\beta|^2}$, even when $|\alpha-\beta|$ is manifestly positive everywhere. In a competition between $|\alpha-\beta|$ growing and $\exp(-\trans|\alpha-\beta|^2)$ shrinking differently with different values of $\trans$, all of our inequalities must be satisfied for all physical states, which limits their singularities.

We can also write expressions for the overlap of two lossy states $\rho_\trans$ and $\sigma_\trans$ and obtain more general results. 
Knowing that the overlap can be computed as an integral of Wigner functions and using Eq.~\eqref{s_quasi_distribution_evolution}, we compute:
\eq{\begin{aligned}
    \tr[\rho_\trans\sigma_\trans]&=\pi\int \dd^2\alpha \ W_{\rho_\trans}(\alpha )W_{\sigma_\trans}(\alpha )\\
&=\pi\int \dd^2\alpha \ P_{\rho_\trans}(\alpha ,0)P_{\sigma_\trans}(\alpha ,0)\\
&=\!\frac{\pi}{\trans^2}\!\int \dd^2\alpha \ P_{\rho_1}\!\left(\!\frac{\alpha }{\sqrt{\trans}},1\!-\!\frac1{\trans}\!\right)\!P_{\sigma_1}\!\left(\!\frac{\alpha }{\sqrt{\trans}},1\!-\!\frac1{\trans}\!\right)\\
&=\frac{\pi}{\trans}\int \dd^2\beta \ P_{\rho_1}\left(\beta ,1-\frac1{\trans}\right)P_{\sigma_1}\left(\beta ,1-\frac1{\trans}\right), 
\label{thm:PurWithQuasiProb}
\end{aligned}}
from which we derive the following two corollaries:

\begin{corollary}\label{corr:rhosigmaHusimi}
 For $\trans\leq1/2$ and $\forall\rho_1,\sigma_1$
 $$
\begin{aligned}
\frac{1}{1-2\trans  }\int &\dd^2\alpha \dd^2\alpha ^\prime\ P_{\rho_1}\left(\alpha ,-1\right)P_{\sigma_1}\left(\alpha ^\prime,-1\right)\\
    \times& \left(\frac{2}{1-2\trans}-\frac{|\alpha -\alpha ^\prime|^2}{(1-2\trans)^2}\right)\e^{-\frac{\trans|\alpha -\alpha ^\prime|^2}{1-2\trans}}\leq0
   \end{aligned}
$$
\end{corollary}
\begin{proof}
Noticing that $s=1-\frac{1}{\trans}\leq -1$ for $\trans\leq 1/2$, we can express the $s-$quasiprobability distribution in Eq.~\eqref{thm:PurWithQuasiProb}
in terms of convolutions of Gaussians to distributions with $\tau=-(1+s)=\frac{1}{\trans}-2$ (as in Eq. \eqref{eq:convolution}):
\begin{equation}
    \begin{aligned}
        \tr[\rho_\trans\sigma_\trans]=&\frac{2^2\pi}{\trans(\pi  \tau)^2}\int \dd^2\beta \ \dd^2\alpha \ \dd^2\alpha ^\prime\ \\
        &\times P_{\rho_1}\left(\alpha ,-1\right)P_{\sigma_1}\left(\alpha ^\prime,-1\right) \e^{-\frac{2|\alpha -\beta |^2}{\tau}}\e^{-\frac{2|\alpha ^\prime-\beta |^2}{\tau}}\\
        =&\frac{1}{\trans  \tau}\!\int \!\dd^2\alpha \dd^2\alpha ^\prime P_{\rho_1}\left(\alpha ,\!-\!1\right)P_{\sigma_1}\left(\alpha ^\prime,\!-\!1\right)\e^{-\frac{|\alpha \!-\!\alpha ^\prime|^2}{\tau}}.
    \end{aligned}
    \end{equation}
Next, we take derivatives of this expression with respect to $\trans$, noting that $\frac{\partial \tau^{-1}}{\partial \trans}=\frac{1}{\tau^2\trans^2}=\frac{1}{(1-2\trans)^2}$:
\begin{equation}
\begin{aligned}
        \frac{\partial\tr[\rho_\trans\sigma_\trans]}{\partial \trans}=&\frac{1}{\trans  \tau}\int \dd^2\alpha \dd^2\alpha ^\prime\, P_{\rho_1}\left(\alpha ,-1\right)P_{\sigma_1}\left(\alpha ^\prime,-1\right)\\&
    \times \left(\frac{2}{1-2\trans}-\frac{|\alpha -\alpha ^\prime|^2}{(1-2\trans)^2}\right) \e^{-\frac{|\alpha -\alpha ^\prime|^2}{\tau}}.\nonumber
   \end{aligned}
\end{equation}
    Corollary \ref{thm:overlap} complete the proof.
\end{proof}
In this last result, it is interesting to note that, while  $P_{\rho_1}(\alpha ,-1)$ and $P_{\sigma_1}(\alpha ^\prime,-1)$ are positive everywhere (they are Husimi functions of a state), it is not obvious to conclude anything about this integral since   $\left(\frac{2}{1-2\trans}-\frac{|\alpha -\alpha ^\prime|^2}{(1-2\trans)^2}\right)$ could be positive or negative. This Corollary can also be used to determine if a specific Husimi function is the quasiprobability distribution of a valid state. For example, in Fig.~\ref{fig:Husimi}, we plot three Gaussian quasiproability distributions. The first one is the Husimi function of the vacuum $P_{\ket{0}}(\alpha ,-1)=\frac{|\alpha |^{2n}}{\pi n!}\e^{-|\alpha |^2}$ \cite{linowski2023relating}, the second one is defined as $\frac{1}{2}P_{\ket{0}}(\alpha /\sqrt{2},-1)$, and the third one as $2 P_{\ket{0}}(\sqrt{2}\alpha  ,-1)$, respectively stretching and compressing the vacuum Husimi function. We verify whether Corollary~\ref{corr:rhosigmaHusimi} is satisfied by these different quasiprobability distributions: Corollary~\ref{corr:rhosigmaHusimi} is verified when choosing $\rho_1$ and $\sigma_2$ to be the vacuum Husimi function and the second Gaussian but not when they are chosen to be the vacuum Husimi function and the third Gaussian distribution, meaning that the third Gaussian distribution does not represent a valid quantum state. Our inequalities thus quantify how the possible widths of features of quasiprobability distributions are constrained for physical states.

\begin{figure*}
    \centering
    \includegraphics[width=0.98\textwidth]{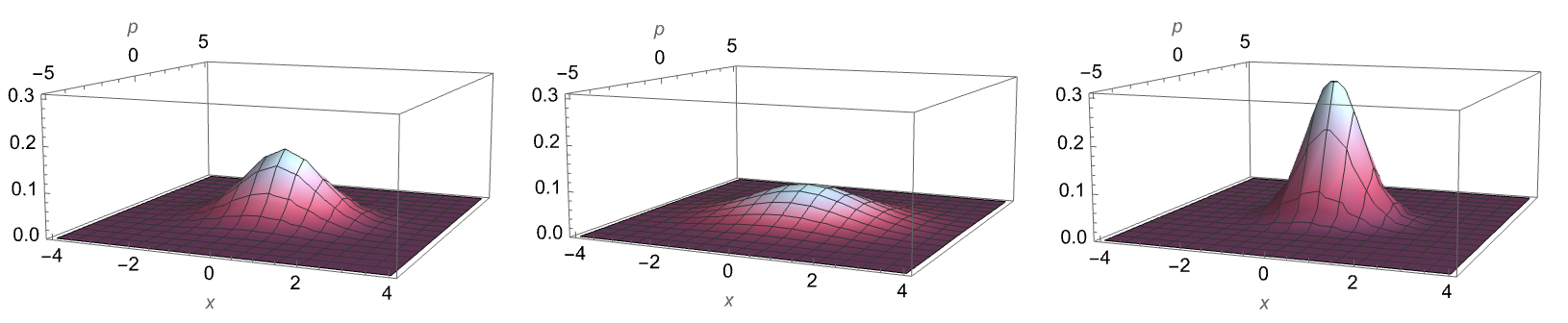}
    \caption{On the left is the Husimi quasiprobability distribution of the vacuum $P_{\ket{0}}(\alpha ,-1)=\frac{|\alpha |^{2n}}{\pi n!}\e^{-|\alpha |^2}$. The middle distribution is defined as $\frac{1}{2}P_{\ket{0}}(\alpha /\sqrt{2},-1)$, a dilated version of the vacuum Husimi distribution, and the right distribution as $2 P_{\ket{0}}(\sqrt{2}\alpha ,-1)$, a compressed version of the vacuum Husimi distribution. }
    \label{fig:Husimi}
\end{figure*}

\begin{corollary}\label{corr:rhosigmaPgeneral} 
For  $\trans\leq1/2$ and $r$ and $r^\prime$ chosen such that $rT-T+1>s>-r^\prime T+T-1$ for some $s\in \R$, 
$$\begin{aligned}
\int \dd^2\alpha \ \dd^2\alpha' \ P_{\rho_1}&\left(\alpha ,r\right) 
P_{\sigma_1}\left(\alpha' ,\tilde{r}-r+2\right)\e^{-\frac{2\trans|\alpha -\alpha '|^2}{2+\tilde{r}\trans}}\\
&\times
\frac{4 \left(2 |\alpha -\alpha '|^2+\tilde{r}\right)+2 \trans \tilde{r}^2}{(\trans \tilde{r}+2)^3}
\geq0
   \end{aligned}$$
 after defining $\tilde{r}=r+r^\prime-2$.
\end{corollary}
\begin{proof}
 From Eq.~\eqref{s_quasi_distribution_evolution}
and the definition of the overlap as in Eq.~\eqref{overlap}, we write for any $s\in\R$
\begin{equation}
\begin{aligned}
\tr[\rho_\trans\sigma_\trans]=\ \pi\ \int \dd^2\alpha \ &P_{\rho_\trans}(\alpha ,s)P_{\sigma_\trans}(\alpha ,-s)\\
=\frac{\pi}{\trans^2}\int \dd^2\alpha \ &P_{\rho_1}\left(\frac{\alpha }{\sqrt{\trans}},\frac{s+\trans-1}{\trans}\right)\\
   \times &P_{\sigma_1}\left(\frac{\alpha }{\sqrt{\trans}},\frac{-s+\trans-1}{\trans}\right)\\
=\!\frac{\pi}{\trans}\!\int\!\dd^2\,\beta P_{\rho_1}&\!\left(\!\beta ,\frac{s\!+\!\trans\!-\!1}{\trans}\!\right)
   P_{\sigma_1}\!\left(\!\beta ,\!\frac{\!-s\!+\!\trans\!-\!1}{\trans}\!\right)
   \end{aligned}
\end{equation}
Then, as in Eq.~\eqref{eq:convolution}, we can write a quasiprobability distribution as a convolution of Gaussian with $r>\frac{s+\trans-1}{\trans}$ and $r'>\frac{-s+\trans-1}{\trans}$:
\begin{equation}
\begin{aligned}
\tr[\rho_\trans\sigma_\trans]
   =&\int\dd^2\beta \ \dd^2\alpha \ \dd^2\alpha' \,\e^{-\frac{2\trans|\alpha -\beta |^2}{(r\trans-s-\trans+1)}-\frac{2\trans|\alpha '-\beta |^2}{(r'\trans+s-\trans+1)}} \\
   &\times    \frac{4\trans P_{\rho_1}\left(\alpha ,r\right)P_{\sigma_1}\left(\alpha' ,r'\right)}{\pi(r\trans-s-\trans+1)(r'\trans+s-\trans+1)}
   \\
    =&\frac{2}{2+\tilde{r}\trans}\int \dd^2\alpha \ \dd^2\alpha' \,P_{\rho_1}\left(\alpha ,r\right)
P_{\sigma_1}\left(\alpha' ,r'\right)\\
&\times\e^{-\frac{2\trans|\alpha -\alpha '|^2}{2+\Tilde{r}\trans}}
   \end{aligned}\label{overlaprrprime}
\end{equation}
where $\tilde{r}=r+r'-2>-2/\trans\geq-4$ for $\trans\leq1/2$.
    Next, we compute the derivative. 
    \begin{equation}
\begin{aligned}
\frac{\partial\tr[\rho_\trans\sigma_\trans]}{\partial T}
    =&\int\dd^2\alpha \ \dd^2\alpha' \, P_{\rho_1}\left(\alpha ,r\right)
P_{\sigma_1}\left(\alpha' ,r'\right)\\
&\times
\frac{-4 \left(2 |\alpha -\alpha '|^2+\tilde{r}\right)-2 \trans \tilde{r}^2}{(\trans \tilde{r}+2)^3}\e^{-\frac{2\trans|\alpha -\alpha '|^2}{2+\tilde{r}\trans}}
\\
   \end{aligned}
\end{equation}
Corollary \ref{thm:overlap} completes the proof.
\end{proof}
Note that when $r=-1$ and $\tilde{r}=-4$, Corollary~\ref{corr:rhosigmaPgeneral} is equivalent to Corollary~\ref{corr:rhosigmaHusimi} and, when $\rho=\sigma$, $r=1$, and $\tilde{r}=0$, 
Eq.~\eqref{overlaprrprime} is equivalent to Eq.~\eqref{thm:PurWithP}.

\section{Corollaries for integrals of characteristic functions}
\label{sec:CharactersiticFunctions}

Similar to the previous section, we can analyze the purity of a state when computed through its characteristic function and derive some corollaries.

Combining Eq.~\eqref{eq:puritywithChi} with Eq.~\eqref{eq:evolutionofchi} yields a formula for the purity of a lossy state $\rho_\trans$ in terms of its $s$-ordered characteristic function
\begin{align}
    \tr[\rho_T^2]
    & = \int \frac{\dd^2 \alpha}{\pi} \ \e^{-s |\alpha|^2} \left | \chi_{\rho_1} \left (\sqrt{T} \alpha, \frac{s + T - 1}{T} \right ) \right |^2 \  \nonumber\\
    & = \int \frac{\dd^2 \alpha}{\pi} \  \frac{\e^{- \frac{s}{T} |\alpha|^2}}{T} \left | \chi_{\rho_1} \left (\alpha, \frac{s + T - 1}{T} \right ) \right |^2 \  \nonumber\\
    & = \int \frac{\dd^2 \alpha}{\pi} \  \frac{\e^{- (s - 1 + \frac1{T}) |\alpha|^2}}{T} \left | \chi_{\rho_1} \left (\alpha, s \right ) \right |^2 \nonumber\\
    &= \int \frac{\dd^2 \alpha}{\pi} \  \frac{\e^{-|\alpha|^2 / T}}{T} \left | \e^{\frac{1-s}{2} |\alpha|^2} \chi_{\rho_1} \left (\alpha, s \right ) \right |^2 \label{eq:puritywithChiforloss}.
\end{align}

Compared to expressions in Eqs.~\eqref{thm:PurWithP} or \eqref{thm:PurWithQuasiProb}, which compute the purity with $s$-quasiprobability distributions, one advantage of Eq.~\eqref{eq:puritywithChiforloss},  which computes the purity with the characteristic function, is to reduce the formula from two phase space integrals to one. 
Moreover, the integrand is manifestly nonnegative, showing that purity is a convex combination of exponential factors $e^{-|\alpha|/\trans}/\trans$. These factors individually are not convex in $\trans$, making it all the more remarkable that purity is convex in $\trans$ (see Theorem~\ref{thm:pur} and Corollary~\ref{thm:overlap}).
On the other hand, the factor $\e^{-|\alpha|/\trans}$ is monotonic in $\trans$, which reveals that $\trans \tr[\rho_T^2]$ is monotonic in $\trans$. 

A disadvantage of this formula, however, is that the $k$th derivative with respect to $T$ no longer takes a simple form. Nevertheless, we can still use Lemma~\ref{lem:sym}, Theorem~\ref{thm:pur}, and Corollary~\ref{thm:overlap} to turn the $k$th derivative into equalities and inequalities for integrals involving the characteristic function of the initial state, similar to what was done in Corollary~\ref{cor:inequalityforquasiprob}. The first two derivatives are
\begin{align}
    \frac{\partial\tr[\rho_\trans^2]}{\partial \trans}\!=&\! \int \frac{\dd^2 \alpha}{\pi}\ \left (\frac{|\alpha|^2 \!-\! \trans}{\trans^2} \right ) \frac{\e^{-\frac{|\alpha|^2 }{\trans}}}{\trans} \left | \e^{\frac{1-s}{2} |\alpha|^2}  \chi_{\rho_1} \!\left (\alpha, s \right ) \right |^2,\\
    \frac{\partial^2\tr[\rho_\trans^2]}{\partial \trans^2} =& \int \frac{\dd^2 \alpha}{\pi}\ \left (\frac{|\alpha|^4 - 4|\alpha|^2 \trans + 2 \trans^2}{\trans^4} \right ) \frac{\e^{-|\alpha|^2 / \trans}}{\trans} \nonumber\\
    &\times \left | \e^{\frac{1-s}{2} |\alpha|^2}  \chi_{\rho_1} \left (\alpha, s \right ) \right |^2.
\end{align}
The second derivative must be nonnegative for an initial pure state. The integrand above is almost always nonnegative, but becomes nonpositive in the region of phase space where $T (2 - \sqrt{2}) \leq |\alpha|^2 \leq T (2 + \sqrt{2})$. Qualitatively, the support of the characteristic function of a pure state in this region must be small compared to everywhere else in phase space.

Reexamining the formula for purity in terms of the characteristic function reveals that purity is proportional to the Laplace transform of a nonnegative function.
\begin{align}
    \tr[\rho_\trans^2] 
    & = \frac{1}{\pi T} \int \dd^2 \alpha \ \e^{-\frac{|\alpha|^2 }{\trans}} |\chi_{\rho_1}(\alpha, 1)|^2 \nonumber\\
    & = \frac{1}{T} \int_0^\infty \dd \tau \ \e^{-\frac{\tau}{\trans}} \int_0^{2 \pi} \frac{\dd \theta}{2 \pi}  |\chi_{\rho_1}(\sqrt{\tau} \e^{i \theta}, 1)|^2 \nonumber\\
    & = \frac{1}{T} L\{|\bar{\chi}_{\rho_1}|^2\} (1/T).
\end{align}
The last line utilizes the definition of the Laplace transform
\begin{equation}
    L\{f\}(\omega) = \int_0^\infty \dd\tau\,\e^{- \omega \tau} f(\tau)
\end{equation}
applied to the phase-averaged squared amplitude of the normally ordered characteristic function
\begin{equation}
    |\bar{\chi}_\rho(\tau)|^2 = \int_0^{2 \pi} \frac{\dd \theta}{2 \pi} |\chi_{\rho} (\sqrt{\tau} \e^{i \theta}, 1)|^2.
\end{equation}
Expressing purity in terms of a Laplace transform allows us to apply Bernstein's theorem \cite{Bernstein}
on monotone functions, which states that the Laplace transform of a nonnegative Borel measure is completely monotonic on the interval $[0, \infty)$. A function $f(x)$ is completely monotonic when it and all of its even derivatives are nonnegative and decreasing on that interval, that is
\begin{equation}
    (-1)^k \frac{d^k f}{dx^k} \geq 0.
\end{equation}
Applied to purity, this theorem shows that $T \tr[\rho_T^2]$ is completely monotonic in $1/T$ and so
\begin{equation}
    (-1)^k \frac{d^k T \tr[\rho_T^2]}{d(1/T)^k} \geq 0.
\end{equation}
Using chain rule simplifications, we thus obtain the following corollary:
\begin{corollary}\label{cor:Bernstein}
    For all states $\rho_1$ undergoing loss $\trans$,
    \begin{equation}
    \left (T^2 \frac{d}{d T} \right )^k \Big( T \tr[\rho_T^2] \Big ) \geq 0.\nonumber
\end{equation}
\end{corollary}
The inequalities corresponding to the first and second derivatives are
\begin{equation}
    T^2 \left (1 + T \frac{d}{d T} \right ) \tr[\rho_T^2] \geq 0 \label{Eq:equivtoQCSpositif}
\end{equation} and
\begin{equation}
    T^3 \left (2 + 4 T \frac{d}{d T} + 4 T^2 \frac{d^2}{d T^2} \right ) \tr[\rho_T^2] \geq 0.
\end{equation}
Note that, using Lemma \ref{thm:QCSintermsofpurity}, Eq.~\eqref{Eq:equivtoQCSpositif} is the same as $\Ccal^2(\rho_\trans)\geq0$, which is always verified since the QCS is a nonnegative quantity. These inequalities can then be applied to any method for expressing purity of a state in terms of operators or phase-space distributions to glean further properties of all quantum states.

\section{Miscellaneous results}
In attempting to prove the monotonicity of $\partial \pur/\partial T$, we proved the following Lemma:
\begin{lem}[Monotonicity of number-operator purity]
    The quantity $\tr[a^\dagger a\rho_T^2]$ is nondecreasing with $T$. Similarly, the quantity $\tr[a \rho_T a^\dagger \rho_T](1-T)/T$ is nonincreasing with $T$.
\end{lem}
\begin{proof}
There is an alternate form of R\'enyi relative entropy that was proven in Ref.~\cite{MosonyiHiai2011} to obey certain monotonicity relations. These stem from the ``$\alpha$-quasi-relative entropies'' defined by
\begin{equation}
    Q_\alpha(A||B)=\text{sign}(\alpha-1)\tr[A^\alpha B^{1-\alpha}]
\end{equation} for $\alpha\in[0,1)$ or $\alpha\in(1,2]$ and the support of $A$ being contained in the support of $B$, for any positive semidefinite operators $A$ and $B$. These relative entropies are monotonically decreasing under any quantum channel $A,B\to \mathcal{E}(A),\mathcal{E}(B)$.

Take $\alpha=2$, $A=\rho_1$, and $B=\frac{1}{\varepsilon}|0\rangle\langle 0|+\sum_{m\geq 1} \frac{1}{m}|m\rangle\langle m|$, where $B$ is approximately $1/(a^\dagger a)$ but is nonsingular for the vacuum state. $B$ has support everywhere for $|\varepsilon|<\infty$. Then, $\tr[\rho_T^2\mathcal{E}_T[B]^{-1}]$
increases monotonically with $T$ due to the multiplicative nature of loss channels. The magic happens by applying the loss channel on each Fock state as in Eq.~\eqref{eq:Fock loss binom}:
\begin{equation}
    \begin{aligned}
    \mathcal{E}_T\Big[\frac{1}{\varepsilon}|0\rangle\langle 0|+&\sum_{m\geq 1} \frac{1}{m}|m\rangle\langle m|\Big]\\
    =&\left(\frac{1}{\varepsilon}+\sum_{m\geq 1}\frac{(1-T)^m}{m}\right)\!|0\rangle\langle 0|\\
    &+\sum_{k=1}^\infty|k\rangle\langle k|
    \sum_{m\geq k}\binom{m}{k}\frac{T^k(1-T)^{m-k}}{m}\\
    =&\left(\frac{1}{\varepsilon}-\log T\right)|0\rangle\langle 0|+\sum_{m\geq 1}\frac{1}{m}|m\rangle\langle m|.
    \end{aligned}
\end{equation}  Putting together $B^{-1}$ and $\mathcal{E}_T[B]^{-1}$ yields, by the monotonicity of relative entropies
\begin{equation}
    \tr[\rho_T^2 a^\dagger a]\leq \tr[\rho_{T^\prime}^2 a^\dagger a]+\varepsilon\langle 0|\frac{\rho_{T^\prime}^2}{1-\varepsilon\log T^\prime}-\frac{\rho_T^2}{1-\varepsilon\log T}|0\rangle 
\end{equation} for all $T\leq T^\prime$. Since $\varepsilon$ can be made arbitrarily small, it follows that $\tr[\rho_T^2 a^\dagger a]$ never decreases with increasing $T$. 
Then, using Eq.~\eqref{eq:transpose trick}, $\tr[a\rho_T a^\dagger\rho_T](1-T)/T$ decreases monotonically with $T$. 
\end{proof}

The results for purity's convexity have implications for mathematical functions when considered for particular states. For example, choosing Fock states, we learn about hypergeometric functions and derivatives thereof that yield more hypergeometric functions. This has arisen recently in the mathematics literature from the perspective of Bernstein polynomials, where log convexity was also conjectured and demonstrated~\cite{GavreaIvan2014,Nikolov2014,Rasa2018,Rasa2019,Alzer2020}.
\begin{corollary}[Convexity of hypergeometric function]
    The function $(1-\trans)^{2 n} \, _2F_1\left(-n,-n;1;\frac{\trans^2}{(\trans-1)^2}\right)$ is convex in $\trans\in \mathbb{R}$ for all nonnegative integers $n$ and symmetric about $\trans=1/2$.
\end{corollary}
\begin{proof}
    The purity of a Fock state subject to loss, given in Eq.~\eqref{eq:Fock loss binom}, is \eq{\tr[\mathcal{E}_T[|n\rangle\langle n|]^2]&=\sum_{k=0}^n\left(\binom{n}{k}T^k(1-T)^{n-k}\right)^2\\
        &=(1-\trans)^{2 n} \, _2F_1\left(-n,-n;1;\frac{\trans^2}{(\trans-1)^2}\right).} 
        This is convex due to Theorem~\ref{thm:pur} and symmetric due to Lemma~\ref{lem:sym}.
\end{proof}

\section{Conclusion}
Continuous-variable quantum information processing bridges information theoretic notions of nonclassicality such as entanglement and entropy to physical descriptions of a state through quasiprobability distributions in the phase space. Our companion paper~\cite{CompanionShortarXiv} provided a set of proofs that resolved a 20-year old conjecture that mixing any state with vacuum at a beam splitter generates the most entanglement when the beam splitter is balanced. By further examining various aspects of this work~\cite{CompanionShortarXiv}, we uncovered numerous results for fundamental and applied quantum optics. The ramifications include the impact of loss on nonclassicality measures and similarity measure of quantum states, equalities and inequalities for quasiprobability distributions, quantum mutual information at the output of a beam splitter, and inequalities for creation and annihilation operators. We unveil links between entanglement and nonclassicality and demonstrate interconnectivity of notions in quantum optics and quantum information processing that will be useful throughout both fields and at the intersection of the two for the underlying theory of photonic, and any other bosonic, quantum technologies.

\begin{acknowledgments}
    The NRC headquarters is located on the traditional unceded territory of the Algonquin Anishinaabe and Mohawk people. NLG acknowledges funding from the NSERC Discovery Grant and Alliance programs.
\end{acknowledgments}

\appendix
\section{Purity in Fock basis}
\label{Appendix:purity Fock}
A mixed state state $\rho = \sum_{k, k^\prime} \rho_{k k^\prime} \ketbra{k}{k^\prime}$ subject to loss has its purity evolve to
\eq{
\sum_{ll^\prime}\trans^{l+l^\prime}\left|\sum_k \rho_{k+l, k+l^\prime} (1-T)^k\sqrt{\binom{k+l}{k}\binom{k+l^\prime}{k}}\right|^2\\
=\sum_{l l^\prime} \frac{\trans^{l + l^\prime}}{l! l^\prime !} |\tr[a^{\dagger l} (1 - \trans)^{a^\dagger a} a^{l^\prime} \rho]|^2.}

\section{Evidence for log-convexity of the purity}
\label{Appendix:logconvexity}

The problem of log-convexity with respect to the original loss parameter $\trans$ amounts to proving the inequality
\begin{align}
    \tr[\rho_1 \otimes \rho_1 (1 - 2 \trans)^{N_- - 1} N_-]^2 \leq \tr[\rho_1 \otimes \rho_1 (1 - 2 \trans)^{N_-}]  \nonumber
    \\ \times\tr[\rho_1 \otimes \rho_1 (1 - 2 \trans)^{N_- - 2} N_- (N_- - 1)]
\end{align}
for $\trans \in [0, 1]$. Contrast this with the convexity inequality
\begin{equation}
    0 \leq \tr[\rho_1 \otimes \rho_1 (1 - 2 \trans)^{N_- - 2} N_- (N_- - 1)],
\end{equation}
which we have already proven for $T \in (-\infty, 1/2]$ for general states (see Corrolary~\ref{thm:overlap}) and for all $T \in \mathbb{R}$ when the initial state is pure (see Theorem~\ref{thm:pur}). Although only positive $T \in [0, 1]$ corresponds to physical loss, the convexity inequality holds for $T \leq 0$. The same cannot be said for the log-convexity inequality. With $\rho_1 = \ketbra{1}{1}$, the log-convexity inequality becomes
\begin{equation}
    (1 - 2 T)^2 \leq \frac{1 + (1-2T)^2}{2},
\end{equation}
which fails as soon as $|T| > 1$.

Beyond the restriction on the set of admissible values for $T$, the log-convexity inequality is also more restrictive over the states that appear inside the trace. The purity formula $\tr[\rho_T^2] = \tr[\rho_1 \otimes \rho_1 (1-2T)^{N_-}]$ features the state $\rho_1 \otimes \rho_1$. One may consider the slightly more general expression $\tr[\Phi (1-2T)^{N_-}]$, where $\Phi$ is any $2$-mode state, potentially even an entangled one. This more general expression is convex for all nonnegative $\Phi$ and all $T \in [0, 1/2]$, but it is not log-convex on that domain. Consider the entangled state $\Phi = (\ket{00} - \ket{11})(\bra{00} - \bra{11}) / 2$.  The log-convexity inequality becomes $1 \leq 0$ which is clearly not satisfied for any value of $T$. Another illuminating example is the separable state $\Phi = \ketbra{01}{01}$ for which the inequality again becomes $1 \leq 0$. This means that a proof of log-convexity must exploit the structure of the state $\rho_1 \otimes \rho_1$, not merely its positivity. 

In fact, preliminary numerical evidence suggests that the expression $\tr[\Phi (1-2T)^{N_-}]$ is log-convex for $T \in [0, 1]$ whenever $\Phi$ can be expressed as a convex combination of states of the form $\rho_1 \otimes \rho_1$. This leads us to the following conjecture.
\begin{conjecture}[Unfairness witness
]
For any state of the form $\Phi = \sum_i P_i \rho_i \otimes \rho_i$, where $\rho_i \succeq 0$ (this denotes a positive semidefinite operator), $P_i \geq 0$, $\sum_i P_i = 1$ and any $\lambda \in [-1, 1]$
\begin{equation}
    \tr[\Phi \ \lambda^{N_-}] \tr[\Phi \ \lambda^{N_- - 2} N_- (N_- - 1)] \geq \tr[\Phi \  \lambda^{N_- - 1} N_-]^2 .
\end{equation}
\end{conjecture}
If the conjecture is true, the inequality provides a direct way to test whether or not the $2$-mode state $\Phi$ is ``fair'' 
in the sense that both modes contain the same state. Assuming the conjecture is true and a given state $\Phi$ fails the inequality for any value of $\lambda$, one can conclude that $\Phi$ does not constitute a ``fair'' distribution of light into the $2$ modes in the sense that each mode always contains the same state. Naturally, the test only provides a sufficient condition. There are ``unfair'' distributions (for example $\Phi = \ketbra{12}{12}$) that nevertheless pass the test. 

Finally, we note two extreme forms of the inequality, corresponding to $\lambda = 0$ and $\lambda = 1$ respectively. Let $\Phi_- = \tr_+[\Phi] = \tr_1[B(1/2) \Phi B(1/2)^\dag]$. Then, for $\lambda = 0$, the conjecture reads
\begin{equation}
    2 \bra{2} \Phi_- \ket{2} \times \bra{0} \Phi_- \ket{0} -\bra{1} \Phi_- \ket{1}^2 \geq 0.
\end{equation}
The quantity on the left-hand side can be measured experimentally using a detector that can distinguish between $n = 0$, $n = 1$, $n = 2$, or $ n \geq 3$ photons. For $\lambda = 1$, the conjecture becomes
\begin{equation}
    \tr[\Phi_- N (N - 1)] - \tr[\Phi_- N]^2 \geq 0.
\end{equation}

Interestingly, for $T \leq 1/2$, the log-convexity inequality can be reformulated in terms of the second-order correlation function $g^{(2)}$ \cite{Glauber1963quantumtheorycoherence,MandelWolf1995}
\begin{equation}
    g^{(2)}(\rho_-) \geq 1
\end{equation}
where
\begin{equation}
    g^{(2)}(\rho) = \frac{\tr[\rho N (N - 1)]}{\tr[\rho N]^2}
\end{equation}
and
\begin{equation}
    \rho_- = \frac{\tr_+[\sqrt{1 - 2 T}^{N_-} \rho_1 \otimes \rho_1 \sqrt{1 - 2 T}^{N_-}]}{\tr[\sqrt{1 - 2 T}^{N_-} \rho_1 \otimes \rho_1 \sqrt{1 - 2 T}^{N_-}]},
\end{equation}
with $\tr_+[\cdot]$ indicating a trace on the light port only (the exit port other than the dark port).
Setting $\trans = 0$ generates a remarkable conjecture: if two identical states are combined on a balanced beam splitter, the number statistics in the dark port can never be sub-Poissonian. This is not to say that the number statistics of the dark port are always classical in the sense of admitting a nonnegative $P$ distribution: two identical squeezed vacua impinging on a balanced beam splitter produce a separable state with a squeezed vacuum in each output port that has no positive $P$ distribution. In fact, it is not difficult to find a state that generates a dark port with nonclassicality in higher order correlations functions $g^{(n)}$, as these vanish for finite-dimensional states whenever $n$ is sufficiently large.

While a full proof of log-convexity remains elusive, we present here our results so far. First, the problem admits a useful simplification. For any $T \in [0, 1]$, $\rho_T$ is a physical state as long as $\rho_1$ is physical. Therefore, we need only prove that log-convexity holds for all physical states at $T = 1$. That is, it suffices to prove for all $\rho$
\begin{equation}
    \tr[\rho_1 \otimes \rho_1 \hat{S} N_-]^2 \leq \tr[\rho_1 \otimes \rho_1 \hat{S}] \tr[\rho_1 \otimes \rho_1 \hat{S} N_- (N_- - 1)]
\end{equation}
 where $\hat{S}$ is the SWAP operator as introduced in Section~\ref{sec:QCS}. This inequality can be recast as a single expectation value over $4$ copies of $\rho_1$.
\begin{equation}
    \tr[\rho_1^{\otimes 4} (\hat{S} \otimes \hat{S} N_- (N_- - 1) - \hat{S} N_- \otimes \hat{S} N_-)] \geq 0.
\end{equation}

If one can prove the log-convexity of the purity, some interesting results can be proven. For example,  we could prove that for an initial pure state--- whose QCS is necessarily greater or equal to 1---the decrease of the purity from $\trans=1$ to $\trans=1/2$ is monotonic with $\trans$. Indeed, 
taking the derivative of Eq.~\eqref{QCSwithlog} and using Conjecture~\ref{log-conjecture}, we have
$$\frac{\partial\Ccal^2(\rho_\trans)}{\partial\trans}\geq\frac{\Ccal^2(\rho_\trans)-1}{\trans}.$$
According to Theorem~\ref{thm:QCSunderlossPure}, $\Ccal^2(\rho_\trans)\geq1$ for $\trans\geq1/2$, which implies $\frac{\partial\Ccal^2(\rho_\trans)}{\partial\trans}\geq0$.
This result would only be valid for initially pure states and for less than 50\% loss. At this point, the QCS of the pure state will decrease to a value of 1. With more loss, the QCS might decrease more, but will eventually increase again to return, after 100\%, to the value 1 which is the QCS of the vacuum.

We finally note that being convex or log-convex under loss is not sufficient for an operator to be nonnegative. This is seen using the counterexample of the operator $\sigma=\frac{2}{3}\ketbra{0}{0}-\frac{1}{3}\ketbra{1}{1}+\frac{2}{3}\ketbra{2}{2}$, which is nonpositive but of trace 1 and of purity 1, evolving under loss.

\section{Second derivative of purity}
\label{app: second deriv}
The second derivative of purity with respect to loss may be written as
\begin{equation}
    \begin{aligned}
        \frac{\partial^2 \pur(T)}{\partial T^2}&=\frac{2}{T^2}\left(-\tr[N \rho^2]
        +2\tr[|a\rho a^\dagger|^2]+\tr[|N\rho|^2]
        \right.\\
        &\left.
        -4\Re\tr[\rho N  a \rho a^\dagger ]+\tr[a\rho a^{\dagger }(a^\dagger \rho a)]\right).
    \end{aligned}
\end{equation}
If we use the state evaluated at $1-T$ instead of $T$ and the transpose trick, we find
\begin{equation}
    \begin{aligned}
        \frac{\partial^2\pur}{\partial T^2}&=\frac{2}{T^2}\left(\tr[|a\rho a^\dagger|^2]+\tr[a\rho a^\dagger(a^\dagger \rho a)]
        \right.\\
        &\left.
        -2\Re\tr[\rho N a\rho a^\dagger]\right)+\frac{2}{(1-T)^2}\left(\tr[|a\rho(1-T) a^\dagger|^2]
        \right.\\
        &\left.
        +\tr[a\rho(1-T) a^\dagger(a^\dagger \rho(1-T) a)]
        \right.\\
        &\left.
        -2\Re\tr[\rho(1-T) N a\rho(1-T) a^\dagger]\right)
    \end{aligned}
\end{equation}
and
\begin{equation}
    \begin{aligned}
        \frac{\partial^2 \pur}{\partial T^2}
        &=\frac{2}{T^2}\left(-\tr[N \rho^2]+\tr[|a\rho a^\dagger|^2]+\tr[|N\rho|^2]
        \right.\\
        &\left.
        -
        2\Re\tr[\rho N a \rho a^\dagger ]\right)
        +
        \frac{2}{(1-T)^2}\left(
        -\tr[N \rho(1-T)^2]
        \right.\\
        &\left.
        +\tr[|a\rho(1-T)a^\dagger|^2]+\tr[|N\rho(1-T)|^2]\right.\\&\left.-
        2\Re\tr[\rho(1-T)N a \rho(1-T) a^\dagger ]\right).
    \end{aligned}
\end{equation}


\begin{thebibliography}{135}%
\makeatletter
\providecommand \@ifxundefined [1]{%
 \@ifx{#1\undefined}
}%
\providecommand \@ifnum [1]{%
 \ifnum #1\expandafter \@firstoftwo
 \else \expandafter \@secondoftwo
 \fi
}%
\providecommand \@ifx [1]{%
 \ifx #1\expandafter \@firstoftwo
 \else \expandafter \@secondoftwo
 \fi
}%
\providecommand \natexlab [1]{#1}%
\providecommand \enquote  [1]{``#1''}%
\providecommand \bibnamefont  [1]{#1}%
\providecommand \bibfnamefont [1]{#1}%
\providecommand \citenamefont [1]{#1}%
\providecommand \href@noop [0]{\@secondoftwo}%
\providecommand \href [0]{\begingroup \@sanitize@url \@href}%
\providecommand \@href[1]{\@@startlink{#1}\@@href}%
\providecommand \@@href[1]{\endgroup#1\@@endlink}%
\providecommand \@sanitize@url [0]{\catcode `\\12\catcode `\$12\catcode `\&12\catcode `\#12\catcode `\^12\catcode `\_12\catcode `\%12\relax}%
\providecommand \@@startlink[1]{}%
\providecommand \@@endlink[0]{}%
\providecommand \url  [0]{\begingroup\@sanitize@url \@url }%
\providecommand \@url [1]{\endgroup\@href {#1}{\urlprefix }}%
\providecommand \urlprefix  [0]{URL }%
\providecommand \Eprint [0]{\href }%
\providecommand \doibase [0]{http://dx.doi.org/}%
\providecommand \selectlanguage [0]{\@gobble}%
\providecommand \bibinfo  [0]{\@secondoftwo}%
\providecommand \bibfield  [0]{\@secondoftwo}%
\providecommand \translation [1]{[#1]}%
\providecommand \BibitemOpen [0]{}%
\providecommand \bibitemStop [0]{}%
\providecommand \bibitemNoStop [0]{.\EOS\space}%
\providecommand \EOS [0]{\spacefactor3000\relax}%
\providecommand \BibitemShut  [1]{\csname bibitem#1\endcsname}%
\let\auto@bib@innerbib\@empty
\bibitem [{\citenamefont {Lupu-Gladstein}\ \emph {et~al.}(2024)\citenamefont {Lupu-Gladstein}, \citenamefont {Hertz}, \citenamefont {Heshami},\ and\ \citenamefont {Goldberg}}]{CompanionShortarXiv}%
  \BibitemOpen
  \bibfield  {author} {\bibinfo {author} {\bibfnamefont {N.}~\bibnamefont {Lupu-Gladstein}}, \bibinfo {author} {\bibfnamefont {A.}~\bibnamefont {Hertz}}, \bibinfo {author} {\bibfnamefont {K.}~\bibnamefont {Heshami}}, \ and\ \bibinfo {author} {\bibfnamefont {A.~Z.}\ \bibnamefont {Goldberg}},\ }\href {https://arxiv.org/abs/2411.03423} {\  (\bibinfo {year} {2024})},\ \Eprint {http://arxiv.org/abs/2411.03423} {arXiv:2411.03423 [quant-ph]} \BibitemShut {NoStop}%
\bibitem [{\citenamefont {Asb\'oth}\ \emph {et~al.}(2005)\citenamefont {Asb\'oth}, \citenamefont {Calsamiglia},\ and\ \citenamefont {Ritsch}}]{Asbothetal2005}%
  \BibitemOpen
  \bibfield  {author} {\bibinfo {author} {\bibfnamefont {J.~K.}\ \bibnamefont {Asb\'oth}}, \bibinfo {author} {\bibfnamefont {J.}~\bibnamefont {Calsamiglia}}, \ and\ \bibinfo {author} {\bibfnamefont {H.}~\bibnamefont {Ritsch}},\ }\href {\doibase 10.1103/PhysRevLett.94.173602} {\bibfield  {journal} {\bibinfo  {journal} {Physical Review Letters}\ }\textbf {\bibinfo {volume} {94}},\ \bibinfo {pages} {173602} (\bibinfo {year} {2005})}\BibitemShut {NoStop}%
\bibitem [{\citenamefont {Asb\'oth}(2024)}]{Asboth2024}%
  \BibitemOpen
  \bibfield  {author} {\bibinfo {author} {\bibfnamefont {J.~K.}\ \bibnamefont {Asb\'oth}},\ }\href@noop {} {\enquote {\bibinfo {title} {Private communication},}\ } (\bibinfo {year} {2024})\BibitemShut {NoStop}%
\bibitem [{\citenamefont {Yuen}\ and\ \citenamefont {Shapiro}(1978)}]{YuenShapiro1978}%
  \BibitemOpen
  \bibfield  {author} {\bibinfo {author} {\bibfnamefont {H.~P.}\ \bibnamefont {Yuen}}\ and\ \bibinfo {author} {\bibfnamefont {J.~H.}\ \bibnamefont {Shapiro}},\ }in\ \href@noop {} {\emph {\bibinfo {booktitle} {Coherence and Quantum Optics IV}}},\ \bibinfo {editor} {edited by\ \bibinfo {editor} {\bibfnamefont {L.}~\bibnamefont {Mandel}}\ and\ \bibinfo {editor} {\bibfnamefont {E.}~\bibnamefont {Wolf}}}\ (\bibinfo  {publisher} {Springer US},\ \bibinfo {address} {Boston, MA},\ \bibinfo {year} {1978})\ pp.\ \bibinfo {pages} {719--727}\BibitemShut {NoStop}%
\bibitem [{\citenamefont {Yuen}\ and\ \citenamefont {Shapiro}(1980)}]{YuenShapiro1980}%
  \BibitemOpen
  \bibfield  {author} {\bibinfo {author} {\bibfnamefont {H.}~\bibnamefont {Yuen}}\ and\ \bibinfo {author} {\bibfnamefont {J.}~\bibnamefont {Shapiro}},\ }\href {\doibase 10.1109/TIT.1980.1056132} {\bibfield  {journal} {\bibinfo  {journal} {IEEE Transactions on Information Theory}\ }\textbf {\bibinfo {volume} {26}},\ \bibinfo {pages} {78} (\bibinfo {year} {1980})}\BibitemShut {NoStop}%
\bibitem [{\citenamefont {Yurke}(1985)}]{Yurke1985}%
  \BibitemOpen
  \bibfield  {author} {\bibinfo {author} {\bibfnamefont {B.}~\bibnamefont {Yurke}},\ }\href {\doibase 10.1103/PhysRevA.32.311} {\bibfield  {journal} {\bibinfo  {journal} {Phys. Rev. A}\ }\textbf {\bibinfo {volume} {32}},\ \bibinfo {pages} {311} (\bibinfo {year} {1985})}\BibitemShut {NoStop}%
\bibitem [{\citenamefont {Mandel}\ and\ \citenamefont {Wolf.}(1995)}]{MandelWolf1995}%
  \BibitemOpen
  \bibfield  {author} {\bibinfo {author} {\bibfnamefont {L.}~\bibnamefont {Mandel}}\ and\ \bibinfo {author} {\bibfnamefont {E.}~\bibnamefont {Wolf.}},\ }\href@noop {} {\emph {\bibinfo {title} {Optical coherence and quantum optics}}}\ (\bibinfo  {publisher} {Cambridge University Press},\ \bibinfo {year} {1995})\BibitemShut {NoStop}%
\bibitem [{\citenamefont {Saleh}\ and\ \citenamefont {Teich}(2007)}]{SalehTeich2007}%
  \BibitemOpen
  \bibfield  {author} {\bibinfo {author} {\bibfnamefont {B.}~\bibnamefont {Saleh}}\ and\ \bibinfo {author} {\bibfnamefont {M.}~\bibnamefont {Teich}},\ }\href {https://books.google.ca/books?id=Ve8eAQAAIAAJ} {\emph {\bibinfo {title} {Fundamentals of Photonics}}},\ Wiley Series in Pure and Applied Optics\ (\bibinfo  {publisher} {Wiley},\ \bibinfo {year} {2007})\BibitemShut {NoStop}%
\bibitem [{\citenamefont {Aaronson}\ and\ \citenamefont {Arkhipov}(2013)}]{AaronsonArkhipov2013}%
  \BibitemOpen
  \bibfield  {author} {\bibinfo {author} {\bibfnamefont {S.}~\bibnamefont {Aaronson}}\ and\ \bibinfo {author} {\bibfnamefont {A.}~\bibnamefont {Arkhipov}},\ }\href {\doibase 10.4086/toc.2013.v009a004} {\bibfield  {journal} {\bibinfo  {journal} {Theory of Computing}\ }\textbf {\bibinfo {volume} {9}},\ \bibinfo {pages} {143} (\bibinfo {year} {2013})}\BibitemShut {NoStop}%
\bibitem [{\citenamefont {Grangier}\ \emph {et~al.}(1986)\citenamefont {Grangier}, \citenamefont {Roger},\ and\ \citenamefont {Aspect}}]{Grangier_1986}%
  \BibitemOpen
  \bibfield  {author} {\bibinfo {author} {\bibfnamefont {P.}~\bibnamefont {Grangier}}, \bibinfo {author} {\bibfnamefont {G.}~\bibnamefont {Roger}}, \ and\ \bibinfo {author} {\bibfnamefont {A.}~\bibnamefont {Aspect}},\ }\href {\doibase 10.1209/0295-5075/1/4/004} {\bibfield  {journal} {\bibinfo  {journal} {Europhysics Letters}\ }\textbf {\bibinfo {volume} {1}},\ \bibinfo {pages} {173} (\bibinfo {year} {1986})}\BibitemShut {NoStop}%
\bibitem [{\citenamefont {Hong}\ \emph {et~al.}(1987)\citenamefont {Hong}, \citenamefont {Ou},\ and\ \citenamefont {Mandel}}]{HongOuMandel1987}%
  \BibitemOpen
  \bibfield  {author} {\bibinfo {author} {\bibfnamefont {C.~K.}\ \bibnamefont {Hong}}, \bibinfo {author} {\bibfnamefont {Z.~Y.}\ \bibnamefont {Ou}}, \ and\ \bibinfo {author} {\bibfnamefont {L.}~\bibnamefont {Mandel}},\ }\href {\doibase 10.1103/PhysRevLett.59.2044} {\bibfield  {journal} {\bibinfo  {journal} {Physical Review Letters}\ }\textbf {\bibinfo {volume} {59}},\ \bibinfo {pages} {2044} (\bibinfo {year} {1987})}\BibitemShut {NoStop}%
\bibitem [{\citenamefont {Born}\ and\ \citenamefont {Wolf}(1999)}]{BornWolf1999}%
  \BibitemOpen
  \bibfield  {author} {\bibinfo {author} {\bibfnamefont {M.}~\bibnamefont {Born}}\ and\ \bibinfo {author} {\bibfnamefont {E.}~\bibnamefont {Wolf}},\ }\href@noop {} {\emph {\bibinfo {title} {Principles of optics: {E}lectromagnetic theory of propagation, interference and diffraction of light.}}},\ \bibinfo {edition} {7th}\ ed.\ (\bibinfo  {publisher} {Cambridge University Press},\ \bibinfo {year} {1999})\BibitemShut {NoStop}%
\bibitem [{\citenamefont {{The L. I. G. O. Scientific Collaboration}}(2011)}]{LIGO2011}%
  \BibitemOpen
  \bibfield  {author} {\bibinfo {author} {\bibnamefont {{The L. I. G. O. Scientific Collaboration}}},\ }\href {http://dx.doi.org/10.1038/nphys2083} {\bibfield  {journal} {\bibinfo  {journal} {Nature Physics}\ }\textbf {\bibinfo {volume} {7}},\ \bibinfo {pages} {962} (\bibinfo {year} {2011})}\BibitemShut {NoStop}%
\bibitem [{\citenamefont {Lucamarini}\ \emph {et~al.}(2018)\citenamefont {Lucamarini}, \citenamefont {Yuan}, \citenamefont {Dynes},\ and\ \citenamefont {Shields}}]{Lucamarinietal2018}%
  \BibitemOpen
  \bibfield  {author} {\bibinfo {author} {\bibfnamefont {M.}~\bibnamefont {Lucamarini}}, \bibinfo {author} {\bibfnamefont {Z.~L.}\ \bibnamefont {Yuan}}, \bibinfo {author} {\bibfnamefont {J.~F.}\ \bibnamefont {Dynes}}, \ and\ \bibinfo {author} {\bibfnamefont {A.~J.}\ \bibnamefont {Shields}},\ }\href {\doibase 10.1038/s41586-018-0066-6} {\bibfield  {journal} {\bibinfo  {journal} {Nature}\ }\textbf {\bibinfo {volume} {557}},\ \bibinfo {pages} {400} (\bibinfo {year} {2018})}\BibitemShut {NoStop}%
\bibitem [{\citenamefont {Flamini}\ \emph {et~al.}(2018)\citenamefont {Flamini}, \citenamefont {Spagnolo},\ and\ \citenamefont {Sciarrino}}]{Flaminietal2019}%
  \BibitemOpen
  \bibfield  {author} {\bibinfo {author} {\bibfnamefont {F.}~\bibnamefont {Flamini}}, \bibinfo {author} {\bibfnamefont {N.}~\bibnamefont {Spagnolo}}, \ and\ \bibinfo {author} {\bibfnamefont {F.}~\bibnamefont {Sciarrino}},\ }\href {\doibase 10.1088/1361-6633/aad5b2} {\bibfield  {journal} {\bibinfo  {journal} {Reports on Progress in Physics}\ }\textbf {\bibinfo {volume} {82}},\ \bibinfo {pages} {016001} (\bibinfo {year} {2018})}\BibitemShut {NoStop}%
\bibitem [{\citenamefont {Grosshans}\ and\ \citenamefont {Grangier}(2002)}]{GrosshansGrangier2002}%
  \BibitemOpen
  \bibfield  {author} {\bibinfo {author} {\bibfnamefont {F.}~\bibnamefont {Grosshans}}\ and\ \bibinfo {author} {\bibfnamefont {P.}~\bibnamefont {Grangier}},\ }\href {\doibase 10.1103/PhysRevLett.88.057902} {\bibfield  {journal} {\bibinfo  {journal} {Physical Review Letters}\ }\textbf {\bibinfo {volume} {88}},\ \bibinfo {pages} {057902} (\bibinfo {year} {2002})}\BibitemShut {NoStop}%
\bibitem [{\citenamefont {Zurek}\ \emph {et~al.}(1993)\citenamefont {Zurek}, \citenamefont {Habib},\ and\ \citenamefont {Paz}}]{Zurek1993}%
  \BibitemOpen
  \bibfield  {author} {\bibinfo {author} {\bibfnamefont {W.~H.}\ \bibnamefont {Zurek}}, \bibinfo {author} {\bibfnamefont {S.}~\bibnamefont {Habib}}, \ and\ \bibinfo {author} {\bibfnamefont {J.~P.}\ \bibnamefont {Paz}},\ }\href {\doibase 10.1103/PhysRevLett.70.1187} {\bibfield  {journal} {\bibinfo  {journal} {Phys. Rev. Lett.}\ }\textbf {\bibinfo {volume} {70}},\ \bibinfo {pages} {1187} (\bibinfo {year} {1993})}\BibitemShut {NoStop}%
\bibitem [{\citenamefont {Dodonov}\ \emph {et~al.}(2000{\natexlab{a}})\citenamefont {Dodonov}, \citenamefont {Mizrahi}, \citenamefont {de~Souza~Silva},\ and\ \citenamefont {Mizrahi}}]{Dodonov_2000}%
  \BibitemOpen
  \bibfield  {author} {\bibinfo {author} {\bibfnamefont {V.~V.}\ \bibnamefont {Dodonov}}, \bibinfo {author} {\bibfnamefont {S.~S.}\ \bibnamefont {Mizrahi}}, \bibinfo {author} {\bibfnamefont {A.~L.}\ \bibnamefont {de~Souza~Silva}}, \ and\ \bibinfo {author} {\bibfnamefont {S.~S.}\ \bibnamefont {Mizrahi}},\ }\href {\doibase 10.1088/1464-4266/2/3/309} {\bibfield  {journal} {\bibinfo  {journal} {Journal of Optics B: Quantum and Semiclassical Optics}\ }\textbf {\bibinfo {volume} {2}},\ \bibinfo {pages} {271} (\bibinfo {year} {2000}{\natexlab{a}})}\BibitemShut {NoStop}%
\bibitem [{\citenamefont {Zurek}(2003)}]{Zurek2003}%
  \BibitemOpen
  \bibfield  {author} {\bibinfo {author} {\bibfnamefont {W.~H.}\ \bibnamefont {Zurek}},\ }\href {\doibase 10.1103/RevModPhys.75.715} {\bibfield  {journal} {\bibinfo  {journal} {Rev. Mod. Phys.}\ }\textbf {\bibinfo {volume} {75}},\ \bibinfo {pages} {715} (\bibinfo {year} {2003})}\BibitemShut {NoStop}%
\bibitem [{\citenamefont {Hertz}\ and\ \citenamefont {De~Bi\`evre}(2020)}]{Hertz}%
  \BibitemOpen
  \bibfield  {author} {\bibinfo {author} {\bibfnamefont {A.}~\bibnamefont {Hertz}}\ and\ \bibinfo {author} {\bibfnamefont {S.}~\bibnamefont {De~Bi\`evre}},\ }\href {\doibase 10.1103/PhysRevLett.124.090402} {\bibfield  {journal} {\bibinfo  {journal} {Phys. Rev. Lett.}\ }\textbf {\bibinfo {volume} {124}},\ \bibinfo {pages} {090402} (\bibinfo {year} {2020})}\BibitemShut {NoStop}%
\bibitem [{\citenamefont {Rosiek}\ \emph {et~al.}(2024)\citenamefont {Rosiek}, \citenamefont {Rossi}, \citenamefont {Schliesser},\ and\ \citenamefont {S\o{}rensen}}]{Rosiek2024}%
  \BibitemOpen
  \bibfield  {author} {\bibinfo {author} {\bibfnamefont {C.~A.}\ \bibnamefont {Rosiek}}, \bibinfo {author} {\bibfnamefont {M.}~\bibnamefont {Rossi}}, \bibinfo {author} {\bibfnamefont {A.}~\bibnamefont {Schliesser}}, \ and\ \bibinfo {author} {\bibfnamefont {A.~S.}\ \bibnamefont {S\o{}rensen}},\ }\href {\doibase 10.1103/PRXQuantum.5.030312} {\bibfield  {journal} {\bibinfo  {journal} {PRX Quantum}\ }\textbf {\bibinfo {volume} {5}},\ \bibinfo {pages} {030312} (\bibinfo {year} {2024})}\BibitemShut {NoStop}%
\bibitem [{\citenamefont {Mari}\ and\ \citenamefont {Eisert}(2012)}]{MariEisert2012}%
  \BibitemOpen
  \bibfield  {author} {\bibinfo {author} {\bibfnamefont {A.}~\bibnamefont {Mari}}\ and\ \bibinfo {author} {\bibfnamefont {J.}~\bibnamefont {Eisert}},\ }\href {\doibase 10.1103/PhysRevLett.109.230503} {\bibfield  {journal} {\bibinfo  {journal} {Phys. Rev. Lett.}\ }\textbf {\bibinfo {volume} {109}},\ \bibinfo {pages} {230503} (\bibinfo {year} {2012})}\BibitemShut {NoStop}%
\bibitem [{\citenamefont {Rahimi-Keshari}\ \emph {et~al.}(2016)\citenamefont {Rahimi-Keshari}, \citenamefont {Ralph},\ and\ \citenamefont {Caves}}]{RahimiKesharietal2016}%
  \BibitemOpen
  \bibfield  {author} {\bibinfo {author} {\bibfnamefont {S.}~\bibnamefont {Rahimi-Keshari}}, \bibinfo {author} {\bibfnamefont {T.~C.}\ \bibnamefont {Ralph}}, \ and\ \bibinfo {author} {\bibfnamefont {C.~M.}\ \bibnamefont {Caves}},\ }\href {\doibase 10.1103/PhysRevX.6.021039} {\bibfield  {journal} {\bibinfo  {journal} {Phys. Rev. X}\ }\textbf {\bibinfo {volume} {6}},\ \bibinfo {pages} {021039} (\bibinfo {year} {2016})}\BibitemShut {NoStop}%
\bibitem [{\citenamefont {Oszmaniec}\ and\ \citenamefont {Brod}(2018)}]{OszmaniecBrod2018}%
  \BibitemOpen
  \bibfield  {author} {\bibinfo {author} {\bibfnamefont {M.}~\bibnamefont {Oszmaniec}}\ and\ \bibinfo {author} {\bibfnamefont {D.~J.}\ \bibnamefont {Brod}},\ }\href {\doibase 10.1088/1367-2630/aadfa8} {\bibfield  {journal} {\bibinfo  {journal} {New Journal of Physics}\ }\textbf {\bibinfo {volume} {20}},\ \bibinfo {pages} {092002} (\bibinfo {year} {2018})}\BibitemShut {NoStop}%
\bibitem [{\citenamefont {Qi}\ \emph {et~al.}(2020)\citenamefont {Qi}, \citenamefont {Brod}, \citenamefont {Quesada},\ and\ \citenamefont {Garc\'{\i}a-Patr\'on}}]{Qietal2020}%
  \BibitemOpen
  \bibfield  {author} {\bibinfo {author} {\bibfnamefont {H.}~\bibnamefont {Qi}}, \bibinfo {author} {\bibfnamefont {D.~J.}\ \bibnamefont {Brod}}, \bibinfo {author} {\bibfnamefont {N.}~\bibnamefont {Quesada}}, \ and\ \bibinfo {author} {\bibfnamefont {R.}~\bibnamefont {Garc\'{\i}a-Patr\'on}},\ }\href {\doibase 10.1103/PhysRevLett.124.100502} {\bibfield  {journal} {\bibinfo  {journal} {Phys. Rev. Lett.}\ }\textbf {\bibinfo {volume} {124}},\ \bibinfo {pages} {100502} (\bibinfo {year} {2020})}\BibitemShut {NoStop}%
\bibitem [{\citenamefont {Goldberg}\ \emph {et~al.}(2023)\citenamefont {Goldberg}, \citenamefont {Thekkadath},\ and\ \citenamefont {Heshami}}]{Goldbergetal2023}%
  \BibitemOpen
  \bibfield  {author} {\bibinfo {author} {\bibfnamefont {A.~Z.}\ \bibnamefont {Goldberg}}, \bibinfo {author} {\bibfnamefont {G.~S.}\ \bibnamefont {Thekkadath}}, \ and\ \bibinfo {author} {\bibfnamefont {K.}~\bibnamefont {Heshami}},\ }\href {\doibase 10.1103/PhysRevA.107.042610} {\bibfield  {journal} {\bibinfo  {journal} {Phys. Rev. A}\ }\textbf {\bibinfo {volume} {107}},\ \bibinfo {pages} {042610} (\bibinfo {year} {2023})}\BibitemShut {NoStop}%
\bibitem [{\citenamefont {Hertz}\ \emph {et~al.}(2024)\citenamefont {Hertz}, \citenamefont {Goldberg},\ and\ \citenamefont {Heshami}}]{HGH}%
  \BibitemOpen
  \bibfield  {author} {\bibinfo {author} {\bibfnamefont {A.}~\bibnamefont {Hertz}}, \bibinfo {author} {\bibfnamefont {A.~Z.}\ \bibnamefont {Goldberg}}, \ and\ \bibinfo {author} {\bibfnamefont {K.}~\bibnamefont {Heshami}},\ }\href {\doibase 10.1103/PhysRevA.110.012408} {\bibfield  {journal} {\bibinfo  {journal} {Phys. Rev. A}\ }\textbf {\bibinfo {volume} {110}},\ \bibinfo {pages} {012408} (\bibinfo {year} {2024})}\BibitemShut {NoStop}%
\bibitem [{\citenamefont {Sperling}\ and\ \citenamefont {Vogel}(2009)}]{SperlingVogel2009negativequasi}%
  \BibitemOpen
  \bibfield  {author} {\bibinfo {author} {\bibfnamefont {J.}~\bibnamefont {Sperling}}\ and\ \bibinfo {author} {\bibfnamefont {W.}~\bibnamefont {Vogel}},\ }\href {\doibase 10.1103/PhysRevA.79.042337} {\bibfield  {journal} {\bibinfo  {journal} {Phys. Rev. A}\ }\textbf {\bibinfo {volume} {79}},\ \bibinfo {pages} {042337} (\bibinfo {year} {2009})}\BibitemShut {NoStop}%
\bibitem [{\citenamefont {Bohmann}\ \emph {et~al.}(2020)\citenamefont {Bohmann}, \citenamefont {Agudelo},\ and\ \citenamefont {Sperling}}]{Bohmann}%
  \BibitemOpen
  \bibfield  {author} {\bibinfo {author} {\bibfnamefont {M.}~\bibnamefont {Bohmann}}, \bibinfo {author} {\bibfnamefont {E.}~\bibnamefont {Agudelo}}, \ and\ \bibinfo {author} {\bibfnamefont {J.}~\bibnamefont {Sperling}},\ }\href {\doibase 10.22331/q-2020-10-15-343} {\bibfield  {journal} {\bibinfo  {journal} {{Quantum}}\ }\textbf {\bibinfo {volume} {4}},\ \bibinfo {pages} {343} (\bibinfo {year} {2020})}\BibitemShut {NoStop}%
\bibitem [{\citenamefont {Steuernagel}\ and\ \citenamefont {Lee}(2023)}]{SteuernagelLee2023arxiv}%
  \BibitemOpen
  \bibfield  {author} {\bibinfo {author} {\bibfnamefont {O.}~\bibnamefont {Steuernagel}}\ and\ \bibinfo {author} {\bibfnamefont {R.-K.}\ \bibnamefont {Lee}},\ }\href {https://arxiv.org/abs/2311.17399} {\  (\bibinfo {year} {2023})},\ \Eprint {http://arxiv.org/abs/2311.17399} {arXiv:2311.17399 [quant-ph]} \BibitemShut {NoStop}%
\bibitem [{\citenamefont {Chabaud}\ \emph {et~al.}(2024)\citenamefont {Chabaud}, \citenamefont {Ghobadi}, \citenamefont {Beigi},\ and\ \citenamefont {Rahimi-Keshari}}]{Chabaudetal2024}%
  \BibitemOpen
  \bibfield  {author} {\bibinfo {author} {\bibfnamefont {U.}~\bibnamefont {Chabaud}}, \bibinfo {author} {\bibfnamefont {R.}~\bibnamefont {Ghobadi}}, \bibinfo {author} {\bibfnamefont {S.}~\bibnamefont {Beigi}}, \ and\ \bibinfo {author} {\bibfnamefont {S.}~\bibnamefont {Rahimi-Keshari}},\ }\href {\doibase 10.22331/q-2024-11-07-1519} {\bibfield  {journal} {\bibinfo  {journal} {{Quantum}}\ }\textbf {\bibinfo {volume} {8}},\ \bibinfo {pages} {1519} (\bibinfo {year} {2024})}\BibitemShut {NoStop}%
\bibitem [{\citenamefont {Sudarshan}(1963)}]{Sudarshan1963}%
  \BibitemOpen
  \bibfield  {author} {\bibinfo {author} {\bibfnamefont {E.~C.~G.}\ \bibnamefont {Sudarshan}},\ }\href {\doibase 10.1103/PhysRevLett.10.277} {\bibfield  {journal} {\bibinfo  {journal} {Physical Review Letters}\ }\textbf {\bibinfo {volume} {10}},\ \bibinfo {pages} {277} (\bibinfo {year} {1963})}\BibitemShut {NoStop}%
\bibitem [{\citenamefont {Englert}(2024)}]{Englert2024}%
  \BibitemOpen
  \bibfield  {author} {\bibinfo {author} {\bibfnamefont {B.-G.}\ \bibnamefont {Englert}},\ }\href {\doibase https://doi.org/10.1016/j.physleta.2023.129278} {\bibfield  {journal} {\bibinfo  {journal} {Physics Letters A}\ }\textbf {\bibinfo {volume} {494}},\ \bibinfo {pages} {129278} (\bibinfo {year} {2024})}\BibitemShut {NoStop}%
\bibitem [{\citenamefont {Weedbrook}\ \emph {et~al.}(2012)\citenamefont {Weedbrook}, \citenamefont {Pirandola}, \citenamefont {Garc\'{\i}a-Patr\'on}, \citenamefont {Cerf}, \citenamefont {Ralph}, \citenamefont {Shapiro},\ and\ \citenamefont {Lloyd}}]{weedbrook}%
  \BibitemOpen
  \bibfield  {author} {\bibinfo {author} {\bibfnamefont {C.}~\bibnamefont {Weedbrook}}, \bibinfo {author} {\bibfnamefont {S.}~\bibnamefont {Pirandola}}, \bibinfo {author} {\bibfnamefont {R.}~\bibnamefont {Garc\'{\i}a-Patr\'on}}, \bibinfo {author} {\bibfnamefont {N.~J.}\ \bibnamefont {Cerf}}, \bibinfo {author} {\bibfnamefont {T.~C.}\ \bibnamefont {Ralph}}, \bibinfo {author} {\bibfnamefont {J.~H.}\ \bibnamefont {Shapiro}}, \ and\ \bibinfo {author} {\bibfnamefont {S.}~\bibnamefont {Lloyd}},\ }\href {\doibase 10.1103/RevModPhys.84.621} {\bibfield  {journal} {\bibinfo  {journal} {Rev. Mod. Phys.}\ }\textbf {\bibinfo {volume} {84}},\ \bibinfo {pages} {621} (\bibinfo {year} {2012})}\BibitemShut {NoStop}%
\bibitem [{\citenamefont {Nielsen}\ and\ \citenamefont {Chuang}(2000)}]{NielsenChuang2000}%
  \BibitemOpen
  \bibfield  {author} {\bibinfo {author} {\bibfnamefont {M.~A.}\ \bibnamefont {Nielsen}}\ and\ \bibinfo {author} {\bibfnamefont {I.~L.}\ \bibnamefont {Chuang}},\ }\href@noop {} {\emph {\bibinfo {title} {{Quantum Computation and Quantum Information}}}}\ (\bibinfo  {publisher} {Cambridge University Press},\ \bibinfo {address} {Cambridge},\ \bibinfo {year} {2000})\BibitemShut {NoStop}%
\bibitem [{\citenamefont {Goldberg}(2024)}]{Goldberg2024}%
  \BibitemOpen
  \bibfield  {author} {\bibinfo {author} {\bibfnamefont {A.~Z.}\ \bibnamefont {Goldberg}},\ }\href {\doibase 10.1364/OPTICAQ.501218} {\bibfield  {journal} {\bibinfo  {journal} {Optica Quantum}\ }\textbf {\bibinfo {volume} {2}},\ \bibinfo {pages} {14} (\bibinfo {year} {2024})}\BibitemShut {NoStop}%
\bibitem [{\citenamefont {Cahill}\ and\ \citenamefont {Glauber}(1969)}]{CahillGlauber}%
  \BibitemOpen
  \bibfield  {author} {\bibinfo {author} {\bibfnamefont {K.~E.}\ \bibnamefont {Cahill}}\ and\ \bibinfo {author} {\bibfnamefont {R.~J.}\ \bibnamefont {Glauber}},\ }\href {\doibase 10.1103/PhysRev.177.1882} {\bibfield  {journal} {\bibinfo  {journal} {Phys. Rev.}\ }\textbf {\bibinfo {volume} {177}},\ \bibinfo {pages} {1882} (\bibinfo {year} {1969})}\BibitemShut {NoStop}%
\bibitem [{\citenamefont {Filip}(2013)}]{RadimFilip}%
  \BibitemOpen
  \bibfield  {author} {\bibinfo {author} {\bibfnamefont {R.}~\bibnamefont {Filip}},\ }\href {\doibase 10.1103/PhysRevA.87.042308} {\bibfield  {journal} {\bibinfo  {journal} {Phys. Rev. A}\ }\textbf {\bibinfo {volume} {87}},\ \bibinfo {pages} {042308} (\bibinfo {year} {2013})}\BibitemShut {NoStop}%
\bibitem [{\citenamefont {Semenov}\ \emph {et~al.}(2006)\citenamefont {Semenov}, \citenamefont {Vasylyev},\ and\ \citenamefont {Lev}}]{se06}%
  \BibitemOpen
  \bibfield  {author} {\bibinfo {author} {\bibfnamefont {A.~A.}\ \bibnamefont {Semenov}}, \bibinfo {author} {\bibfnamefont {D.}~\bibnamefont {Vasylyev}}, \ and\ \bibinfo {author} {\bibfnamefont {B.~I.}\ \bibnamefont {Lev}},\ }\href@noop {} {\bibfield  {journal} {\bibinfo  {journal} {J. Phys. B: At. Mol. Opt. Phys.}\ }\textbf {\bibinfo {volume} {39}},\ \bibinfo {pages} {905} (\bibinfo {year} {2006})}\BibitemShut {NoStop}%
\bibitem [{\citenamefont {Kim}\ \emph {et~al.}(2002)\citenamefont {Kim}, \citenamefont {Son}, \citenamefont {Bu\ifmmode~\check{z}\else \v{z}\fi{}ek},\ and\ \citenamefont {Knight}}]{Kimetal2002}%
  \BibitemOpen
  \bibfield  {author} {\bibinfo {author} {\bibfnamefont {M.~S.}\ \bibnamefont {Kim}}, \bibinfo {author} {\bibfnamefont {W.}~\bibnamefont {Son}}, \bibinfo {author} {\bibfnamefont {V.}~\bibnamefont {Bu\ifmmode~\check{z}\else \v{z}\fi{}ek}}, \ and\ \bibinfo {author} {\bibfnamefont {P.~L.}\ \bibnamefont {Knight}},\ }\href {\doibase 10.1103/PhysRevA.65.032323} {\bibfield  {journal} {\bibinfo  {journal} {Physical Review A}\ }\textbf {\bibinfo {volume} {65}},\ \bibinfo {pages} {032323} (\bibinfo {year} {2002})}\BibitemShut {NoStop}%
\bibitem [{\citenamefont {Xiang-bin}(2002)}]{Xiangbin2002}%
  \BibitemOpen
  \bibfield  {author} {\bibinfo {author} {\bibfnamefont {W.}~\bibnamefont {Xiang-bin}},\ }\href {\doibase 10.1103/PhysRevA.66.024303} {\bibfield  {journal} {\bibinfo  {journal} {Physical Review A}\ }\textbf {\bibinfo {volume} {66}},\ \bibinfo {pages} {024303} (\bibinfo {year} {2002})}\BibitemShut {NoStop}%
\bibitem [{\citenamefont {Tan}\ \emph {et~al.}(1991)\citenamefont {Tan}, \citenamefont {Walls},\ and\ \citenamefont {Collett}}]{Tanetal1991}%
  \BibitemOpen
  \bibfield  {author} {\bibinfo {author} {\bibfnamefont {S.~M.}\ \bibnamefont {Tan}}, \bibinfo {author} {\bibfnamefont {D.~F.}\ \bibnamefont {Walls}}, \ and\ \bibinfo {author} {\bibfnamefont {M.~J.}\ \bibnamefont {Collett}},\ }\href {\doibase 10.1103/PhysRevLett.66.252} {\bibfield  {journal} {\bibinfo  {journal} {Physical Review Letters}\ }\textbf {\bibinfo {volume} {66}},\ \bibinfo {pages} {252} (\bibinfo {year} {1991})}\BibitemShut {NoStop}%
\bibitem [{\citenamefont {Sanders}(1992)}]{Sanders1992}%
  \BibitemOpen
  \bibfield  {author} {\bibinfo {author} {\bibfnamefont {B.~C.}\ \bibnamefont {Sanders}},\ }\href {\doibase 10.1103/PhysRevA.45.6811} {\bibfield  {journal} {\bibinfo  {journal} {Phys. Rev. A}\ }\textbf {\bibinfo {volume} {45}},\ \bibinfo {pages} {6811} (\bibinfo {year} {1992})}\BibitemShut {NoStop}%
\bibitem [{\citenamefont {Huang}\ and\ \citenamefont {Agarwal}(1994)}]{HuangAgarwal1994}%
  \BibitemOpen
  \bibfield  {author} {\bibinfo {author} {\bibfnamefont {H.}~\bibnamefont {Huang}}\ and\ \bibinfo {author} {\bibfnamefont {G.~S.}\ \bibnamefont {Agarwal}},\ }\href {\doibase 10.1103/PhysRevA.49.52} {\bibfield  {journal} {\bibinfo  {journal} {Physical Review A}\ }\textbf {\bibinfo {volume} {49}},\ \bibinfo {pages} {52} (\bibinfo {year} {1994})}\BibitemShut {NoStop}%
\bibitem [{\citenamefont {Arvind}\ and\ \citenamefont {Mukunda}(1999)}]{ArvindMukunda1999}%
  \BibitemOpen
  \bibfield  {author} {\bibinfo {author} {\bibnamefont {Arvind}}\ and\ \bibinfo {author} {\bibfnamefont {N.}~\bibnamefont {Mukunda}},\ }\href {\doibase https://doi.org/10.1016/S0375-9601(99)00471-5} {\bibfield  {journal} {\bibinfo  {journal} {Physics Letters A}\ }\textbf {\bibinfo {volume} {259}},\ \bibinfo {pages} {421} (\bibinfo {year} {1999})}\BibitemShut {NoStop}%
\bibitem [{\citenamefont {Paris}(1999)}]{Paris1999}%
  \BibitemOpen
  \bibfield  {author} {\bibinfo {author} {\bibfnamefont {M.~G.~A.}\ \bibnamefont {Paris}},\ }\href {\doibase 10.1103/PhysRevA.59.1615} {\bibfield  {journal} {\bibinfo  {journal} {Physical Review A}\ }\textbf {\bibinfo {volume} {59}},\ \bibinfo {pages} {1615} (\bibinfo {year} {1999})}\BibitemShut {NoStop}%
\bibitem [{\citenamefont {Wolf}\ \emph {et~al.}(2003)\citenamefont {Wolf}, \citenamefont {Eisert},\ and\ \citenamefont {Plenio}}]{Wolfetal2003}%
  \BibitemOpen
  \bibfield  {author} {\bibinfo {author} {\bibfnamefont {M.~M.}\ \bibnamefont {Wolf}}, \bibinfo {author} {\bibfnamefont {J.}~\bibnamefont {Eisert}}, \ and\ \bibinfo {author} {\bibfnamefont {M.~B.}\ \bibnamefont {Plenio}},\ }\href {\doibase 10.1103/PhysRevLett.90.047904} {\bibfield  {journal} {\bibinfo  {journal} {Physical Review Letters}\ }\textbf {\bibinfo {volume} {90}},\ \bibinfo {pages} {047904} (\bibinfo {year} {2003})}\BibitemShut {NoStop}%
\bibitem [{\citenamefont {Ivan}\ \emph {et~al.}(2006)\citenamefont {Ivan}, \citenamefont {Mukunda},\ and\ \citenamefont {Simon}}]{Ivanetal2006arxiv}%
  \BibitemOpen
  \bibfield  {author} {\bibinfo {author} {\bibfnamefont {J.~S.}\ \bibnamefont {Ivan}}, \bibinfo {author} {\bibfnamefont {N.}~\bibnamefont {Mukunda}}, \ and\ \bibinfo {author} {\bibfnamefont {R.}~\bibnamefont {Simon}},\ }\href {https://arxiv.org/abs/quant-ph/0603255} {\enquote {\bibinfo {title} {Generation of npt entanglement from nonclassical photon statistics},}\ } (\bibinfo {year} {2006}),\ \Eprint {http://arxiv.org/abs/quant-ph/0603255} {arXiv:quant-ph/0603255 [quant-ph]} \BibitemShut {NoStop}%
\bibitem [{\citenamefont {A.~R. Usha~Devi}\ and\ \citenamefont {Uma.}(2006)}]{UshaDevietal2006}%
  \BibitemOpen
  \bibfield  {author} {\bibinfo {author} {\bibfnamefont {R.~P.}\ \bibnamefont {A.~R. Usha~Devi}}\ and\ \bibinfo {author} {\bibfnamefont {M.}~\bibnamefont {Uma.}},\ }\href {\doibase 10.1140/epjd/e2006-00135-x} {\bibfield  {journal} {\bibinfo  {journal} {Eur. Phys. J. D}\ }\textbf {\bibinfo {volume} {40}},\ \bibinfo {pages} {133–138} (\bibinfo {year} {2006})}\BibitemShut {NoStop}%
\bibitem [{\citenamefont {Tahira}\ \emph {et~al.}(2009)\citenamefont {Tahira}, \citenamefont {Ikram}, \citenamefont {Nha},\ and\ \citenamefont {Zubairy}}]{Tahiraetal2009}%
  \BibitemOpen
  \bibfield  {author} {\bibinfo {author} {\bibfnamefont {R.}~\bibnamefont {Tahira}}, \bibinfo {author} {\bibfnamefont {M.}~\bibnamefont {Ikram}}, \bibinfo {author} {\bibfnamefont {H.}~\bibnamefont {Nha}}, \ and\ \bibinfo {author} {\bibfnamefont {M.~S.}\ \bibnamefont {Zubairy}},\ }\href {\doibase 10.1103/PhysRevA.79.023816} {\bibfield  {journal} {\bibinfo  {journal} {Physical Review A}\ }\textbf {\bibinfo {volume} {79}},\ \bibinfo {pages} {023816} (\bibinfo {year} {2009})}\BibitemShut {NoStop}%
\bibitem [{\citenamefont {Springer}\ \emph {et~al.}(2009)\citenamefont {Springer}, \citenamefont {Lee}, \citenamefont {Bellini},\ and\ \citenamefont {Kim}}]{Springeretal2009}%
  \BibitemOpen
  \bibfield  {author} {\bibinfo {author} {\bibfnamefont {S.~C.}\ \bibnamefont {Springer}}, \bibinfo {author} {\bibfnamefont {J.}~\bibnamefont {Lee}}, \bibinfo {author} {\bibfnamefont {M.}~\bibnamefont {Bellini}}, \ and\ \bibinfo {author} {\bibfnamefont {M.~S.}\ \bibnamefont {Kim}},\ }\href {\doibase 10.1103/PhysRevA.79.062303} {\bibfield  {journal} {\bibinfo  {journal} {Phys. Rev. A}\ }\textbf {\bibinfo {volume} {79}},\ \bibinfo {pages} {062303} (\bibinfo {year} {2009})}\BibitemShut {NoStop}%
\bibitem [{\citenamefont {Li}\ \emph {et~al.}(2010)\citenamefont {Li}, \citenamefont {Li}, \citenamefont {Wang}, \citenamefont {Zhu},\ and\ \citenamefont {Zhang}}]{Lietal2010}%
  \BibitemOpen
  \bibfield  {author} {\bibinfo {author} {\bibfnamefont {J.}~\bibnamefont {Li}}, \bibinfo {author} {\bibfnamefont {G.}~\bibnamefont {Li}}, \bibinfo {author} {\bibfnamefont {J.-M.}\ \bibnamefont {Wang}}, \bibinfo {author} {\bibfnamefont {S.-Y.}\ \bibnamefont {Zhu}}, \ and\ \bibinfo {author} {\bibfnamefont {T.-C.}\ \bibnamefont {Zhang}},\ }\href {\doibase 10.1088/0953-4075/43/8/085504} {\bibfield  {journal} {\bibinfo  {journal} {Journal of Physics B: Atomic, Molecular and Optical Physics}\ }\textbf {\bibinfo {volume} {43}},\ \bibinfo {pages} {085504} (\bibinfo {year} {2010})}\BibitemShut {NoStop}%
\bibitem [{\citenamefont {Ivan}\ \emph {et~al.}(2011)\citenamefont {Ivan}, \citenamefont {Chaturvedi}, \citenamefont {Ercolessi}, \citenamefont {Marmo}, \citenamefont {Morandi}, \citenamefont {Mukunda},\ and\ \citenamefont {Simon}}]{Ivanetal2011}%
  \BibitemOpen
  \bibfield  {author} {\bibinfo {author} {\bibfnamefont {J.~S.}\ \bibnamefont {Ivan}}, \bibinfo {author} {\bibfnamefont {S.}~\bibnamefont {Chaturvedi}}, \bibinfo {author} {\bibfnamefont {E.}~\bibnamefont {Ercolessi}}, \bibinfo {author} {\bibfnamefont {G.}~\bibnamefont {Marmo}}, \bibinfo {author} {\bibfnamefont {G.}~\bibnamefont {Morandi}}, \bibinfo {author} {\bibfnamefont {N.}~\bibnamefont {Mukunda}}, \ and\ \bibinfo {author} {\bibfnamefont {R.}~\bibnamefont {Simon}},\ }\href {\doibase 10.1103/PhysRevA.83.032118} {\bibfield  {journal} {\bibinfo  {journal} {Phys. Rev. A}\ }\textbf {\bibinfo {volume} {83}},\ \bibinfo {pages} {032118} (\bibinfo {year} {2011})}\BibitemShut {NoStop}%
\bibitem [{\citenamefont {Piani}\ \emph {et~al.}(2011)\citenamefont {Piani}, \citenamefont {Gharibian}, \citenamefont {Adesso}, \citenamefont {Calsamiglia}, \citenamefont {Horodecki},\ and\ \citenamefont {Winter}}]{Pianietal2011}%
  \BibitemOpen
  \bibfield  {author} {\bibinfo {author} {\bibfnamefont {M.}~\bibnamefont {Piani}}, \bibinfo {author} {\bibfnamefont {S.}~\bibnamefont {Gharibian}}, \bibinfo {author} {\bibfnamefont {G.}~\bibnamefont {Adesso}}, \bibinfo {author} {\bibfnamefont {J.}~\bibnamefont {Calsamiglia}}, \bibinfo {author} {\bibfnamefont {P.}~\bibnamefont {Horodecki}}, \ and\ \bibinfo {author} {\bibfnamefont {A.}~\bibnamefont {Winter}},\ }\href {\doibase 10.1103/PhysRevLett.106.220403} {\bibfield  {journal} {\bibinfo  {journal} {Physical Review Letters}\ }\textbf {\bibinfo {volume} {106}},\ \bibinfo {pages} {220403} (\bibinfo {year} {2011})}\BibitemShut {NoStop}%
\bibitem [{\citenamefont {Ivan}\ \emph {et~al.}(2012)\citenamefont {Ivan}, \citenamefont {Mukunda},\ and\ \citenamefont {Simon}}]{Ivanetal2012generation}%
  \BibitemOpen
  \bibfield  {author} {\bibinfo {author} {\bibfnamefont {J.~S.}\ \bibnamefont {Ivan}}, \bibinfo {author} {\bibfnamefont {N.}~\bibnamefont {Mukunda}}, \ and\ \bibinfo {author} {\bibfnamefont {R.}~\bibnamefont {Simon}},\ }\href {\doibase 10.1007/s11128-011-0316-0} {\bibfield  {journal} {\bibinfo  {journal} {Quantum Information Processing}\ }\textbf {\bibinfo {volume} {11}},\ \bibinfo {pages} {873} (\bibinfo {year} {2012})}\BibitemShut {NoStop}%
\bibitem [{\citenamefont {Abdel-Khalek}\ \emph {et~al.}(2012)\citenamefont {Abdel-Khalek}, \citenamefont {Berrada},\ and\ \citenamefont {Raymond~Ooi}}]{AbdelKhaleketal2012}%
  \BibitemOpen
  \bibfield  {author} {\bibinfo {author} {\bibfnamefont {S.}~\bibnamefont {Abdel-Khalek}}, \bibinfo {author} {\bibfnamefont {K.}~\bibnamefont {Berrada}}, \ and\ \bibinfo {author} {\bibfnamefont {C.~H.}\ \bibnamefont {Raymond~Ooi}},\ }\href {\doibase 10.1134/S1054660X12090010} {\bibfield  {journal} {\bibinfo  {journal} {Laser Physics}\ }\textbf {\bibinfo {volume} {22}},\ \bibinfo {pages} {1449} (\bibinfo {year} {2012})}\BibitemShut {NoStop}%
\bibitem [{\citenamefont {Jiang}\ \emph {et~al.}(2013)\citenamefont {Jiang}, \citenamefont {Lang},\ and\ \citenamefont {Caves}}]{Jiangetal2013}%
  \BibitemOpen
  \bibfield  {author} {\bibinfo {author} {\bibfnamefont {Z.}~\bibnamefont {Jiang}}, \bibinfo {author} {\bibfnamefont {M.~D.}\ \bibnamefont {Lang}}, \ and\ \bibinfo {author} {\bibfnamefont {C.~M.}\ \bibnamefont {Caves}},\ }\href {\doibase 10.1103/PhysRevA.88.044301} {\bibfield  {journal} {\bibinfo  {journal} {Physical Review A}\ }\textbf {\bibinfo {volume} {88}},\ \bibinfo {pages} {044301} (\bibinfo {year} {2013})}\BibitemShut {NoStop}%
\bibitem [{\citenamefont {Berrada}\ \emph {et~al.}(2013)\citenamefont {Berrada}, \citenamefont {Abdel-Khalek}, \citenamefont {Eleuch},\ and\ \citenamefont {Hassouni}}]{Berradaetal2013}%
  \BibitemOpen
  \bibfield  {author} {\bibinfo {author} {\bibfnamefont {K.}~\bibnamefont {Berrada}}, \bibinfo {author} {\bibfnamefont {S.}~\bibnamefont {Abdel-Khalek}}, \bibinfo {author} {\bibfnamefont {H.}~\bibnamefont {Eleuch}}, \ and\ \bibinfo {author} {\bibfnamefont {Y.}~\bibnamefont {Hassouni}},\ }\href {\doibase 10.1007/s11128-011-0344-9} {\bibfield  {journal} {\bibinfo  {journal} {Quantum Information Processing}\ }\textbf {\bibinfo {volume} {12}},\ \bibinfo {pages} {69} (\bibinfo {year} {2013})}\BibitemShut {NoStop}%
\bibitem [{\citenamefont {Killoran}\ \emph {et~al.}(2014)\citenamefont {Killoran}, \citenamefont {Cramer},\ and\ \citenamefont {Plenio}}]{Killoranetal2014}%
  \BibitemOpen
  \bibfield  {author} {\bibinfo {author} {\bibfnamefont {N.}~\bibnamefont {Killoran}}, \bibinfo {author} {\bibfnamefont {M.}~\bibnamefont {Cramer}}, \ and\ \bibinfo {author} {\bibfnamefont {M.~B.}\ \bibnamefont {Plenio}},\ }\href {\doibase 10.1103/PhysRevLett.112.150501} {\bibfield  {journal} {\bibinfo  {journal} {Physical Review Letters}\ }\textbf {\bibinfo {volume} {112}},\ \bibinfo {pages} {150501} (\bibinfo {year} {2014})}\BibitemShut {NoStop}%
\bibitem [{\citenamefont {Vogel}\ and\ \citenamefont {Sperling}(2014)}]{VogelSperling2014}%
  \BibitemOpen
  \bibfield  {author} {\bibinfo {author} {\bibfnamefont {W.}~\bibnamefont {Vogel}}\ and\ \bibinfo {author} {\bibfnamefont {J.}~\bibnamefont {Sperling}},\ }\href {\doibase 10.1103/PhysRevA.89.052302} {\bibfield  {journal} {\bibinfo  {journal} {Physical Review A}\ }\textbf {\bibinfo {volume} {89}},\ \bibinfo {pages} {052302} (\bibinfo {year} {2014})}\BibitemShut {NoStop}%
\bibitem [{\citenamefont {Monteiro}\ \emph {et~al.}(2015)\citenamefont {Monteiro}, \citenamefont {Vivoli}, \citenamefont {Guerreiro}, \citenamefont {Martin}, \citenamefont {Bancal}, \citenamefont {Zbinden}, \citenamefont {Thew},\ and\ \citenamefont {Sangouard}}]{Monteiroetal2015}%
  \BibitemOpen
  \bibfield  {author} {\bibinfo {author} {\bibfnamefont {F.}~\bibnamefont {Monteiro}}, \bibinfo {author} {\bibfnamefont {V.~C.}\ \bibnamefont {Vivoli}}, \bibinfo {author} {\bibfnamefont {T.}~\bibnamefont {Guerreiro}}, \bibinfo {author} {\bibfnamefont {A.}~\bibnamefont {Martin}}, \bibinfo {author} {\bibfnamefont {J.-D.}\ \bibnamefont {Bancal}}, \bibinfo {author} {\bibfnamefont {H.}~\bibnamefont {Zbinden}}, \bibinfo {author} {\bibfnamefont {R.~T.}\ \bibnamefont {Thew}}, \ and\ \bibinfo {author} {\bibfnamefont {N.}~\bibnamefont {Sangouard}},\ }\href {\doibase 10.1103/PhysRevLett.114.170504} {\bibfield  {journal} {\bibinfo  {journal} {Phys. Rev. Lett.}\ }\textbf {\bibinfo {volume} {114}},\ \bibinfo {pages} {170504} (\bibinfo {year} {2015})}\BibitemShut {NoStop}%
\bibitem [{\citenamefont {Ge}\ \emph {et~al.}(2015)\citenamefont {Ge}, \citenamefont {Tasgin},\ and\ \citenamefont {Zubairy}}]{Geetal2015}%
  \BibitemOpen
  \bibfield  {author} {\bibinfo {author} {\bibfnamefont {W.}~\bibnamefont {Ge}}, \bibinfo {author} {\bibfnamefont {M.~E.}\ \bibnamefont {Tasgin}}, \ and\ \bibinfo {author} {\bibfnamefont {M.~S.}\ \bibnamefont {Zubairy}},\ }\href {\doibase 10.1103/PhysRevA.92.052328} {\bibfield  {journal} {\bibinfo  {journal} {Physical Review A}\ }\textbf {\bibinfo {volume} {92}},\ \bibinfo {pages} {052328} (\bibinfo {year} {2015})}\BibitemShut {NoStop}%
\bibitem [{\citenamefont {Miranowicz}\ \emph {et~al.}(2015)\citenamefont {Miranowicz}, \citenamefont {Bartkiewicz}, \citenamefont {Lambert}, \citenamefont {Chen},\ and\ \citenamefont {Nori}}]{Miranowiczetal2015}%
  \BibitemOpen
  \bibfield  {author} {\bibinfo {author} {\bibfnamefont {A.}~\bibnamefont {Miranowicz}}, \bibinfo {author} {\bibfnamefont {K.}~\bibnamefont {Bartkiewicz}}, \bibinfo {author} {\bibfnamefont {N.}~\bibnamefont {Lambert}}, \bibinfo {author} {\bibfnamefont {Y.-N.}\ \bibnamefont {Chen}}, \ and\ \bibinfo {author} {\bibfnamefont {F.}~\bibnamefont {Nori}},\ }\href {\doibase 10.1103/PhysRevA.92.062314} {\bibfield  {journal} {\bibinfo  {journal} {Phys. Rev. A}\ }\textbf {\bibinfo {volume} {92}},\ \bibinfo {pages} {062314} (\bibinfo {year} {2015})}\BibitemShut {NoStop}%
\bibitem [{\citenamefont {Streltsov}\ \emph {et~al.}(2015)\citenamefont {Streltsov}, \citenamefont {Singh}, \citenamefont {Dhar}, \citenamefont {Bera},\ and\ \citenamefont {Adesso}}]{Streltsovetal2015}%
  \BibitemOpen
  \bibfield  {author} {\bibinfo {author} {\bibfnamefont {A.}~\bibnamefont {Streltsov}}, \bibinfo {author} {\bibfnamefont {U.}~\bibnamefont {Singh}}, \bibinfo {author} {\bibfnamefont {H.~S.}\ \bibnamefont {Dhar}}, \bibinfo {author} {\bibfnamefont {M.~N.}\ \bibnamefont {Bera}}, \ and\ \bibinfo {author} {\bibfnamefont {G.}~\bibnamefont {Adesso}},\ }\href {\doibase 10.1103/PhysRevLett.115.020403} {\bibfield  {journal} {\bibinfo  {journal} {Phys. Rev. Lett.}\ }\textbf {\bibinfo {volume} {115}},\ \bibinfo {pages} {020403} (\bibinfo {year} {2015})}\BibitemShut {NoStop}%
\bibitem [{\citenamefont {Brunelli}\ \emph {et~al.}(2015)\citenamefont {Brunelli}, \citenamefont {Benedetti}, \citenamefont {Olivares}, \citenamefont {Ferraro},\ and\ \citenamefont {Paris}}]{Brunellietal2015}%
  \BibitemOpen
  \bibfield  {author} {\bibinfo {author} {\bibfnamefont {M.}~\bibnamefont {Brunelli}}, \bibinfo {author} {\bibfnamefont {C.}~\bibnamefont {Benedetti}}, \bibinfo {author} {\bibfnamefont {S.}~\bibnamefont {Olivares}}, \bibinfo {author} {\bibfnamefont {A.}~\bibnamefont {Ferraro}}, \ and\ \bibinfo {author} {\bibfnamefont {M.~G.~A.}\ \bibnamefont {Paris}},\ }\href {\doibase 10.1103/PhysRevA.91.062315} {\bibfield  {journal} {\bibinfo  {journal} {Phys. Rev. A}\ }\textbf {\bibinfo {volume} {91}},\ \bibinfo {pages} {062315} (\bibinfo {year} {2015})}\BibitemShut {NoStop}%
\bibitem [{\citenamefont {Arkhipov}\ \emph {et~al.}(2016)\citenamefont {Arkhipov}, \citenamefont {Pe\ifmmode~\check{r}\else \v{r}\fi{}ina}, \citenamefont {Pe\ifmmode~\check{r}\else \v{r}\fi{}ina},\ and\ \citenamefont {Miranowicz}}]{Arkhipovetal2016}%
  \BibitemOpen
  \bibfield  {author} {\bibinfo {author} {\bibfnamefont {I.~I.}\ \bibnamefont {Arkhipov}}, \bibinfo {author} {\bibfnamefont {J.}~\bibnamefont {Pe\ifmmode~\check{r}\else \v{r}\fi{}ina}}, \bibinfo {author} {\bibfnamefont {J.}~\bibnamefont {Pe\ifmmode~\check{r}\else \v{r}\fi{}ina}}, \ and\ \bibinfo {author} {\bibfnamefont {A.}~\bibnamefont {Miranowicz}},\ }\href {\doibase 10.1103/PhysRevA.94.013807} {\bibfield  {journal} {\bibinfo  {journal} {Phys. Rev. A}\ }\textbf {\bibinfo {volume} {94}},\ \bibinfo {pages} {013807} (\bibinfo {year} {2016})}\BibitemShut {NoStop}%
\bibitem [{\citenamefont {Rohith}\ and\ \citenamefont {Sudheesh}(2016)}]{RohithSudheesh2016}%
  \BibitemOpen
  \bibfield  {author} {\bibinfo {author} {\bibfnamefont {M.}~\bibnamefont {Rohith}}\ and\ \bibinfo {author} {\bibfnamefont {C.}~\bibnamefont {Sudheesh}},\ }\href {\doibase 10.1364/JOSAB.33.000126} {\bibfield  {journal} {\bibinfo  {journal} {J. Opt. Soc. Am. B}\ }\textbf {\bibinfo {volume} {33}},\ \bibinfo {pages} {126} (\bibinfo {year} {2016})}\BibitemShut {NoStop}%
\bibitem [{\citenamefont {Gholipour}\ and\ \citenamefont {Shahandeh}(2016)}]{GholipourShahandeh2016}%
  \BibitemOpen
  \bibfield  {author} {\bibinfo {author} {\bibfnamefont {H.}~\bibnamefont {Gholipour}}\ and\ \bibinfo {author} {\bibfnamefont {F.}~\bibnamefont {Shahandeh}},\ }\href {\doibase 10.1103/PhysRevA.93.062318} {\bibfield  {journal} {\bibinfo  {journal} {Physical Review A}\ }\textbf {\bibinfo {volume} {93}},\ \bibinfo {pages} {062318} (\bibinfo {year} {2016})}\BibitemShut {NoStop}%
\bibitem [{\citenamefont {Ma}\ \emph {et~al.}(2016)\citenamefont {Ma}, \citenamefont {Yadin}, \citenamefont {Girolami}, \citenamefont {Vedral},\ and\ \citenamefont {Gu}}]{Maetal2016}%
  \BibitemOpen
  \bibfield  {author} {\bibinfo {author} {\bibfnamefont {J.}~\bibnamefont {Ma}}, \bibinfo {author} {\bibfnamefont {B.}~\bibnamefont {Yadin}}, \bibinfo {author} {\bibfnamefont {D.}~\bibnamefont {Girolami}}, \bibinfo {author} {\bibfnamefont {V.}~\bibnamefont {Vedral}}, \ and\ \bibinfo {author} {\bibfnamefont {M.}~\bibnamefont {Gu}},\ }\href {\doibase 10.1103/PhysRevLett.116.160407} {\bibfield  {journal} {\bibinfo  {journal} {Physical Review Letters}\ }\textbf {\bibinfo {volume} {116}},\ \bibinfo {pages} {160407} (\bibinfo {year} {2016})}\BibitemShut {NoStop}%
\bibitem [{\citenamefont {Killoran}\ \emph {et~al.}(2016{\natexlab{a}})\citenamefont {Killoran}, \citenamefont {Steinhoff},\ and\ \citenamefont {Plenio}}]{Killoranetal2016}%
  \BibitemOpen
  \bibfield  {author} {\bibinfo {author} {\bibfnamefont {N.}~\bibnamefont {Killoran}}, \bibinfo {author} {\bibfnamefont {F.~E.~S.}\ \bibnamefont {Steinhoff}}, \ and\ \bibinfo {author} {\bibfnamefont {M.~B.}\ \bibnamefont {Plenio}},\ }\href {\doibase 10.1103/PhysRevLett.116.080402} {\bibfield  {journal} {\bibinfo  {journal} {Physical Review Letters}\ }\textbf {\bibinfo {volume} {116}},\ \bibinfo {pages} {080402} (\bibinfo {year} {2016}{\natexlab{a}})}\BibitemShut {NoStop}%
\bibitem [{\citenamefont {Bose}\ and\ \citenamefont {Kumar}(2017)}]{BoseKumar2017}%
  \BibitemOpen
  \bibfield  {author} {\bibinfo {author} {\bibfnamefont {S.}~\bibnamefont {Bose}}\ and\ \bibinfo {author} {\bibfnamefont {M.~S.}\ \bibnamefont {Kumar}},\ }\href {\doibase 10.1103/PhysRevA.95.012330} {\bibfield  {journal} {\bibinfo  {journal} {Phys. Rev. A}\ }\textbf {\bibinfo {volume} {95}},\ \bibinfo {pages} {012330} (\bibinfo {year} {2017})}\BibitemShut {NoStop}%
\bibitem [{\citenamefont {Park}\ \emph {et~al.}(2017)\citenamefont {Park}, \citenamefont {Lu}, \citenamefont {Lee}, \citenamefont {Shen}, \citenamefont {Zhang}, \citenamefont {Zhang}, \citenamefont {Zubairy}, \citenamefont {Kim},\ and\ \citenamefont {Nha}}]{Parketal2017}%
  \BibitemOpen
  \bibfield  {author} {\bibinfo {author} {\bibfnamefont {J.}~\bibnamefont {Park}}, \bibinfo {author} {\bibfnamefont {Y.}~\bibnamefont {Lu}}, \bibinfo {author} {\bibfnamefont {J.}~\bibnamefont {Lee}}, \bibinfo {author} {\bibfnamefont {Y.}~\bibnamefont {Shen}}, \bibinfo {author} {\bibfnamefont {K.}~\bibnamefont {Zhang}}, \bibinfo {author} {\bibfnamefont {S.}~\bibnamefont {Zhang}}, \bibinfo {author} {\bibfnamefont {M.~S.}\ \bibnamefont {Zubairy}}, \bibinfo {author} {\bibfnamefont {K.}~\bibnamefont {Kim}}, \ and\ \bibinfo {author} {\bibfnamefont {H.}~\bibnamefont {Nha}},\ }\href {\doibase 10.1073/pnas.1617621114} {\bibfield  {journal} {\bibinfo  {journal} {Proceedings of the National Academy of Sciences}\ }\textbf {\bibinfo {volume} {114}},\ \bibinfo {pages} {891} (\bibinfo {year} {2017})}\BibitemShut {NoStop}%
\bibitem [{\citenamefont {Goldberg}\ and\ \citenamefont {James}(2018)}]{GoldbergJames2018}%
  \BibitemOpen
  \bibfield  {author} {\bibinfo {author} {\bibfnamefont {A.~Z.}\ \bibnamefont {Goldberg}}\ and\ \bibinfo {author} {\bibfnamefont {D.~F.~V.}\ \bibnamefont {James}},\ }\href {\doibase 10.1088/1751-8121/aad7c6} {\bibfield  {journal} {\bibinfo  {journal} {Journal of Physics A: Mathematical and Theoretical}\ }\textbf {\bibinfo {volume} {51}},\ \bibinfo {pages} {385303} (\bibinfo {year} {2018})}\BibitemShut {NoStop}%
\bibitem [{\citenamefont {Kwon}\ \emph {et~al.}(2019{\natexlab{a}})\citenamefont {Kwon}, \citenamefont {Tan}, \citenamefont {Volkoff},\ and\ \citenamefont {Jeong}}]{Kwonetal2019}%
  \BibitemOpen
  \bibfield  {author} {\bibinfo {author} {\bibfnamefont {H.}~\bibnamefont {Kwon}}, \bibinfo {author} {\bibfnamefont {K.~C.}\ \bibnamefont {Tan}}, \bibinfo {author} {\bibfnamefont {T.}~\bibnamefont {Volkoff}}, \ and\ \bibinfo {author} {\bibfnamefont {H.}~\bibnamefont {Jeong}},\ }\href {\doibase 10.1103/PhysRevLett.122.040503} {\bibfield  {journal} {\bibinfo  {journal} {Phys. Rev. Lett.}\ }\textbf {\bibinfo {volume} {122}},\ \bibinfo {pages} {040503} (\bibinfo {year} {2019}{\natexlab{a}})}\BibitemShut {NoStop}%
\bibitem [{\citenamefont {Tasgin}\ and\ \citenamefont {Zubairy}(2020)}]{Tasginetal2020}%
  \BibitemOpen
  \bibfield  {author} {\bibinfo {author} {\bibfnamefont {M.~E.}\ \bibnamefont {Tasgin}}\ and\ \bibinfo {author} {\bibfnamefont {M.~S.}\ \bibnamefont {Zubairy}},\ }\href {\doibase 10.1103/PhysRevA.101.012324} {\bibfield  {journal} {\bibinfo  {journal} {Phys. Rev. A}\ }\textbf {\bibinfo {volume} {101}},\ \bibinfo {pages} {012324} (\bibinfo {year} {2020})}\BibitemShut {NoStop}%
\bibitem [{\citenamefont {Fu}\ \emph {et~al.}(2020)\citenamefont {Fu}, \citenamefont {Luo},\ and\ \citenamefont {Zhang}}]{Fuetal2020}%
  \BibitemOpen
  \bibfield  {author} {\bibinfo {author} {\bibfnamefont {S.}~\bibnamefont {Fu}}, \bibinfo {author} {\bibfnamefont {S.}~\bibnamefont {Luo}}, \ and\ \bibinfo {author} {\bibfnamefont {Y.}~\bibnamefont {Zhang}},\ }\href {\doibase 10.1209/0295-5075/128/30003} {\bibfield  {journal} {\bibinfo  {journal} {{EPL} (Europhysics Letters)}\ }\textbf {\bibinfo {volume} {128}},\ \bibinfo {pages} {30003} (\bibinfo {year} {2020})}\BibitemShut {NoStop}%
\bibitem [{\citenamefont {Tserkis}\ \emph {et~al.}(2020)\citenamefont {Tserkis}, \citenamefont {Thompson}, \citenamefont {Lund}, \citenamefont {Ralph}, \citenamefont {Lam}, \citenamefont {Gu},\ and\ \citenamefont {Assad}}]{Tserkisetal2020}%
  \BibitemOpen
  \bibfield  {author} {\bibinfo {author} {\bibfnamefont {S.}~\bibnamefont {Tserkis}}, \bibinfo {author} {\bibfnamefont {J.}~\bibnamefont {Thompson}}, \bibinfo {author} {\bibfnamefont {A.~P.}\ \bibnamefont {Lund}}, \bibinfo {author} {\bibfnamefont {T.~C.}\ \bibnamefont {Ralph}}, \bibinfo {author} {\bibfnamefont {P.~K.}\ \bibnamefont {Lam}}, \bibinfo {author} {\bibfnamefont {M.}~\bibnamefont {Gu}}, \ and\ \bibinfo {author} {\bibfnamefont {S.~M.}\ \bibnamefont {Assad}},\ }\href {\doibase 10.1103/PhysRevA.102.052418} {\bibfield  {journal} {\bibinfo  {journal} {Phys. Rev. A}\ }\textbf {\bibinfo {volume} {102}},\ \bibinfo {pages} {052418} (\bibinfo {year} {2020})}\BibitemShut {NoStop}%
\bibitem [{\citenamefont {Goldberg}\ and\ \citenamefont {Heshami}(2021)}]{GoldbergHeshami2021}%
  \BibitemOpen
  \bibfield  {author} {\bibinfo {author} {\bibfnamefont {A.~Z.}\ \bibnamefont {Goldberg}}\ and\ \bibinfo {author} {\bibfnamefont {K.}~\bibnamefont {Heshami}},\ }\href {\doibase 10.1103/PhysRevA.104.032425} {\bibfield  {journal} {\bibinfo  {journal} {Phys. Rev. A}\ }\textbf {\bibinfo {volume} {104}},\ \bibinfo {pages} {032425} (\bibinfo {year} {2021})}\BibitemShut {NoStop}%
\bibitem [{\citenamefont {Ataman}(2022)}]{Ataman2022}%
  \BibitemOpen
  \bibfield  {author} {\bibinfo {author} {\bibfnamefont {S.}~\bibnamefont {Ataman}},\ }\href {\doibase 10.1140/epjd/s10053-022-00559-4} {\bibfield  {journal} {\bibinfo  {journal} {The European Physical Journal D}\ }\textbf {\bibinfo {volume} {76}},\ \bibinfo {pages} {233} (\bibinfo {year} {2022})}\BibitemShut {NoStop}%
\bibitem [{\citenamefont {Akella}\ \emph {et~al.}(2022)\citenamefont {Akella}, \citenamefont {Thapliyal}, \citenamefont {Mani},\ and\ \citenamefont {Pathak}}]{Akellaetal2022}%
  \BibitemOpen
  \bibfield  {author} {\bibinfo {author} {\bibfnamefont {S.}~\bibnamefont {Akella}}, \bibinfo {author} {\bibfnamefont {K.}~\bibnamefont {Thapliyal}}, \bibinfo {author} {\bibfnamefont {H.~S.}\ \bibnamefont {Mani}}, \ and\ \bibinfo {author} {\bibfnamefont {A.}~\bibnamefont {Pathak}},\ }\href {\doibase 10.1364/JOSAB.459265} {\bibfield  {journal} {\bibinfo  {journal} {J. Opt. Soc. Am. B}\ }\textbf {\bibinfo {volume} {39}},\ \bibinfo {pages} {1829} (\bibinfo {year} {2022})}\BibitemShut {NoStop}%
\bibitem [{\citenamefont {Li}\ \emph {et~al.}(2023)\citenamefont {Li}, \citenamefont {Das}, \citenamefont {Tserkis}, \citenamefont {Narang}, \citenamefont {Lam},\ and\ \citenamefont {Assad}}]{Lietal2023}%
  \BibitemOpen
  \bibfield  {author} {\bibinfo {author} {\bibfnamefont {B.}~\bibnamefont {Li}}, \bibinfo {author} {\bibfnamefont {A.}~\bibnamefont {Das}}, \bibinfo {author} {\bibfnamefont {S.}~\bibnamefont {Tserkis}}, \bibinfo {author} {\bibfnamefont {P.}~\bibnamefont {Narang}}, \bibinfo {author} {\bibfnamefont {P.~K.}\ \bibnamefont {Lam}}, \ and\ \bibinfo {author} {\bibfnamefont {S.~M.}\ \bibnamefont {Assad}},\ }\href {\doibase 10.1038/s41598-023-38572-1} {\bibfield  {journal} {\bibinfo  {journal} {Scientific Reports}\ }\textbf {\bibinfo {volume} {13}},\ \bibinfo {pages} {11722} (\bibinfo {year} {2023})}\BibitemShut {NoStop}%
\bibitem [{\citenamefont {Simonetti}\ \emph {et~al.}(2023)\citenamefont {Simonetti}, \citenamefont {Perri},\ and\ \citenamefont {Gervasi}}]{Simonettietal2023}%
  \BibitemOpen
  \bibfield  {author} {\bibinfo {author} {\bibfnamefont {M.}~\bibnamefont {Simonetti}}, \bibinfo {author} {\bibfnamefont {D.}~\bibnamefont {Perri}}, \ and\ \bibinfo {author} {\bibfnamefont {O.}~\bibnamefont {Gervasi}},\ }in\ \href@noop {} {\emph {\bibinfo {booktitle} {Computational Science and Its Applications -- ICCSA 2023 Workshops}}},\ \bibinfo {editor} {edited by\ \bibinfo {editor} {\bibfnamefont {O.}~\bibnamefont {Gervasi}}, \bibinfo {editor} {\bibfnamefont {B.}~\bibnamefont {Murgante}}, \bibinfo {editor} {\bibfnamefont {A.~M. A.~C.}\ \bibnamefont {Rocha}}, \bibinfo {editor} {\bibfnamefont {C.}~\bibnamefont {Garau}}, \bibinfo {editor} {\bibfnamefont {F.}~\bibnamefont {Scorza}}, \bibinfo {editor} {\bibfnamefont {Y.}~\bibnamefont {Karaca}}, \ and\ \bibinfo {editor} {\bibfnamefont {C.~M.}\ \bibnamefont {Torre}}}\ (\bibinfo  {publisher} {Springer Nature Switzerland},\ \bibinfo {address} {Cham},\ \bibinfo {year} {2023})\ pp.\ \bibinfo {pages} {116--129}\BibitemShut {NoStop}%
\bibitem [{\citenamefont {~}\ and\ \citenamefont {Chatterjee}(2023)}]{DeepakChatterjee2023}%
  \BibitemOpen
  \bibfield  {author} {\bibinfo {author} {\bibfnamefont {D.}~\bibnamefont {~}}\ and\ \bibinfo {author} {\bibfnamefont {A.}~\bibnamefont {Chatterjee}},\ }\href {\doibase 10.1139/cjp-2023-0085} {\bibfield  {journal} {\bibinfo  {journal} {Canadian Journal of Physics}\ }\textbf {\bibinfo {volume} {101}},\ \bibinfo {pages} {560} (\bibinfo {year} {2023})},\ \Eprint {http://arxiv.org/abs/https://doi.org/10.1139/cjp-2023-0085} {https://doi.org/10.1139/cjp-2023-0085} \BibitemShut {NoStop}%
\bibitem [{\citenamefont {Steinhoff}(2024)}]{Steinhoff2024}%
  \BibitemOpen
  \bibfield  {author} {\bibinfo {author} {\bibfnamefont {F.~E.~S.}\ \bibnamefont {Steinhoff}},\ }\href {\doibase 10.1103/PhysRevA.110.022409} {\bibfield  {journal} {\bibinfo  {journal} {Phys. Rev. A}\ }\textbf {\bibinfo {volume} {110}},\ \bibinfo {pages} {022409} (\bibinfo {year} {2024})}\BibitemShut {NoStop}%
\bibitem [{\citenamefont {Serrano-Ensástiga}\ \emph {et~al.}(2024)\citenamefont {Serrano-Ensástiga}, \citenamefont {Galindo}, \citenamefont {Maytorena},\ and\ \citenamefont {Chryssomalakos}}]{SerranoEnsastigaetal2024arxiv}%
  \BibitemOpen
  \bibfield  {author} {\bibinfo {author} {\bibfnamefont {E.}~\bibnamefont {Serrano-Ensástiga}}, \bibinfo {author} {\bibfnamefont {D.~M.}\ \bibnamefont {Galindo}}, \bibinfo {author} {\bibfnamefont {J.~A.}\ \bibnamefont {Maytorena}}, \ and\ \bibinfo {author} {\bibfnamefont {C.}~\bibnamefont {Chryssomalakos}},\ }\href {https://arxiv.org/abs/2410.03361} {\enquote {\bibinfo {title} {Entangling power of spin-j systems: a geometrical approach},}\ } (\bibinfo {year} {2024}),\ \Eprint {http://arxiv.org/abs/2410.03361} {arXiv:2410.03361 [quant-ph]} \BibitemShut {NoStop}%
\bibitem [{\citenamefont {Li}\ \emph {et~al.}(2024)\citenamefont {Li}, \citenamefont {Wang}, \citenamefont {Tang},\ and\ \citenamefont {Ian}}]{Lietal2024arxiv}%
  \BibitemOpen
  \bibfield  {author} {\bibinfo {author} {\bibfnamefont {M.}~\bibnamefont {Li}}, \bibinfo {author} {\bibfnamefont {W.}~\bibnamefont {Wang}}, \bibinfo {author} {\bibfnamefont {Z.}~\bibnamefont {Tang}}, \ and\ \bibinfo {author} {\bibfnamefont {H.}~\bibnamefont {Ian}},\ }\href {https://arxiv.org/abs/2310.04065} {\enquote {\bibinfo {title} {Entanglement and classical nonseparability convertible from orthogonal polarizations},}\ } (\bibinfo {year} {2024}),\ \Eprint {http://arxiv.org/abs/2310.04065} {arXiv:2310.04065 [quant-ph]} \BibitemShut {NoStop}%
\bibitem [{\citenamefont {Manfredi}\ and\ \citenamefont {Feix}(2000)}]{Manfredi2000}%
  \BibitemOpen
  \bibfield  {author} {\bibinfo {author} {\bibfnamefont {G.}~\bibnamefont {Manfredi}}\ and\ \bibinfo {author} {\bibfnamefont {M.~R.}\ \bibnamefont {Feix}},\ }\href {\doibase 10.1103/PhysRevE.62.4665} {\bibfield  {journal} {\bibinfo  {journal} {Phys. Rev. E}\ }\textbf {\bibinfo {volume} {62}},\ \bibinfo {pages} {4665} (\bibinfo {year} {2000})}\BibitemShut {NoStop}%
\bibitem [{\citenamefont {VonNeumann}(1932)}]{VonNeumann1932}%
  \BibitemOpen
  \bibfield  {author} {\bibinfo {author} {\bibfnamefont {J.}~\bibnamefont {VonNeumann}},\ }\href {http://eudml.org/doc/203794} {\emph {\bibinfo {title} {Mathematische Grundlagen der Quantenmechanik}}}\ (\bibinfo  {publisher} {Springer},\ \bibinfo {year} {1932})\BibitemShut {NoStop}%
\bibitem [{\citenamefont {Vidal}(2000)}]{Vidal2000}%
  \BibitemOpen
  \bibfield  {author} {\bibinfo {author} {\bibfnamefont {G.}~\bibnamefont {Vidal}},\ }\href {\doibase 10.1080/09500340008244048} {\bibfield  {journal} {\bibinfo  {journal} {Journal of Modern Optics}\ }\textbf {\bibinfo {volume} {47}},\ \bibinfo {pages} {355} (\bibinfo {year} {2000})}\BibitemShut {NoStop}%
\bibitem [{\citenamefont {Bengtsson}\ and\ \citenamefont {Zyczkowski}(2006)}]{BengtssonZyczkowski2006}%
  \BibitemOpen
  \bibfield  {author} {\bibinfo {author} {\bibfnamefont {I.}~\bibnamefont {Bengtsson}}\ and\ \bibinfo {author} {\bibfnamefont {K.}~\bibnamefont {Zyczkowski}},\ }\href@noop {} {\emph {\bibinfo {title} {Geometry of Quantum States: An Introduction to Quantum Entanglement}}}\ (\bibinfo  {publisher} {Cambridge University Press},\ \bibinfo {year} {2006})\BibitemShut {NoStop}%
\bibitem [{\citenamefont {Horodecki}\ \emph {et~al.}(2009)\citenamefont {Horodecki}, \citenamefont {Horodecki}, \citenamefont {Horodecki},\ and\ \citenamefont {Horodecki}}]{Horodeckietal2009}%
  \BibitemOpen
  \bibfield  {author} {\bibinfo {author} {\bibfnamefont {R.}~\bibnamefont {Horodecki}}, \bibinfo {author} {\bibfnamefont {P.}~\bibnamefont {Horodecki}}, \bibinfo {author} {\bibfnamefont {M.}~\bibnamefont {Horodecki}}, \ and\ \bibinfo {author} {\bibfnamefont {K.}~\bibnamefont {Horodecki}},\ }\href {\doibase 10.1103/RevModPhys.81.865} {\bibfield  {journal} {\bibinfo  {journal} {Reviews of Modern Physics}\ }\textbf {\bibinfo {volume} {81}},\ \bibinfo {pages} {865} (\bibinfo {year} {2009})}\BibitemShut {NoStop}%
\bibitem [{\citenamefont {De~Bi\`evre}\ \emph {et~al.}(2019)\citenamefont {De~Bi\`evre}, \citenamefont {Horoshko}, \citenamefont {Patera},\ and\ \citenamefont {Kolobov}}]{Debievre}%
  \BibitemOpen
  \bibfield  {author} {\bibinfo {author} {\bibfnamefont {S.}~\bibnamefont {De~Bi\`evre}}, \bibinfo {author} {\bibfnamefont {D.~B.}\ \bibnamefont {Horoshko}}, \bibinfo {author} {\bibfnamefont {G.}~\bibnamefont {Patera}}, \ and\ \bibinfo {author} {\bibfnamefont {M.~I.}\ \bibnamefont {Kolobov}},\ }\href {\doibase 10.1103/PhysRevLett.122.080402} {\bibfield  {journal} {\bibinfo  {journal} {Phys. Rev. Lett.}\ }\textbf {\bibinfo {volume} {122}},\ \bibinfo {pages} {080402} (\bibinfo {year} {2019})}\BibitemShut {NoStop}%
\bibitem [{\citenamefont {Titulaer}\ and\ \citenamefont {{G}lauber}(1965)}]{Titulaer}%
  \BibitemOpen
  \bibfield  {author} {\bibinfo {author} {\bibfnamefont {U.~M.}\ \bibnamefont {Titulaer}}\ and\ \bibinfo {author} {\bibfnamefont {R.~J.}\ \bibnamefont {{G}lauber}},\ }\href {\doibase 10.1103/PhysRev.140.B676} {\bibfield  {journal} {\bibinfo  {journal} {Phys. Rev.}\ }\textbf {\bibinfo {volume} {140}},\ \bibinfo {pages} {B676} (\bibinfo {year} {1965})}\BibitemShut {NoStop}%
\bibitem [{\citenamefont {Hillery}(1987)}]{Hillery1}%
  \BibitemOpen
  \bibfield  {author} {\bibinfo {author} {\bibfnamefont {M.}~\bibnamefont {Hillery}},\ }\href {\doibase 10.1103/PhysRevA.35.725} {\bibfield  {journal} {\bibinfo  {journal} {Phys. Rev. A}\ }\textbf {\bibinfo {volume} {35}},\ \bibinfo {pages} {725} (\bibinfo {year} {1987})}\BibitemShut {NoStop}%
\bibitem [{\citenamefont {Kenfack}\ and\ \citenamefont {Zyczkowski}(2004)}]{Kenfack}%
  \BibitemOpen
  \bibfield  {author} {\bibinfo {author} {\bibfnamefont {A.}~\bibnamefont {Kenfack}}\ and\ \bibinfo {author} {\bibfnamefont {K.}~\bibnamefont {Zyczkowski}},\ }\href {\doibase 10.1088/1464-4266/6/10/003} {\bibfield  {journal} {\bibinfo  {journal} {Journal of Optics B: Quantum and Semiclassical Optics}\ }\textbf {\bibinfo {volume} {6}},\ \bibinfo {pages} {396} (\bibinfo {year} {2004})}\BibitemShut {NoStop}%
\bibitem [{\citenamefont {Bach}\ and\ \citenamefont {Luxmann-Ellinghaus}(1986)}]{Bach}%
  \BibitemOpen
  \bibfield  {author} {\bibinfo {author} {\bibfnamefont {A.}~\bibnamefont {Bach}}\ and\ \bibinfo {author} {\bibfnamefont {U.}~\bibnamefont {Luxmann-Ellinghaus}},\ }\href {\doibase 10.1007/BF01205485} {\bibfield  {journal} {\bibinfo  {journal} {Commun. Math. Phys.}\ }\textbf {\bibinfo {volume} {107}} (\bibinfo {year} {1986}),\ 10.1007/BF01205485}\BibitemShut {NoStop}%
\bibitem [{\citenamefont {Hillery}(1989)}]{Hillery3}%
  \BibitemOpen
  \bibfield  {author} {\bibinfo {author} {\bibfnamefont {M.}~\bibnamefont {Hillery}},\ }\href {\doibase 10.1103/PhysRevA.39.2994} {\bibfield  {journal} {\bibinfo  {journal} {Phys. Rev. A}\ }\textbf {\bibinfo {volume} {39}},\ \bibinfo {pages} {2994} (\bibinfo {year} {1989})}\BibitemShut {NoStop}%
\bibitem [{\citenamefont {Lee}(1995)}]{Lee}%
  \BibitemOpen
  \bibfield  {author} {\bibinfo {author} {\bibfnamefont {C.~T.}\ \bibnamefont {Lee}},\ }\href {\doibase 10.1103/PhysRevA.52.3374} {\bibfield  {journal} {\bibinfo  {journal} {Phys. Rev. A}\ }\textbf {\bibinfo {volume} {52}},\ \bibinfo {pages} {3374} (\bibinfo {year} {1995})}\BibitemShut {NoStop}%
\bibitem [{\citenamefont {Agarwal}\ and\ \citenamefont {Tara}(1992)}]{Agarwal2}%
  \BibitemOpen
  \bibfield  {author} {\bibinfo {author} {\bibfnamefont {G.~S.}\ \bibnamefont {Agarwal}}\ and\ \bibinfo {author} {\bibfnamefont {K.}~\bibnamefont {Tara}},\ }\href {\doibase 10.1103/PhysRevA.46.485} {\bibfield  {journal} {\bibinfo  {journal} {Phys. Rev. A}\ }\textbf {\bibinfo {volume} {46}},\ \bibinfo {pages} {485} (\bibinfo {year} {1992})}\BibitemShut {NoStop}%
\bibitem [{\citenamefont {L\"utkenhaus}\ and\ \citenamefont {Barnett}(1995)}]{Lutkenhaus}%
  \BibitemOpen
  \bibfield  {author} {\bibinfo {author} {\bibfnamefont {N.}~\bibnamefont {L\"utkenhaus}}\ and\ \bibinfo {author} {\bibfnamefont {S.~M.}\ \bibnamefont {Barnett}},\ }\href {\doibase 10.1103/PhysRevA.51.3340} {\bibfield  {journal} {\bibinfo  {journal} {Phys. Rev. A}\ }\textbf {\bibinfo {volume} {51}},\ \bibinfo {pages} {3340} (\bibinfo {year} {1995})}\BibitemShut {NoStop}%
\bibitem [{\citenamefont {Dodonov}\ \emph {et~al.}(2000{\natexlab{b}})\citenamefont {Dodonov}, \citenamefont {Man'ko}, \citenamefont {Man'ko},\ and\ \citenamefont {Wünsche}}]{Dodonov}%
  \BibitemOpen
  \bibfield  {author} {\bibinfo {author} {\bibfnamefont {V.~V.}\ \bibnamefont {Dodonov}}, \bibinfo {author} {\bibfnamefont {O.~V.}\ \bibnamefont {Man'ko}}, \bibinfo {author} {\bibfnamefont {V.~I.}\ \bibnamefont {Man'ko}}, \ and\ \bibinfo {author} {\bibfnamefont {A.}~\bibnamefont {Wünsche}},\ }\href {\doibase 10.1080/09500340008233385} {\bibfield  {journal} {\bibinfo  {journal} {Journal of Modern Optics}\ }\textbf {\bibinfo {volume} {47}},\ \bibinfo {pages} {633} (\bibinfo {year} {2000}{\natexlab{b}})}\BibitemShut {NoStop}%
\bibitem [{\citenamefont {Marian}\ \emph {et~al.}(2002)\citenamefont {Marian}, \citenamefont {Marian},\ and\ \citenamefont {Scutaru}}]{Marian}%
  \BibitemOpen
  \bibfield  {author} {\bibinfo {author} {\bibfnamefont {P.}~\bibnamefont {Marian}}, \bibinfo {author} {\bibfnamefont {T.~A.}\ \bibnamefont {Marian}}, \ and\ \bibinfo {author} {\bibfnamefont {H.}~\bibnamefont {Scutaru}},\ }\href {\doibase 10.1103/PhysRevLett.88.153601} {\bibfield  {journal} {\bibinfo  {journal} {Phys. Rev. Lett.}\ }\textbf {\bibinfo {volume} {88}},\ \bibinfo {pages} {153601} (\bibinfo {year} {2002})}\BibitemShut {NoStop}%
\bibitem [{\citenamefont {Richter}\ and\ \citenamefont {Vogel}(2002)}]{Richter}%
  \BibitemOpen
  \bibfield  {author} {\bibinfo {author} {\bibfnamefont {T.}~\bibnamefont {Richter}}\ and\ \bibinfo {author} {\bibfnamefont {W.}~\bibnamefont {Vogel}},\ }\href {\doibase 10.1103/PhysRevLett.89.283601} {\bibfield  {journal} {\bibinfo  {journal} {Phys. Rev. Lett.}\ }\textbf {\bibinfo {volume} {89}},\ \bibinfo {pages} {283601} (\bibinfo {year} {2002})}\BibitemShut {NoStop}%
\bibitem [{\citenamefont {Ryl}\ \emph {et~al.}(2015)\citenamefont {Ryl}, \citenamefont {Sperling}, \citenamefont {Agudelo}, \citenamefont {Mraz}, \citenamefont {K\"ohnke}, \citenamefont {Hage},\ and\ \citenamefont {Vogel}}]{Ryl}%
  \BibitemOpen
  \bibfield  {author} {\bibinfo {author} {\bibfnamefont {S.}~\bibnamefont {Ryl}}, \bibinfo {author} {\bibfnamefont {J.}~\bibnamefont {Sperling}}, \bibinfo {author} {\bibfnamefont {E.}~\bibnamefont {Agudelo}}, \bibinfo {author} {\bibfnamefont {M.}~\bibnamefont {Mraz}}, \bibinfo {author} {\bibfnamefont {S.}~\bibnamefont {K\"ohnke}}, \bibinfo {author} {\bibfnamefont {B.}~\bibnamefont {Hage}}, \ and\ \bibinfo {author} {\bibfnamefont {W.}~\bibnamefont {Vogel}},\ }\href {\doibase 10.1103/PhysRevA.92.011801} {\bibfield  {journal} {\bibinfo  {journal} {Phys. Rev. A}\ }\textbf {\bibinfo {volume} {92}},\ \bibinfo {pages} {011801} (\bibinfo {year} {2015})}\BibitemShut {NoStop}%
\bibitem [{\citenamefont {Sperling}\ and\ \citenamefont {Vogel}(2015)}]{Sperling}%
  \BibitemOpen
  \bibfield  {author} {\bibinfo {author} {\bibfnamefont {J.}~\bibnamefont {Sperling}}\ and\ \bibinfo {author} {\bibfnamefont {W.}~\bibnamefont {Vogel}},\ }\href {\doibase 10.1088/0031-8949/90/7/074024} {\bibfield  {journal} {\bibinfo  {journal} {Physica Scripta}\ }\textbf {\bibinfo {volume} {90}},\ \bibinfo {pages} {074024} (\bibinfo {year} {2015})}\BibitemShut {NoStop}%
\bibitem [{\citenamefont {Killoran}\ \emph {et~al.}(2016{\natexlab{b}})\citenamefont {Killoran}, \citenamefont {Steinhoff},\ and\ \citenamefont {Plenio}}]{Killoran}%
  \BibitemOpen
  \bibfield  {author} {\bibinfo {author} {\bibfnamefont {N.}~\bibnamefont {Killoran}}, \bibinfo {author} {\bibfnamefont {F.~E.~S.}\ \bibnamefont {Steinhoff}}, \ and\ \bibinfo {author} {\bibfnamefont {M.~B.}\ \bibnamefont {Plenio}},\ }\href {\doibase 10.1103/PhysRevLett.116.080402} {\bibfield  {journal} {\bibinfo  {journal} {Phys. Rev. Lett.}\ }\textbf {\bibinfo {volume} {116}},\ \bibinfo {pages} {080402} (\bibinfo {year} {2016}{\natexlab{b}})}\BibitemShut {NoStop}%
\bibitem [{\citenamefont {Alexanian}(2018)}]{Alexanian}%
  \BibitemOpen
  \bibfield  {author} {\bibinfo {author} {\bibfnamefont {M.}~\bibnamefont {Alexanian}},\ }\href {\doibase 10.1080/09500340.2017.1374481} {\bibfield  {journal} {\bibinfo  {journal} {Journal of Modern Optics}\ }\textbf {\bibinfo {volume} {65}},\ \bibinfo {pages} {16} (\bibinfo {year} {2018})}\BibitemShut {NoStop}%
\bibitem [{\citenamefont {Nair}(2017)}]{Nair}%
  \BibitemOpen
  \bibfield  {author} {\bibinfo {author} {\bibfnamefont {R.}~\bibnamefont {Nair}},\ }\href {\doibase 10.1103/PhysRevA.95.063835} {\bibfield  {journal} {\bibinfo  {journal} {Phys. Rev. A}\ }\textbf {\bibinfo {volume} {95}},\ \bibinfo {pages} {063835} (\bibinfo {year} {2017})}\BibitemShut {NoStop}%
\bibitem [{\citenamefont {Ryl}\ \emph {et~al.}(2017)\citenamefont {Ryl}, \citenamefont {Sperling},\ and\ \citenamefont {Vogel}}]{Ryl2}%
  \BibitemOpen
  \bibfield  {author} {\bibinfo {author} {\bibfnamefont {S.}~\bibnamefont {Ryl}}, \bibinfo {author} {\bibfnamefont {J.}~\bibnamefont {Sperling}}, \ and\ \bibinfo {author} {\bibfnamefont {W.}~\bibnamefont {Vogel}},\ }\href {\doibase 10.1103/PhysRevA.95.053825} {\bibfield  {journal} {\bibinfo  {journal} {Phys. Rev. A}\ }\textbf {\bibinfo {volume} {95}},\ \bibinfo {pages} {053825} (\bibinfo {year} {2017})}\BibitemShut {NoStop}%
\bibitem [{\citenamefont {Yadin}\ \emph {et~al.}(2018)\citenamefont {Yadin}, \citenamefont {Binder}, \citenamefont {Thompson}, \citenamefont {Narasimhachar}, \citenamefont {Gu},\ and\ \citenamefont {Kim}}]{Yadin}%
  \BibitemOpen
  \bibfield  {author} {\bibinfo {author} {\bibfnamefont {B.}~\bibnamefont {Yadin}}, \bibinfo {author} {\bibfnamefont {F.~C.}\ \bibnamefont {Binder}}, \bibinfo {author} {\bibfnamefont {J.}~\bibnamefont {Thompson}}, \bibinfo {author} {\bibfnamefont {V.}~\bibnamefont {Narasimhachar}}, \bibinfo {author} {\bibfnamefont {M.}~\bibnamefont {Gu}}, \ and\ \bibinfo {author} {\bibfnamefont {M.~S.}\ \bibnamefont {Kim}},\ }\href {\doibase 10.1103/PhysRevX.8.041038} {\bibfield  {journal} {\bibinfo  {journal} {Phys. Rev. X}\ }\textbf {\bibinfo {volume} {8}},\ \bibinfo {pages} {041038} (\bibinfo {year} {2018})}\BibitemShut {NoStop}%
\bibitem [{\citenamefont {Kwon}\ \emph {et~al.}(2019{\natexlab{b}})\citenamefont {Kwon}, \citenamefont {Tan}, \citenamefont {Volkoff},\ and\ \citenamefont {Jeong}}]{Kwon2}%
  \BibitemOpen
  \bibfield  {author} {\bibinfo {author} {\bibfnamefont {H.}~\bibnamefont {Kwon}}, \bibinfo {author} {\bibfnamefont {K.~C.}\ \bibnamefont {Tan}}, \bibinfo {author} {\bibfnamefont {T.}~\bibnamefont {Volkoff}}, \ and\ \bibinfo {author} {\bibfnamefont {H.}~\bibnamefont {Jeong}},\ }\href {\doibase 10.1103/PhysRevLett.122.040503} {\bibfield  {journal} {\bibinfo  {journal} {Phys. Rev. Lett.}\ }\textbf {\bibinfo {volume} {122}},\ \bibinfo {pages} {040503} (\bibinfo {year} {2019}{\natexlab{b}})}\BibitemShut {NoStop}%
\bibitem [{\citenamefont {Takagi}\ and\ \citenamefont {Zhuang}(2018)}]{Takagi18}%
  \BibitemOpen
  \bibfield  {author} {\bibinfo {author} {\bibfnamefont {R.}~\bibnamefont {Takagi}}\ and\ \bibinfo {author} {\bibfnamefont {Q.}~\bibnamefont {Zhuang}},\ }\href {\doibase 10.1103/PhysRevA.97.062337} {\bibfield  {journal} {\bibinfo  {journal} {Phys. Rev. A}\ }\textbf {\bibinfo {volume} {97}},\ \bibinfo {pages} {062337} (\bibinfo {year} {2018})}\BibitemShut {NoStop}%
\bibitem [{\citenamefont {Horoshko}\ \emph {et~al.}(2019)\citenamefont {Horoshko}, \citenamefont {De~Bi\`evre}, \citenamefont {Patera},\ and\ \citenamefont {Kolobov}}]{Horoshko}%
  \BibitemOpen
  \bibfield  {author} {\bibinfo {author} {\bibfnamefont {D.~B.}\ \bibnamefont {Horoshko}}, \bibinfo {author} {\bibfnamefont {S.}~\bibnamefont {De~Bi\`evre}}, \bibinfo {author} {\bibfnamefont {G.}~\bibnamefont {Patera}}, \ and\ \bibinfo {author} {\bibfnamefont {M.~I.}\ \bibnamefont {Kolobov}},\ }\href {\doibase 10.1103/PhysRevA.100.053831} {\bibfield  {journal} {\bibinfo  {journal} {Phys. Rev. A}\ }\textbf {\bibinfo {volume} {100}},\ \bibinfo {pages} {053831} (\bibinfo {year} {2019})}\BibitemShut {NoStop}%
\bibitem [{\citenamefont {Luo}\ and\ \citenamefont {Zhang}(2019)}]{Luo}%
  \BibitemOpen
  \bibfield  {author} {\bibinfo {author} {\bibfnamefont {S.}~\bibnamefont {Luo}}\ and\ \bibinfo {author} {\bibfnamefont {Y.}~\bibnamefont {Zhang}},\ }\href {\doibase 10.1103/PhysRevA.100.032116} {\bibfield  {journal} {\bibinfo  {journal} {Phys. Rev. A}\ }\textbf {\bibinfo {volume} {100}},\ \bibinfo {pages} {032116} (\bibinfo {year} {2019})}\BibitemShut {NoStop}%
\bibitem [{\citenamefont {Tan}\ \emph {et~al.}(2020)\citenamefont {Tan}, \citenamefont {Choi},\ and\ \citenamefont {Jeong}}]{Tan2020}%
  \BibitemOpen
  \bibfield  {author} {\bibinfo {author} {\bibfnamefont {K.~C.}\ \bibnamefont {Tan}}, \bibinfo {author} {\bibfnamefont {S.}~\bibnamefont {Choi}}, \ and\ \bibinfo {author} {\bibfnamefont {H.}~\bibnamefont {Jeong}},\ }\href {\doibase 10.1103/PhysRevLett.124.110404} {\bibfield  {journal} {\bibinfo  {journal} {Phys. Rev. Lett.}\ }\textbf {\bibinfo {volume} {124}},\ \bibinfo {pages} {110404} (\bibinfo {year} {2020})}\BibitemShut {NoStop}%
\bibitem [{\citenamefont {Daley}\ \emph {et~al.}(2012)\citenamefont {Daley}, \citenamefont {Pichler}, \citenamefont {Schachenmayer},\ and\ \citenamefont {Zoller}}]{Daleyetal2012}%
  \BibitemOpen
  \bibfield  {author} {\bibinfo {author} {\bibfnamefont {A.~J.}\ \bibnamefont {Daley}}, \bibinfo {author} {\bibfnamefont {H.}~\bibnamefont {Pichler}}, \bibinfo {author} {\bibfnamefont {J.}~\bibnamefont {Schachenmayer}}, \ and\ \bibinfo {author} {\bibfnamefont {P.}~\bibnamefont {Zoller}},\ }\href {\doibase 10.1103/PhysRevLett.109.020505} {\bibfield  {journal} {\bibinfo  {journal} {Phys. Rev. Lett.}\ }\textbf {\bibinfo {volume} {109}},\ \bibinfo {pages} {020505} (\bibinfo {year} {2012})}\BibitemShut {NoStop}%
\bibitem [{\citenamefont {Islam}\ \emph {et~al.}(2015)\citenamefont {Islam}, \citenamefont {Ma}, \citenamefont {Preiss}, \citenamefont {Eric~Tai}, \citenamefont {Lukin}, \citenamefont {Rispoli},\ and\ \citenamefont {Greiner}}]{Islametal2015}%
  \BibitemOpen
  \bibfield  {author} {\bibinfo {author} {\bibfnamefont {R.}~\bibnamefont {Islam}}, \bibinfo {author} {\bibfnamefont {R.}~\bibnamefont {Ma}}, \bibinfo {author} {\bibfnamefont {P.~M.}\ \bibnamefont {Preiss}}, \bibinfo {author} {\bibfnamefont {M.}~\bibnamefont {Eric~Tai}}, \bibinfo {author} {\bibfnamefont {A.}~\bibnamefont {Lukin}}, \bibinfo {author} {\bibfnamefont {M.}~\bibnamefont {Rispoli}}, \ and\ \bibinfo {author} {\bibfnamefont {M.}~\bibnamefont {Greiner}},\ }\href {\doibase 10.1038/nature15750} {\bibfield  {journal} {\bibinfo  {journal} {Nature}\ }\textbf {\bibinfo {volume} {528}},\ \bibinfo {pages} {77} (\bibinfo {year} {2015})}\BibitemShut {NoStop}%
\bibitem [{\citenamefont {Griffet}\ \emph {et~al.}(2023)\citenamefont {Griffet}, \citenamefont {Arnhem}, \citenamefont {De~Bi\`evre},\ and\ \citenamefont {Cerf}}]{Griffet}%
  \BibitemOpen
  \bibfield  {author} {\bibinfo {author} {\bibfnamefont {C.}~\bibnamefont {Griffet}}, \bibinfo {author} {\bibfnamefont {M.}~\bibnamefont {Arnhem}}, \bibinfo {author} {\bibfnamefont {S.}~\bibnamefont {De~Bi\`evre}}, \ and\ \bibinfo {author} {\bibfnamefont {N.~J.}\ \bibnamefont {Cerf}},\ }\href {\doibase 10.1103/PhysRevA.108.023730} {\bibfield  {journal} {\bibinfo  {journal} {Phys. Rev. A}\ }\textbf {\bibinfo {volume} {108}},\ \bibinfo {pages} {023730} (\bibinfo {year} {2023})}\BibitemShut {NoStop}%
\bibitem [{\citenamefont {Bovino}\ \emph {et~al.}(2005)\citenamefont {Bovino}, \citenamefont {Castagnoli}, \citenamefont {Ekert}, \citenamefont {Horodecki}, \citenamefont {Alves},\ and\ \citenamefont {Sergienko}}]{Bovinoetal2005}%
  \BibitemOpen
  \bibfield  {author} {\bibinfo {author} {\bibfnamefont {F.~A.}\ \bibnamefont {Bovino}}, \bibinfo {author} {\bibfnamefont {G.}~\bibnamefont {Castagnoli}}, \bibinfo {author} {\bibfnamefont {A.}~\bibnamefont {Ekert}}, \bibinfo {author} {\bibfnamefont {P.}~\bibnamefont {Horodecki}}, \bibinfo {author} {\bibfnamefont {C.~M.}\ \bibnamefont {Alves}}, \ and\ \bibinfo {author} {\bibfnamefont {A.~V.}\ \bibnamefont {Sergienko}},\ }\href {\doibase 10.1103/PhysRevLett.95.240407} {\bibfield  {journal} {\bibinfo  {journal} {Phys. Rev. Lett.}\ }\textbf {\bibinfo {volume} {95}},\ \bibinfo {pages} {240407} (\bibinfo {year} {2005})}\BibitemShut {NoStop}%
\bibitem [{\citenamefont {Agarwal}\ and\ \citenamefont {Wolf}(1970)}]{AgarwalWolf1970}%
  \BibitemOpen
  \bibfield  {author} {\bibinfo {author} {\bibfnamefont {G.~S.}\ \bibnamefont {Agarwal}}\ and\ \bibinfo {author} {\bibfnamefont {E.}~\bibnamefont {Wolf}},\ }\href {\doibase 10.1103/PhysRevD.2.2161} {\bibfield  {journal} {\bibinfo  {journal} {Phys. Rev. D}\ }\textbf {\bibinfo {volume} {2}},\ \bibinfo {pages} {2161} (\bibinfo {year} {1970})}\BibitemShut {NoStop}%
\bibitem [{\citenamefont {Lee}(1991)}]{Lee1991}%
  \BibitemOpen
  \bibfield  {author} {\bibinfo {author} {\bibfnamefont {C.~T.}\ \bibnamefont {Lee}},\ }\href {\doibase 10.1103/PhysRevA.44.R2775} {\bibfield  {journal} {\bibinfo  {journal} {Phys. Rev. A}\ }\textbf {\bibinfo {volume} {44}},\ \bibinfo {pages} {R2775} (\bibinfo {year} {1991})}\BibitemShut {NoStop}%
\bibitem [{\citenamefont {Lee}(1992)}]{Lee1992}%
  \BibitemOpen
  \bibfield  {author} {\bibinfo {author} {\bibfnamefont {C.~T.}\ \bibnamefont {Lee}},\ }\href {\doibase 10.1103/PhysRevA.45.6586} {\bibfield  {journal} {\bibinfo  {journal} {Phys. Rev. A}\ }\textbf {\bibinfo {volume} {45}},\ \bibinfo {pages} {6586} (\bibinfo {year} {1992})}\BibitemShut {NoStop}%
\bibitem [{\citenamefont {Kiesel}\ and\ \citenamefont {Vogel}(2010)}]{KieselVogel2010}%
  \BibitemOpen
  \bibfield  {author} {\bibinfo {author} {\bibfnamefont {T.}~\bibnamefont {Kiesel}}\ and\ \bibinfo {author} {\bibfnamefont {W.}~\bibnamefont {Vogel}},\ }\href {\doibase 10.1103/PhysRevA.82.032107} {\bibfield  {journal} {\bibinfo  {journal} {Phys. Rev. A}\ }\textbf {\bibinfo {volume} {82}},\ \bibinfo {pages} {032107} (\bibinfo {year} {2010})}\BibitemShut {NoStop}%
\bibitem [{\citenamefont {Agudelo}\ \emph {et~al.}(2013)\citenamefont {Agudelo}, \citenamefont {Sperling},\ and\ \citenamefont {Vogel}}]{Agudeloetal2013}%
  \BibitemOpen
  \bibfield  {author} {\bibinfo {author} {\bibfnamefont {E.}~\bibnamefont {Agudelo}}, \bibinfo {author} {\bibfnamefont {J.}~\bibnamefont {Sperling}}, \ and\ \bibinfo {author} {\bibfnamefont {W.}~\bibnamefont {Vogel}},\ }\href {\doibase 10.1103/PhysRevA.87.033811} {\bibfield  {journal} {\bibinfo  {journal} {Phys. Rev. A}\ }\textbf {\bibinfo {volume} {87}},\ \bibinfo {pages} {033811} (\bibinfo {year} {2013})}\BibitemShut {NoStop}%
\bibitem [{\citenamefont {Sperling}(2016)}]{Sperling2016}%
  \BibitemOpen
  \bibfield  {author} {\bibinfo {author} {\bibfnamefont {J.}~\bibnamefont {Sperling}},\ }\href {\doibase 10.1103/PhysRevA.94.013814} {\bibfield  {journal} {\bibinfo  {journal} {Physical Review A}\ }\textbf {\bibinfo {volume} {94}},\ \bibinfo {pages} {013814} (\bibinfo {year} {2016})}\BibitemShut {NoStop}%
\bibitem [{\citenamefont {Lemos}\ \emph {et~al.}(2018)\citenamefont {Lemos}, \citenamefont {Almeida}, \citenamefont {Amaral},\ and\ \citenamefont {Oliveira}}]{Lemosetal2018}%
  \BibitemOpen
  \bibfield  {author} {\bibinfo {author} {\bibfnamefont {H.~C.}\ \bibnamefont {Lemos}}, \bibinfo {author} {\bibfnamefont {A.~C.}\ \bibnamefont {Almeida}}, \bibinfo {author} {\bibfnamefont {B.}~\bibnamefont {Amaral}}, \ and\ \bibinfo {author} {\bibfnamefont {A.~C.}\ \bibnamefont {Oliveira}},\ }\href {\doibase https://doi.org/10.1016/j.physleta.2018.01.023} {\bibfield  {journal} {\bibinfo  {journal} {Physics Letters A}\ }\textbf {\bibinfo {volume} {382}},\ \bibinfo {pages} {823} (\bibinfo {year} {2018})}\BibitemShut {NoStop}%
\bibitem [{\citenamefont {Linowski}\ and\ \citenamefont {Rudnicki}(2024)}]{linowski2023relating}%
  \BibitemOpen
  \bibfield  {author} {\bibinfo {author} {\bibfnamefont {T.}~\bibnamefont {Linowski}}\ and\ \bibinfo {author} {\bibfnamefont {L.}~\bibnamefont {Rudnicki}},\ }\href {\doibase 10.1103/PhysRevA.109.023715} {\bibfield  {journal} {\bibinfo  {journal} {Phys. Rev. A}\ }\textbf {\bibinfo {volume} {109}},\ \bibinfo {pages} {023715} (\bibinfo {year} {2024})}\BibitemShut {NoStop}%
\bibitem [{\citenamefont {Bernstein}(1929)}]{Bernstein}%
  \BibitemOpen
  \bibfield  {author} {\bibinfo {author} {\bibfnamefont {S.}~\bibnamefont {Bernstein}},\ }\href {\doibase 10.1007/BF02592679} {\bibfield  {journal} {\bibinfo  {journal} {Acta Mathematica}\ }\textbf {\bibinfo {volume} {52}},\ \bibinfo {pages} {1 } (\bibinfo {year} {1929})}\BibitemShut {NoStop}%
\bibitem [{\citenamefont {Mosonyi}\ and\ \citenamefont {Hiai}(2011)}]{MosonyiHiai2011}%
  \BibitemOpen
  \bibfield  {author} {\bibinfo {author} {\bibfnamefont {M.}~\bibnamefont {Mosonyi}}\ and\ \bibinfo {author} {\bibfnamefont {F.}~\bibnamefont {Hiai}},\ }\href {\doibase 10.1109/TIT.2011.2110050} {\bibfield  {journal} {\bibinfo  {journal} {IEEE Transactions on Information Theory}\ }\textbf {\bibinfo {volume} {57}},\ \bibinfo {pages} {2474} (\bibinfo {year} {2011})}\BibitemShut {NoStop}%
\bibitem [{\citenamefont {Gavrea}\ and\ \citenamefont {Ivan}(2014)}]{GavreaIvan2014}%
  \BibitemOpen
  \bibfield  {author} {\bibinfo {author} {\bibfnamefont {I.}~\bibnamefont {Gavrea}}\ and\ \bibinfo {author} {\bibfnamefont {M.}~\bibnamefont {Ivan}},\ }\href {\doibase https://doi.org/10.1016/j.amc.2014.05.003} {\bibfield  {journal} {\bibinfo  {journal} {Applied Mathematics and Computation}\ }\textbf {\bibinfo {volume} {241}},\ \bibinfo {pages} {70} (\bibinfo {year} {2014})}\BibitemShut {NoStop}%
\bibitem [{\citenamefont {Nikolov}(2014)}]{Nikolov2014}%
  \BibitemOpen
  \bibfield  {author} {\bibinfo {author} {\bibfnamefont {G.}~\bibnamefont {Nikolov}},\ }\href {\doibase https://doi.org/10.1016/j.jmaa.2014.04.022} {\bibfield  {journal} {\bibinfo  {journal} {Journal of Mathematical Analysis and Applications}\ }\textbf {\bibinfo {volume} {418}},\ \bibinfo {pages} {852} (\bibinfo {year} {2014})}\BibitemShut {NoStop}%
\bibitem [{\citenamefont {Raşa}(2018)}]{Rasa2018}%
  \BibitemOpen
  \bibfield  {author} {\bibinfo {author} {\bibfnamefont {I.}~\bibnamefont {Raşa}},\ }\href {\doibase 10.1007/s00025-018-0868-8} {\bibfield  {journal} {\bibinfo  {journal} {Results in Mathematics}\ }\textbf {\bibinfo {volume} {73}},\ \bibinfo {pages} {105} (\bibinfo {year} {2018})}\BibitemShut {NoStop}%
\bibitem [{\citenamefont {Raşa}(2019)}]{Rasa2019}%
  \BibitemOpen
  \bibfield  {author} {\bibinfo {author} {\bibfnamefont {I.}~\bibnamefont {Raşa}},\ }\href {\doibase 10.1007/s00025-019-1081-0} {\bibfield  {journal} {\bibinfo  {journal} {Results in Mathematics}\ }\textbf {\bibinfo {volume} {74}},\ \bibinfo {pages} {154} (\bibinfo {year} {2019})}\BibitemShut {NoStop}%
\bibitem [{\citenamefont {Alzer}(2020)}]{Alzer2020}%
  \BibitemOpen
  \bibfield  {author} {\bibinfo {author} {\bibfnamefont {H.}~\bibnamefont {Alzer}},\ }\href {\doibase 10.1007/s00025-020-1156-y} {\bibfield  {journal} {\bibinfo  {journal} {Results in Mathematics}\ }\textbf {\bibinfo {volume} {75}},\ \bibinfo {pages} {29} (\bibinfo {year} {2020})}\BibitemShut {NoStop}%
\bibitem [{\citenamefont {Glauber}(1963)}]{Glauber1963quantumtheorycoherence}%
  \BibitemOpen
  \bibfield  {author} {\bibinfo {author} {\bibfnamefont {R.~J.}\ \bibnamefont {Glauber}},\ }\href {\doibase 10.1103/PhysRev.130.2529} {\bibfield  {journal} {\bibinfo  {journal} {Physical Review}\ }\textbf {\bibinfo {volume} {130}},\ \bibinfo {pages} {2529} (\bibinfo {year} {1963})}\BibitemShut {NoStop}%
\end{thebibliography}
\end{document}